\documentclass{article}

\usepackage[english]{babel}
\usepackage[utf8]{inputenc}
\usepackage[T1]{fontenc}
\usepackage{amsmath}
\usepackage{amssymb}
\usepackage{mathtools}
\usepackage{amsthm}
\usepackage{framed}
\usepackage{algorithm}
\usepackage{algpseudocode}
\usepackage{enumerate}
\usepackage{bbm}
\usepackage[authordate-trad,
    backend=biber,noibid]{biblatex-chicago}
  \usepackage{diagbox}
  \usepackage{tikz}
  \usepackage{listings}
\usepackage{xcolor}
\usepackage[title]{appendix}

 \usepackage{csquotes}
\usepackage{graphicx}
\newcommand{\N}{\mathbb{N}}

\newcommand{\R}{\mathbb{R}}

\newcommand{\VaR}{\operatorname{VaR}}

\newcommand{\ES}{\operatorname{ES}}
\newcommand{\Dis}{\operatorname{Dis}}
\newcommand{\CTM}{\operatorname{CTM}}
\DeclareMathOperator*{\argmax}{arg\,max}
\DeclareMathOperator*{\argmin}{arg\,min}
\newtheorem{theorem}{Theorem}[section]
\newtheorem{lemma}[theorem]{Lemma}
\newtheorem{corollary}[theorem]{Corollary}
\theoremstyle{definition}
\newtheorem{definition}[theorem]{Definition}
\newtheorem{model}[theorem]{Model}
\newtheorem{konvention}[theorem]{Convention}
\newtheorem{beispiel}[theorem]{Example}

\newtheorem{bemerkung}[theorem]{Remark}
\usepackage{hyperref}

\providecommand{\keywords}[1]
{
  \small	
  \textbf{\textit{Keywords---}} #1
}

\title{On E-Backtesting: Generalizations and Sample Size Determination}
\author{
Dennis Oestmann$^{1}$\thanks{Email: \texttt{deoe@uni-bremen.de}} , 
Thorsten Dickhaus$^{2}$\thanks{Email: \texttt{dickhaus@uni-bremen.de}} \\
\small $^{1,2}$ Institute for Statistics, University of Bremen, 28359, Bremen, Germany
}
\date{\today}

\begin{document}

\maketitle 

\begin{abstract}
We present an approach for determining sample sizes required to detect underestimations of the expected shortfall with a prescribed power when applying the recently proposed e-backtesting procedure. We consider scenarios in which the value-at-risk at level $p$ is always estimated correctly, while the difference between the true expected shortfall and the value-at-risk is underestimated by a given factor $r$. We show that exploiting the structure of the backtest e-statistic 
proposed for backtesting the expected shortfall at level $p$ enables the derivation of approximate lower bounds for the required sample sizes by considering a sequence of independent and identically distributed Bernoulli-distributed random variables.
We also discuss potential limitations of this approximation and compare the resulting sample size requirements with those obtained in practical applications using Monte Carlo simulations. Furthermore, we present generalizations of the e-backtesting procedure, in particular to risk measures which constitute Bayes pairs.   
\end{abstract}


\section{Introduction}
In many real-world scenarios, one encounters improbable, but severely bad
scenarios. With the term risk, we refer to a quantification of the magnitude of these possible scenarios. One important example of this is given by financial risk: In this context, one is faced with possible large financial losses that might endanger a financial institution or even the entire financial system. It is thus of high interest for both society and financial institutions to assess and control these risks. This field is being referred to as (financial) risk management (see, for instance, \cite{mcneil2015quantitative}).  

The main regulatory authority for financial institutions is the Basel Committee on Banking Supervision (BCBS) which publishes guidelines, commonly referred to as Basel Accords, for proper risk management to be followed by the financial institutions. One pillar of the Basel Accords is the minimum capital requirement, which involves (among other risks) the assessment of the market risk, which describes the risk of large losses in a given portfolio due to changes in the stock market. The assessment of the market risk should be based on specific risk measures, in previous BCBS guidelines on the value-at-risk, in newer guidelines on the expected shortfall (see \cite{BCBS2019d457}). The guidelines further require the financial institutions to backtest their models used to analyze these risks with respect to the danger of underestimating the true risk. This means applying the model to historical data and comparing the projected risks based on the historical time points with the actual observed losses. While there are standard methods to backtest the value-at-risk (see \cite{BCBS2013}), suggested methods to backtest the expected shortfall often have noticeable downsides: For example, they often rely on restrictive model assumptions, only work for a fixed sample size, or only attain their target level asymptotically. Recently, \textcite{Hauptquelle} introduced a backtesting method based on the (relatively novel) concept of sequential e-values. Their ``e-backtesting'' method does not have the downsides stated before. In the present work, we first describe and analyze the e-backtesting method, and then we present possible generalizations. We also investigate issues regarding the applicability of the method in practice. In particular, we explore the question of appropriate sample sizes in order to detect underestimations with a prescribed power. 

The remainder of this work is structured as follows.
Section \ref{sec2} introduces basic notions and definitions.
In Section \ref{sec3}, we derive the framework of the e-backtesting method, which is based on backtest e-statistics. We also explain a  way to construct such backtest e-statistics for so-called Bayes pairs that fulfill a certain boundedness condition. 
We will see that the e-processes used for our method largely depend on so-called betting processes. Therefore, we will present different ways to construct these betting processes in Section \ref{sec4}, and we will analyze the resulting betting processes theoretically and in simulations. 
Section \ref{sec5} concerns the question of an appropriate threshold for the described e-backtesting procedure to be used for statistical testing. 
In Section \ref{sec6}, we discuss how sample sizes should be chosen when applying the described backtesting procedure in practice. For this, we investigate common types and sources of underestimations of risk and run numerical experiments to examine the rejection probabilities of the resulting sequential tests for different thresholds and sample sizes. 
We also detail a more theoretical argument to determine the sample size, which is based on a simple application of Jensen's inequality to the e-power of the described e-process.
We conclude with a discussion in Section \ref{sec8}.
For supplementary reading, Appendix \ref{sec7} provides general results on the (non-)existence of backtest e-statistics.

\section{Notation and preliminaries}\label{sec2}

\begin{konvention}
Throughout the remainder, we use the following notations:
\begin{itemize}
  \item 
  We denote $\N:=\{1,2,\ldots\}$ and $\N_0:=\{0,1,\ldots\}$.
  \item 
  For $q\in \N$, $\mathcal{M}_q$ denotes the set of distribution functions with existing $q$th moment. With $\mathcal{M}_0$ we denote the set of all distribution functions.
  \item 
  For two distribution functions $F,G$, we say $F\leq G$ if and only if $F(x)\geq G(x)$ for all $x\in \R$.
  \item 
  With $\Phi$ and $\phi$, we denote the cumulative distribution function (cdf) and the probability density function (pdf), respectively, of a standard normal random variable $Z\sim\mathcal{N}(0,1)$.
  \item 
  With $L_0$ we denote the space of all random variables defined on a probability space $(\Omega, \mathcal{A},\mathbb{P})$ and with $L_0^+$ the space of all non-negative random variables.
\end{itemize}
\end{konvention}

\subsection{E-variables, e-values, and e-processes}

This section introduces the concepts of e-variables, e-values, and e-processes. In this, we mostly follows the presentation in \textcite{RamdasBuch}.
\begin{definition}\label{E-Variable}
  Let $(\mathcal{X},\mathcal{A},\mathcal{P})$ be a statistical model and $\mathcal{P}_0\subset \mathcal{P}$. Further, let $E:\mathcal{X}\to [0,\infty]$ be a random variable. 
  \begin{enumerate}[(i)]
    \item 
    If $\mathbb{E}_\mathbb{P}[E]\leq 1$ for all $\mathbb{P}\in \mathcal{P}_0$, we call $E$ an e-variable for $\mathcal{P}_0$, and for any observed $x\in \mathcal{X}$ we call $E(x)$ an e-value.
      \item 
      If $\mathbb{E}_\mathbb{P}[E]= 1$ for all $\mathbb{P}\in \mathcal{P}_0$ we call $E$ an exact e-variable for $\mathcal{P}_0$.
  \end{enumerate} 
\end{definition}

\begin{definition}\label{e-power}
  Let $(\mathcal{X},\mathcal{A},\mathcal{P})$ be a statistical model, $\mathcal{P}_0\subset \mathcal{P}$ and $E:\mathcal{X}\to [0,\infty]$ be an e-variable for $\mathcal{P}_0$. For any $Q\in \mathcal{P}\backslash \mathcal{P}_0$, we call $\mathbb{E}^Q[\log E]$ the e-power of $E$ under $Q$.
\end{definition}

In our analysis, we will mostly look at sequences of e-values that constitute e-processes in discrete time. Letting the set of considered time indices be denoted by $T$, we call $T$ discrete if 
$T=\N_0$,  $T=\N$, or $T=\{0,\ldots,n\}$ for some $n\in \N$.

\begin{definition}\label{E-process}
Let $T$ be discrete and $\mathcal{P}$ be a set of probability measures such that $(\Omega,\mathcal{F},(\mathcal{F}_n)_{n\in T},\mathbb{P})$ is a filtered probability space for each $\mathbb{P}\in \mathcal{P}$. Further, let $\mathcal{P}_0\subset \mathcal{P}$. A sequence $(E_n)_{n\in T}$ of e-variables that is adapted to $(\mathcal{F}_n)_{n\in T}$ and fulfills $\mathbb{E}^\mathbb{P}[E_\tau]\leq 1$ for all $\mathbb{P}\in \mathcal{P}_0$ and stopping times $\tau$ is called an e-process for $\mathcal{P}_0$.
\end{definition} 
One of the most important results with regard to e-processes is Ville's inequality; see \textcite{Ville1939}. This result is, in principle, a generalization of Markov's inequality 
and allows for an easy construction of sequential hypothesis tests from e-variables for individual observations. To this end, the concept of test supermartingales is required.

\begin{definition}
 Let $T$ be discrete and $\mathcal{P}$ be a set of probability measures such that $(\Omega,\mathcal{F},(\mathcal{F}_n)_{n\in T},\mathbb{P})$ is a filtered probability space for each $\mathbb{P}\in \mathcal{P}$. Further, let $\mathcal{P}_0\subset \mathcal{P}$. A stochastic process $X=(X_n)_{n\in T}$ adapted to $(\mathcal{F}_n)_{n\in T}$ is called a test supermartingale for $\mathcal{P}_0$ if, for all $\mathbb{P}\in \mathcal{P}_0$, we have $X_n\geq 0$ $\mathbb{P}$-almost surely for all $n\in T$, $X$ is a supermartingale under $\mathbb{P}$ and $\mathbb{E}^\mathbb{P}[X_0]\leq 1$. 
\end{definition}
\begin{konvention}
In the following, we will always assume the existence of a discrete $T$, a set $\mathcal{P}$ of probability measures such that $(\Omega,\mathcal{F},(\mathcal{F}_n)_{n\in T},\mathbb{P})$ is a filtered probability space for each $\mathbb{P}\in \mathcal{P}$ as well as a subset $\mathcal{P}_0\subset \mathcal{P}$, without always restating these conditions.
\end{konvention}
\begin{lemma}[see Section 7.3 in \cite{RamdasBuch}]
Let $X=(X_n)_{n\in T}$ be a test supermartingale for $\mathcal{P}_0$. Then $X$ is an e-process for $\mathcal{P}_0$. 
\end{lemma}

\begin{bemerkung}
  One can also show that, for every e-process $E$, there exists a family of stochastic processes $(M^\mathbb{P})_{\mathbb{P}\in \mathcal{P}_0}$ such that for all $\mathbb{P}\in \mathcal{P}_0$ the process  $M^\mathbb{P}$ is a test supermartingale under $\mathbb{P}$ and $E\leq M^\mathbb{P}$ $\mathbb{P}$-almost surely. 
\end{bemerkung}
\begin{theorem}[Ville's inequality]\label{VilleUngleichung}
  Let $(X_n)_{n\in T}$ be a non-negative  supermartingale. Then, it holds for all $\alpha>0$
  \begin{align*}
    \mathbb{P}\left[\sup_{n\in T}X_n\geq \frac{1}{\alpha}\right]\leq \alpha \mathbb{E}[X_0] \, .
  \end{align*}
\end{theorem}

Ville's Inequality has the following direct  remarkable consequence.
\begin{corollary}\label{SupermartingaletoSequential}
  Let $(X_n)_{n\in T}$ be a test supermartingale for $\mathcal{P}_0$ and $\alpha \in (0,1)$. Then the binary process $(\phi_n)_{n\in T}$ defined by $\phi_n=\mathbbm{1}_{X_t\geq 1/\alpha}$ has the property
  \begin{align*}
    \sup_{\mathbb{P}\in \mathcal{P}_0}\mathbb{P}(\exists n\in \N: \phi_n=1)\leq \alpha \, .
  \end{align*}
  \end{corollary}
  Thus, if a test supermartingale for $\mathcal{P}_0\subset \mathcal{P}$ is available, one can derive a level $\alpha$ test for the null hypothesis $\mathcal{P}_0\subset \mathcal{P}$ by observing the test supermartingale and rejecting the null hypothesis as soon as one observes an e-value larger than $1 / \alpha$. The resulting test then belongs to the class of sequential level-$\alpha$ tests, which are defined as follows.
  \begin{definition}\label{sequentialTests}
  A stochastic process $(\phi_n)_{n\in \N}\in \{0,1\}^\N$ adapted to $(\mathcal{F}_n)_{n\in T}$ such that
    \begin{align*}
   \sup_{\mathbb{P}\in \mathcal{P}_0}\mathbb{P}(\exists n\in \N: \phi_n=1)\leq \alpha \, ,   
    \end{align*}
    is called a \text{level-$\alpha$ sequential test} for $\mathcal{P}_0$.
  \end{definition}
Corollary \ref{SupermartingaletoSequential} describes how one can construct sequential tests from test supermartingales. However, we can also derive e-processes from sequential tests such that the sequential test is equivalent to thresholding the corresponding e-process:
\begin{lemma}
 Let $(\phi_n)_{n\in T}\in \{0,1\}^T$ be a sequential level-$\alpha$ test for $\mathcal{P}_0$ with $\alpha\in (0,1)$. Then, the process $(X_n)_{n\in T}$ defined by $X_n=\phi_n/\alpha$ is an e-process for $\mathcal{P}_0$.
\end{lemma}
\begin{proof}
  Obviously $X_n\geq 0$ for all $n\in T$. Further, we have for any stopping time $\tau$ and $\mathbb{P}\in\mathcal{P}_0$ by definition of $\phi$:
  \begin{align*}
    \mathbb{E}^\mathbb{P}[X_\tau]=\frac{\mathbb{E}^\mathbb{P}[\phi_\tau]}{\alpha}=\frac{\mathbb{P}(\phi_\tau=1)}{\alpha}\leq \frac{\mathbb{P}(\exists n \in\N: \phi_n=1)}{\alpha}\leq 1
  \end{align*}
  Since this also holds for all constant stopping times $\tau:=n\in T$, all $X_n$ are e-variables and $(X_n)_{n\in \N}$ is an e-process.
\end{proof}

\subsection{Risk Measures}\label{subsec-risk-measures} 
In this section, we examine different measures of risk. Those are, in general, functionals mapping from a predefined set of random variables (in this context, also referred to as \textbf{loss variables}) to the real line, thereby representing the (financial) risk associated with this variable. If the loss variable in question represents losses of a financial asset, this risk should determine the amount of capital a controller needs to reserve in order to cover for large financial losses. 

The two most commonly used risk measures are the value-at-risk and the expected shortfall,
which are defined as follows.
\begin{definition}\label{VarundES}
  Let $L$ be a real-valued random variable with cdf $F$. 
  \begin{enumerate}[(i)]
    \item 
    For $\alpha\in (0,1)$, the \text{value-at-risk (VaR)} of $L$ at level $\alpha$ is defined as
    \begin{align*}
      \VaR_\alpha(L):=F^{-1}(\alpha)=\inf\{x\in \R| F(x)\geq \alpha\} \, .
    \end{align*} 
    \item 
    If $F\in \mathcal{M}_1$, we define, for $\alpha \in (0,1)$, the expected shortfall (ES) of $L$ at level $\alpha$ as
    \begin{align*}
      \ES_\alpha(L):= \frac{1}{1-\alpha}\int_\alpha^1 \VaR_{\zeta}(L)d\zeta \, .
    \end{align*}
  \end{enumerate}
  \end{definition}
Sometimes, instead of the expected shortfall, one also considers the conditional value-at-risk (see, for example, \cite{mcneil2015quantitative}):
\begin{definition}\label{CVar}
  Let $L$ be a real-valued random variable with cdf $F\in \mathcal{M}_1$ and $\alpha\in (0,1)$. Then, the conditional value-at-risk at level $\alpha$ is defined as
 \begin{align*}
      \operatorname{CVaR}_\alpha(L):=\mathbb{E}[L|L> \VaR_{\alpha}(L)]\, .
    \end{align*} 
  \end{definition}
  \begin{bemerkung}
  For a real-valued and continuous random variable, the expected shortfall at level $\alpha$ is equal to the conditional value-at-risk.
  \end{bemerkung}
  
In practice, coherent risk measures are of special importance. 

\begin{definition}\label{CoherentRiskMeasure}
Let $M$ be a convex cone of random variables such that it contains all constant random variables. Further, let $\rho:M\to \R$ be a risk measure.
\begin{enumerate}[(i)]
  \item 
  If $\rho(L+c)=\rho(L)+c$ for all $L\in M$ and $c\in \R$, we call $\rho$ translationally invariant. 
  \item 
  If $\rho(L_1+L_2)\leq \rho(L_1)+\rho(L_2)$ for all $L_1,L_2\in M$, we call $\rho$ subadditive.
  \item 
  If $\rho(\lambda L)=\lambda \rho(L)$ for all $L\in M$ and $\lambda>0$, we call $\rho$ positively homogeneous. 
  \item 
  If $\rho(L_1)\leq \rho(L_2)$ for all $L_1,L_2\in M$ such that $L_1\leq L_2$ almost surely, we call $\rho$ monotone.
  \item 
  If $\rho$ is translationally invariant, subadditive, positively homogeneous and monotone, we call $\rho$ coherent.
\end{enumerate} 
\end{definition}

\begin{bemerkung} $ $
\begin{itemize}
    \item[(i)]
Sometimes, in place of subadditivity and positive homogeneity, the property of convexity is also considered (see \cite{follmer2010convex}). This means that 
\begin{align*}
  \rho(\lambda L_1+(1-\lambda)L_2)\leq \lambda \rho(L_1)+(1-\lambda)\rho(L_2) \qquad \forall L_1,L_2\in M, \lambda \in [0,1] \, .
\end{align*} 
While any subadditive and positively homogeneous $\rho$ is obviously also convex, the back direction is in general not true. An important example of this is given by the class of entropic risk measures. These are defined as $\rho_\theta:M\to \R, L \to \frac{1}{\theta}\log(\mathbb{E}[e^{\theta L}])$ for some $\theta>0$ and an appropriate convex cone $M$, e.g. $L^\infty$ (see also, for example, \cite{detlefsen2005conditional}). Monotonicity of these risk measures is obvious, and translation invariance follows from a simple calculation. They are also convex, which can be derived by utilizing Hölder's inequality for $\frac{1}{\lambda}$ and $\frac{1}{1-\lambda}$. However, they are not positively homogeneous and thus not coherent. This can be seen by considering a standard normally distributed random variable $Z$. For some $\theta,\lambda>0, \lambda \neq 1$ we have
\begin{align*}
  \rho_\theta(\lambda Z)=\frac{1}{\theta} \log(\mathbb{E}[e^{\theta\lambda Z}])=\frac{\theta\lambda^2 }{2}\neq \frac{\theta\lambda}{2}=\frac{\lambda}{\theta}\log(\mathbb{E}[e^{\theta Z}])=\lambda \rho_\theta(Z) \, .
\end{align*}

\item[(ii)]
The different properties, which a coherent risk measure needs to fulfill, have clear interpretations if the loss variables examined are associated with financial losses of given assets: Translation invariance ensures that reserving capital $c$ equal to the risk of a loss variable $L$ leads to the resulting loss variable $L-c$ having risk $0$. Subadditivity means that diversification of a portfolio never increases the risk according to $\rho$. Positive homogeneity implies an increase of one's share in a given asset by some factor leads to an increase of the risk by the same factor. Monotonicity ensures that assets associated with higher losses are assigned a higher risk.  
\end{itemize}
\end{bemerkung}
Relating to monotonicity, a first result on the previous definitions that is also helpful in practice is the following:
\begin{lemma}\label{MonotonicityLemma}
  Let $M$ be a convex cone of random variables such that it contains all constant random variables. Further, let $\rho:M\to \R$ be a subadditive and positively homogeneous risk measure. Then, $\rho$ is monotone if and only if $\rho(L)\leq 0$ for all $L\leq 0$.
\end{lemma}
\begin{proof}
  First, assume that $\rho$ is monotone and let $L\in M,L\leq 0$. Since $\rho(0)=\rho(\lambda 0)=\lambda \rho(0)$ for all $\lambda>0$ by positive homogeneity, we derive $\rho(0)=0$. By monotonicity, we then have $\rho(L)\leq \rho(0)=0$. \\ 

  Next, assume $\rho(L)\leq 0$ for all $L\leq 0$. Let $L_1,L_2\in M, L_1\leq L_2$. Then $\rho(L_1-L_2)\leq 0$ by assumption. Subadditivity then yields
  \begin{align*}
    \rho(L_1)=\rho(L_1-L_2+L_2)\leq \rho(L_1-L_2)+\rho(L_2)\leq \rho(L_2) \, , 
  \end{align*}
  which proves monotonicity.
\end{proof}

Next, we examine as an example if the risk measures introduced in the last subsection are coherent:
\begin{theorem} $ $
\begin{itemize}
    \item[(i)]  Let $\alpha \in (0,1)$ and $\rho:L_0\to \R, L\to \VaR_\alpha(L)$. Then, $\rho$ is, in general, \textbf{not} subadditive and thus not coherent.
\item[(ii)]
Let $\alpha \in (0,1)$ and $M:=\{L\in L_0:L\sim F\in \mathcal{M}_1\}$. Then, $\rho:M\to \R, L\to \ES_\alpha(L)$ is coherent.
\end{itemize}
\end{theorem}
\begin{proof}
Part (i) is shown in Section 6.1.2 of \textcite{mcneil2005quantitative}, and Part (ii) is shown in Proposition 6.9 of \textcite{mcneil2005quantitative}. 
\end{proof}

Coherent risk measures have a surprising connection to e-variables by means of the following lemma:
\begin{lemma}[see Corollary 16.4 in \cite{RamdasBuch}]\label{EVariableCoherent}
  Let $(\mathcal{X},\mathcal{A},\mathcal{P})$ be a statistical model and $\mathcal{P}_0\subset \mathcal{P}$. Define $M:=\{L\in L_0^+:\mathbb{E}_\mathbb{P}[L]<\infty \quad  \forall \mathbb{P}\in \mathcal{P}\} $. Then, there exists a coherent risk measure $\rho:M \to \R$ such that $\rho(L)\leq 1$ if and only if $L$ is an e-variable for $\mathcal{P}_0$.
\end{lemma}
\begin{proof}
  Let $\rho:M\to \R$ be defined as 
  \begin{align}
    \rho(L):=\sup_{\mathbb{P}\in \mathcal{P}_0}\mathbb{E}_\mathbb{P}[L] \qquad \forall L \in M\, . \label{repcoherent}
  \end{align}
  Obviously, $\rho(L)\leq 1$ precisely if $L$ is an e-variable. It remains to show that $\rho$ is coherent: Translation invariance and positive homogeneity follow directly from the linearity of the expectation. Monotonicity is implied by monotonicity of the expectation. For subadditivity, let $L_1,L_2\in M$. Then, we have
  \begin{align*}
    \rho(L_1+L_2)=\sup_{\mathbb{P}\in \mathcal{P_0}}(\mathbb{E}_\mathbb{P}[L_1+L_2])\leq \sup_{\mathbb{P}\in \mathcal{P_0}}\mathbb{E}_\mathbb{P}[L_1]+\sup_{\mathbb{P}\in \mathcal{P_0}}\mathbb{E}_\mathbb{P}[L_2]=\rho(L_1)+\rho(L_2) \, .
  \end{align*}
\end{proof}
\begin{bemerkung}
Related to the previous lemma, one can even show that, under some conditional constraints, all coherent risk measures have a representation like in Equation \eqref{repcoherent}. For more details, see Theorems 4.16 and 4.22 in \cite{FolStoc2025} and Theorem 16.3 in \cite{RamdasBuch}.  
\end{bemerkung}

\section{Backtest E-Statistics}\label{sec3}
This section is concerned with constructing e-variables to backtest risk models for a potential underestimation of risk. For this, we first define our general setup.
\begin{definition}\label{BacktestESTatistic}
  Let $\mathcal{M}$ be a set of distribution functions and $\mathcal{P}\subset \mathcal{M}$. Further, let $\psi:\mathcal{M}\to \R^d$.
  \begin{enumerate}[(i)]
    \item 
    A measurable function $e:\R\times \psi(\mathcal{M})\to [0,\infty]$ which fulfills
    \begin{align*}
      \int_\R e(x,\psi(F))dF(x)\leq 1, \qquad \forall F\in \mathcal{P} \, ,
    \end{align*}
    is called a $\mathcal{P}$-point e-statistic for $\psi$. 
    \item 
    Assume that $\psi=(\rho,\phi)$ with $\rho:\mathcal{M}\to \R$.  A measurable function $e:\R\times \psi(\mathcal{M})\to [0,\infty]$ which fulfills
    \begin{align*}
      \int_\R e(x,r,\phi(F))dF(x)\leq 1 \qquad \forall F\in \mathcal{P} \text{ and } r\geq \rho(F) \, ,
    \end{align*}
    is called a $\mathcal{P}$-one-sided e-statistic for $\psi=(\rho,\phi)$.
  \end{enumerate}
\end{definition}
Those e-statistics can be easily used to construct e-variables as in Definition \ref{E-Variable}.
\begin{corollary}\label{EstatEvalue}
Let $\mathcal{M}$ be a set of distribution functions and $\mathcal{P}\subset \mathcal{M}$. Further, let $X\sim F$ be a random variable for some (unknown) $F\in M$  and $\psi:\mathcal{M}\to \R^d$.
  \begin{enumerate}[(i)]
    \item 
    Let $e: \R\times \psi(\mathcal{M})\to [0,\infty]$ be a $\mathcal{P}$-point e-statistic for $\psi$ in the sense of Definition \ref{BacktestESTatistic}.(i). Then, for any $z\in \psi(\mathcal{M})$, $e(X, z)$ is an e-variable for the null hypothesis 
    \begin{align*}
      F\in \{G\in \mathcal{P}|\psi(G)=z\} \, .
    \end{align*}
    \item 
Let $e: \R\times \psi(\mathcal{M})\to [0,\infty]$ be a $\mathcal{P}$-one-sided e-statistic for $\psi=(\rho,\phi)$ in the sense of Definition \ref{BacktestESTatistic}.(ii). Then, for any $r\in \rho(\mathcal{M})$ and $z\in \phi(\mathcal{M})$, $e(X,r ,z)$ is an e-variable for the null hypothesis 
\begin{align*}
      F\in \{G\in \mathcal{P}|\rho(G)\leq r \textrm{ and }\phi(G)=z\} \, .
    \end{align*}
  \end{enumerate}
\end{corollary}
Next, we define two helpful properties of e-statistics.
\begin{definition}
 Let $\mathcal{M}$ be a set of distribution functions and $\mathcal{P}\subset \mathcal{M}$. Further, let $\psi=(\rho,\phi):\mathcal{M}\to \R\times\R^{d-1}$ and $e:\R\times \psi(\mathcal{M})\to \R$ be a $\mathcal{P}$-one-sided e-statistic for $\psi$.
 \begin{enumerate}[(i)]
  \item 
  If, for all $(r,z)\in \psi(\mathcal{P})$ and $F\in \mathcal{P}$ with $\rho(F)>r$, it holds
  \begin{align*}
    \int_\R e(x,r,z) dF(x) >1 \, ,
  \end{align*} 
  we call $e$ a ($\mathcal{P}$)-backtest e-statistic for $\psi$.
  \item  
  If $e$ is a backtest e-statistic for $\psi$ and $e$ is decreasing in $r$ for all $x\in \R$ and $z\in \phi(\mathcal{M})$, we call $e$ monotone.
 \end{enumerate}  
\end{definition} 
An e-variable derived from a backtest e-statistic for $\psi=(\rho,\phi)$ according to Corollary \ref{EstatEvalue}.(ii) has the important property that rejecting a null hypothesis where $\rho$ is underestimated is never less likely than rejecting a true null hypothesis, regardless of a possible misspecification of $\phi$. The property of monotonicity ensures that the power of the corresponding test is increasing in the underestimation of $r$. \\

A first, simple example of a backtest e-statistic is given by considering the case that $\rho$ equals the variance and $\phi$ equals the expected value.
\begin{beispiel}\label{ESVaRBeispiel}
Let  $\psi:=(\rho,\phi):\mathcal{M}_2\to \R^2$ with $(\rho(F),\phi(F)):=(\operatorname{Var}(X),\mathbb{E}[X])$ for $X\sim F$. Then, $e:\R\times \R_{\geq 0}\times\R\to \R, e(x,r,z)=\frac{(x-z)^2}{r}$ for $r>0$ and $e(x,0,z)=\infty \mathbbm{1}_{x\neq z}$ is a monotone $\mathcal{M}_2$-backtest e-statistic for $\psi=(\operatorname{Var},\mathbb{E})$.
\end{beispiel}
\begin{proof}
  Let $r>0, X\sim F\in \mathcal{M}_2$ with $\operatorname{Var}(X) \leq r$. Then, it holds
  \begin{align*}
    \int_\R \frac{(x-\mathbb{E}[X])^2}{r}dF(x)=\frac{\operatorname{Var}(X)}{r}\leq 1 \, . 
  \end{align*}
  Thus, $e$ is a $\mathcal{M}_2$-one-sided e-statistic for $\psi$. Now, let $z\in \R$ and $X\sim F\in \mathcal{M}_2$ with $\operatorname{Var}(X) > r$. Then, we have
  \begin{align*}
    \int_\R \frac{(x-z)^2}{r}dF(x)&=\int_\R \frac{(x-\mathbb{E}[X])^2+2(x-\mathbb{E}[X])(\mathbb{E}[X]-z)+(\mathbb{E}[X]-z)^2}{r}dF(x) \\
    &= \frac{\operatorname{Var}(X)}{r}+0+\frac{(E[X]-z)^2}{r}> 1 \, .
  \end{align*}
  Thus, $e$ is also a $\mathcal{M}_2$-backtest e-statistic. Monotonicity of $e$ is obvious.
\end{proof}

\subsection{Bayes pairs}
In this section, we will derive a general way to construct monotone backtest $e$-statistics. This involves the notions of Bayes pairs, which were first described by \textcite{https://doi.org/10.1111/mafi.12313} and are defined as follows.
\begin{definition}\label{BayesPairs}
  Let $\mathcal{M}$ be a set of distribution functions and $\psi=(\rho,\phi):\mathcal{M}\to \R\times \R^{d-1}$. If there exists a function $L:\R^{d}\to \R$ such that, for all $F\in \mathcal{M}$, $\min_{z\in \R^{d-1}} \int_\R L(z,x)dF(x)$ exists and it holds
  \begin{align*}
    \phi(F)\in \argmin_{z\in \R^{d-1} }\int_\R L(z,x)dF(x) \textrm{ and } \rho(F)=\min_{z\in \R^{d-1}} \int_\R L(z,x)dF(x) \, ,
  \end{align*}
we call $(\rho, \phi)$ a \text{Bayes pair} w.r.t. $\mathcal{M}$ and loss function $L$. 
\end{definition}

As a first example of a Bayes pair, we look at the variance and expectation of random variables with existing second moments:
\begin{lemma}\label{VarEBeispiel}
Let $X$ be a real-valued random variable with cdf $F\in \mathcal{M}_2$. Then, it holds
\begin{align*}
  \mathbb{E}[X]&=\argmin_{z\in\R} \mathbb{E}[(X-z)^2] \, , \\
\operatorname{Var}(X)&=\min_{z\in\R} \mathbb{E}[(X-z)^2] \, .
\end{align*}
Thus, $(\operatorname{Var},\mathbb{E})$ constitutes a Bayes pair w.r.t. $M_2$ and loss function $L(z,x)=(x-z)^2$. 
\end{lemma}
\begin{proof}
  Let $f:\R\to \R, z\to \mathbb{E}[(X-z)^2]=\mathbb{E}[X^2]-2\mathbb{E}[X]z+z^2$. Differentiating $f$ and setting the derivative to $0$ yields $f'(z)=0\iff z=\mathbb{E}[X]$. Since $f''(z)=2 >0$, $\mathbb{E}[X]$ is the unique minimizer of $f$. By definition, we then have $\min_{z\in\R} \mathbb{E}[(X-z)^2]=\mathbb{E}[(X-\mathbb{E}[X])^2]=\operatorname{Var}(X)$. 
\end{proof}

Under one additional prerequisite, we can always derive a monotone backtest e-statistic from a Bayes pair.
\begin{lemma}\label{BayesPairToEstatistic}
 Let $\mathcal{M}$ be a set of distribution functions and $\psi=(\rho, \phi)$ be a \text{Bayes pair} w.r.t. $\mathcal{M}$ and loss function $L$. If, for every $z\in \R^{d-1}$, $L(z,\cdot)$ is lower bounded by $c(z)\in \R$, the function $e:\R\times\psi(\mathcal{M})\to \R$ with
 \begin{align*}
  e(x,r,z)=(L(z,x)-c(z))/(r-c(z)) \, ,
 \end{align*} 
 is a monotone $\mathcal{M}$-backtest e-statistic for $(\rho,\phi)$.
\end{lemma}
\begin{proof}
  Let $(r,z)\in \psi(\mathcal{M})$. By definitions of $c(z)$ and $\rho$, we have $L(z,x),r\geq c(z)$ and thus $e(x,r,z)\geq 0$ for all $x\in \R$. Now, let $F\in \mathcal{M}$ with $\rho(F)\leq r$. Then, we have
  \begin{align*}
    \int_\R e(x,r,\phi(F))dF(x)&=\int_\R \frac{L(\phi(F),x)-c(\phi(F))}{r-c(\phi(F))}dF(x)\\
    &\leq \frac{\int_R L(\phi(F),x)dF(x)-c(\phi(F))}{\rho(F)-c(\phi(F))} \\
    &=\frac{\rho(F)-c(\phi(F))}{\rho(F)-c(\phi(F))}=1 \, ,
  \end{align*} 
  where the second-to-last equality is due to $(\rho,\phi)$ being a Bayes pair.

   Next, let $\rho(F)>r$. Then,
   \begin{align*}
    \int_\R e(x,r,z)dF(x)= \int_\R \frac{L(z,x)-c(z)}{r-c(z)}dF(x)\geq \frac{\rho(F)-c(z)}{r-c(z)}>1 \, ,
   \end{align*}
   where the second-to-last inequality is due to the definition of $\rho$.

   Finally, monotonicity of $e$ in $r$ is again obvious.
\end{proof}
Looking back at our example derived from Lemma \ref{VarEBeispiel}, we get that $e^{\operatorname{Var}}(x,r,z):=(x-z)^2 / r$ is a monotone backtest e-statistic for $(\operatorname{Var},\mathbb{E})$ since $(x-z)^2\geq 0$ always holds true, confirming our results from Example \ref{ESVaRBeispiel}. 

Our methodology for constructing monotone backtest e-statistics from Bayes pairs can also be easily applied to the value-at-risk and the expected shortfall by means of the following lemma which is due to \textcite{ROCKAFELLAR20021443}.
\begin{lemma}[Theorem 10 in \cite{ROCKAFELLAR20021443}]\label{BayesPairESVaR}
  Let $X$ be a real-valued random variable with cdf $F\in \mathcal{M}_1$. Then, it holds for all $\alpha\in (0,1)$
 \begin{align*}
  \VaR_\alpha(X)\in \argmin_{z\in \R} \left\{ z+\frac{1}{1-\alpha}\mathbb{E}[(X-z)_+]\right\} \, , \\ 
  \ES_\alpha(X)=\min_{z\in \R} \left\{ z+\frac{1}{1-\alpha}\mathbb{E}[(X-z)_+]\right\} \, .
 \end{align*}
\end{lemma}

Lemma \ref{BayesPairESVaR} yields that the value-at-risk and expected shortfall, both at level $\alpha$, constitute a Bayes pair w.r.t. $\mathcal{M}_1$ and loss function $L(z,x)=z+\frac{(x-z)_+}{1-\alpha}$. Since $L(z,x)\geq z$ for all $z,x \in \R$, Lemma \ref{BayesPairToEstatistic} provides us with a monotone backtest e-statistic for $(\ES_\alpha,\VaR_\alpha)$ given by
\begin{align*}
  e^{\ES}_\alpha(x,r,z):=\frac{(x-z)_+}{(1-\alpha)(r-z)} \, .
\end{align*}

\subsection{Minimization of \texorpdfstring{$L^p$}{Lp} loss}
Lemma \ref{BayesPairToEstatistic} can also be used for the construction of Bayes pairs for a given loss function satisfying the conditions of the lemma as well as deriving a monotone backtest e-statistic for the resulting risk measures. In this subsection, we study an example of this by looking at the $L^p$ loss.
\begin{definition}\label{LpLoss}
Let $p>0$ and $X\in L_p$ be a random variable with existing $p$th moment.
\begin{enumerate}[(i)]
  \item 
  We call $\mathbb{E}[|X|^p]$ the $L^p$ loss of $X$.
  \item 
  Let $z\in \R$. We call $\mathbb{E}[|X-z|^p]$ the generalized $L^p$ loss of $X$ w.r.t. $z$. 
  \item 
  We call $f:\R \to \R, z\to \mathbb{E}[|X-z|^p]$, the $L^p$ loss function of $X$. 
\end{enumerate}
\end{definition}
To apply our theory on Bayes pairs, we next consider minimizers of $L^p$ loss functions.
\begin{definition}\label{LpMinimizers}
  Let $p>0$ and $X\in L_p$ be a random variable with existing $p$th moment.
  \begin{enumerate}[(i)]
    \item 
  We call any $z_0\in \argmin_{z\in \R} \mathbb{E}[|X-z|^p]$ an $L^p$ minimizer of $X$ and denote the smallest such minimizer by $\operatorname{CTM}_p(X)$ (central tendency measure). 
  \item 
  We call $ \operatorname{Dis}_p(X):=\min_{z\in \R} \mathbb{E}[|X-z|^p]$ the $L^p$ dispersion of $X$. 
  \end{enumerate}
\end{definition}
Such $L^p$ minimizers and $L^p$ dispersions provide a generalization of expectations and variances to higher powers. Indeed, if $X\in L_2$, we have $\operatorname{CTM}_2(X)=\mathbb{E}[X]$ and $\operatorname{Dis}_2(X)=\operatorname{Var}(X)$. 
Obviously, $\operatorname{CTM}_p$ and $\operatorname{Dis}_p$ constitute a Bayes pair, which, since $|x-z|^p\geq 0$ always holds true, can be backtested by Lemma \ref{BayesPairToEstatistic} using the monotone backtest e-statistic
\begin{align*}
  e_{\operatorname{Dis}_p}(x,r,z):=\frac{|x-z|^p}{r} \, .
\end{align*} 
We are also interested in the $p$-th root of $\operatorname{Dis}_p$ for reasons explained in the following paragraphs. Since taking roots is a monotone transformation, backtesting $\operatorname{Dis}_p$ is essentially equivalent to backtesting $\operatorname{Dis}_p^{1/p}$. 
To this end, we examine whether $\operatorname{Dis}_p^{1/p}$ fulfills wishful properties of risk measures discussed in Subsection \ref{subsec-risk-measures}. 
\begin{lemma}\label{DispersionSubLinear}
Let $p>0$ and $\rho:L^p\to \R, L\to \operatorname{Dis}^{1/p}_p(L)$. Then, $\rho$ is positively homogeneous and subadditive if $p\geq 1$.
\end{lemma}
\begin{proof}
  For establishing positive homogeneity, let $L\in L^p$ and $\lambda>0$ and notice 
  \begin{align*}
  \min_{z\in \R} \mathbb{E}[|\lambda L-z|^p]=\min_{z\in \R} \mathbb{E}[|\lambda L-\lambda z|^p] = \lambda^p \min_{z\in \R} \mathbb{E}[|L- z|^p] \, .
  \end{align*}
  Thus, $\Dis^{1/p}_p(\lambda X)=\lambda \Dis^{1/p}_p(X)$.

  For establishing subadditivity, let $p\geq 1$ and $K,L \in L^p$. Using the definitions of $\Dis_p$ and $\CTM_p$, we derive
  \begin{align*}
  \Dis_p^{1/p}(K+L)=(\min_{z\in \R} \mathbb{E}[|K+L-z|^p])^{1/p} &\leq (\mathbb{E}[|K-\CTM(K)+L-\CTM(L)|^p])^{1/p} \\ 
  &\leq \Dis_p^{1/p}(K)+\Dis_p^{1/p}(L) \, ,
  \end{align*}
  where the last inequality is due to the Minkowski inequality.
\end{proof}
Considering translation invariance and monotonicity, $\operatorname{Dis}^{1/p}_p(L)$ is obviously not a translation invariant risk measure as $\operatorname{Dis}^{1/p}_p(L+c)=\operatorname{Dis}^{1/p}_p(L)$ for all $L\in L^p$ and $c\in \R$. It is also not monotone by Lemma \ref{MonotonicityLemma},  since $\operatorname{Dis}_p^{1/p}(L)>0$ for every non-constant, possible negative $L\in L^p$. 
One possible remedy to at least solve the problem of the risk measure not being translation invariant lies in adding a sublinear functional, e.\ g., the expectation, to $\operatorname{Dis}^{1/p}$, i.\ e., to consider $\mathbb{E}[L]+\operatorname{Dis}_p^{1/p}(L)$. However, this solution does only solve the problem of monotonicity for $p\leq 1$.
\begin{lemma}\label{DisperionMonotone}
Let $p>0$ and $\rho:L^{\max\{1,p\}}\to \R, L\to \mathbb{E}[L]+\operatorname{Dis}^{1/p}_p(L)$. Then, $\rho$ is monotone if and only if $p\leq 1$.
\end{lemma}
\begin{proof}
  First, assume $p\leq 1$. Let $L\leq 0$. By the definition of $\Dis_p$ and Jensen's inequality applied to concave functions, we derive
  \begin{align*}
  \operatorname{Dis}^{1/p}_p(L)= (\min_{z\in \R} \mathbb{E}[|L-z|^p])^{1/p} \leq \mathbb{E}[(-L)^p]^{1/p}\leq -\mathbb{E}[L]\, .
  \end{align*}
  Thus, $\mathbb{E}[L]+\operatorname{Dis}^{1/p}_p(L)\leq 0$. Monotonicity thus follows from Lemma \ref{MonotonicityLemma}. \\ 

  Next, assume $p>1$. Consider for $\epsilon\in (0,1)$ the random variable $L_\epsilon$ with $\mathbb{P}(L_\epsilon=-1/\epsilon)=\epsilon$ and $\mathbb{P}(L_\epsilon=0)=1-\epsilon$. Obviously, $\mathbb{E}[L]=-1$ and for $-z_\epsilon\in (-1/\epsilon,0)$, we have
  \begin{align*}
    \mathbb{E}[|L_\epsilon+z_\epsilon|^p]^{1/p} =\left((1-\epsilon)z_\epsilon^p+\epsilon \left(\frac{1}{\epsilon}-z_\epsilon\right)^p\right)^{1/p} \, .
  \end{align*}
  In order for this term to be smaller than $-\mathbb{E}[L]$, both $(1-\epsilon)z_\epsilon^p$ and $\epsilon \left(\frac{1}{\epsilon}-z_\epsilon\right)^p$ have to be smaller than $1$. The former condition leads to $z_\epsilon\leq ((1/(1-\epsilon))^{1/p}\xrightarrow{\epsilon \to 0} 1$ and the latter condition leads to $z_\epsilon \geq \frac{1}{\epsilon}- \left(\frac{1}{\epsilon}\right)^{1/p}\xrightarrow{\epsilon\to 0 } \infty$ since $p>1$. Thus, both conditions cannot be fulfilled for a sufficiently small $\epsilon$ where we then have $\mathbb{E}[L_\epsilon]+\Dis_p^{1/p}(L_\epsilon)\geq 0$. The assertion then follows again from Lemma \ref{MonotonicityLemma}.
\end{proof}
From Lemmas \ref{DispersionSubLinear} and \ref{DisperionMonotone}, we conclude that the risk measure $\mathbb{E}[\cdot]+ \Dis_p^{1/p}(\cdot)$ is only coherent for $p=1$. Still, one might also consider this risk measure for $p>1$ for the reasons mentioned before. However, one needs to remember that this risk measure might assign high risks to loss variables with heavy left tails. 

\subsection{E-processes based on backtest e-statistics}\label{sec33}
Our next goal is to construct e-processes from backtest e-statistics. For this purpose, let $M$ be a set of distribution functions, $\mathcal{P}\subset M$, $\psi=(\rho,\phi):\mathcal{M}\to \R\times \R^{d-1}$, and let $e:\R\times \psi(\mathcal{M})\to [0,\infty]$ be a $\mathcal{P}$-one-sided e-statistic for $(\rho,\phi)$. Now, assume there exists a stochastic process $(X_n)_{n\in T}$ for a discrete $T$ adapted to a filtration $(\mathcal{F}_n)_{n\in T}$. We may think of $(X_n)_{n\in T}$ as a sequence of realized losses. Furthermore, assume that we want to test the null hypothesis:
\begin{align*}
  H_0: (X_n|\mathcal{F}_{n-1})\sim F_n\in \mathcal{P} \textrm{ and } \rho(F_n)\leq r_n \textrm{ and } \phi(F_n)= z_n \quad \forall n\in \N,
\end{align*}
where $r_n$ and $z_n$ are forecasts for $\rho(F_n)$ and $\phi(F_n)$ which are predictable, i.e. $\mathcal{F}_{n-1}$-measurable. If $\rho$ represents a risk measure, the null hypothesis, in essence, means that the risk forecast $r_n$ never underestimates the true risk and the additional forecast $z_n$ of the auxiliary information $\phi$ is always correct.

After stating our null hypothesis, we construct our e-process by means of the following lemma.
\begin{lemma}\label{ProcessisMartingale}
  Under the conditions described before, let $Y_n:=e(X_n,r_n,z_n)$ for $n\in T$ and let $\mathbf{\lambda}=(\lambda_n)_{n\in \N}$ be a predictable process. Then the process $(M_n)_{n\in T}$ defined by $M_0=1$ and
  \begin{align*}
    M_{n+1} \equiv M_{n+1}(\mathbf{\lambda})=(1-\lambda_{n+1}+\lambda_{n+1} Y_{n+1})M_n(\mathbf{\lambda}) \qquad \forall n \in T\backslash\{0\}\, ,
  \end{align*}
  is a test supermartingale for $H_0$ and thus an e-process.
\end{lemma}
\begin{proof}
Since $e$ is a $\mathcal{P}$-one-sided e-statistic for $(\rho,\phi)$, we have by Corollary \ref{EstatEvalue} that, for all $n\in T$, $Y_{n+1}|\mathcal{F}_n$ is an e-variable for the null hypothesis 
\begin{align*}
  Y_{n+1}|\mathcal{F}_{n}\sim F_{n+1} \in \mathcal{P}, \rho(F_{n+1})\leq r_{n+1} \textrm{ and } \phi(F_{n+1})=z_{n+1},
\end{align*}
which always holds under $H_0$. Thus, for all $n\in T$, we have $Y_n\geq 0$, which also implies $M_n\geq 0$ via induction. Further, notice that for all $\mathbb{P}\in H_0$ and $n\in T$ the relationship 
\begin{align*}
 \mathbb{E}^\mathbb{P}[M_{n+1}|\mathcal{F}_{n}]=(1-\lambda_{n+1}+\lambda_{n+1} \mathbb{E}^\mathbb{P}[Y_{n+1}|\mathcal{F}_{n}]) M_{n}\leq M_{n}
\end{align*}
holds true, since $Y_{n+1}|\mathcal{F}_{n}$ is an e-variable under $H_0$. Therefore, $(M_n)_n$ is a supermartingale under $\mathbb{P}$. Lastly, we have $M_0=1$ by definition.
\end{proof}
A possible backtesting procedure can then be described as follows: For a fixed level $\alpha$, we observe the test supermartingale $(M_n)_{n\in \N}$ and reject $H_0$ as soon as some $M_n$ exceeds the value $1 / \alpha$. By Ville's inequality (Theorem \ref{VilleUngleichung}), this test procedure always keeps the type-1 error probability below $\alpha$.
\begin{bemerkung}
  Since Ville's inequality holds for arbitrary stopping rules, one can also observe the process $(M_n)_{n\in \N}$ up to an arbitrary stopping time $\tau$ and report the ``observed significance level''
  \begin{align*}
    \alpha_{\text{obs}}:=\min\left\{1,\frac{1}{\max_{1\leq n\leq \tau}\{ M_n\}}\right\}.
  \end{align*}
  Rejecting $H_0$ then has a type-1 error probability of at most $\alpha_{\text{obs}}$. This remarkable property of $e$-processes is also being referred to as them being \text{anytime valid}.
\end{bemerkung}

In practice, applying a threshold of $1 / \alpha$ often leads to the resulting sequential level-$\alpha$ tests becoming very conservative. Therefore, other thresholds are also considered. This issue is discussed in more detail in Section \ref{sec6}.

\section{The betting process}\label{sec4}
In the previous section, we proposed the e-process $(M_n)_{n\in \N}$ for backtesting risk measures. This process depends on the betting process $(\lambda_n)_{n\in \N}$. As argued in the previous section, for all $\alpha\in(0,1)$, a sequential level-$\alpha$ test based on the e-process $(M_n)_{n\in \N}$ will always keep the type-1 error probability of the test below $\alpha$ if the significance threshold  $1 / \alpha$ is applied. Thus, we can choose the betting process in order to maximize the power of the resulting test. For this, we aim at optimizing the e-power of the e-variables $(1-\lambda_{n}+\lambda_{n} e(x_n,r_n,z_n))$ for all time steps $n\in\N$ under a probability distribution $Q$ belonging to the alternative. First, we note that for a backtest e-statistic $e$, if $X_n\sim F$ and $\rho(F)>r$, we can always find a  $\lambda\in [0,1]$ such that $(1-\lambda_{n}+\lambda_{n} e(x_n,r_n,z_n))$ has positive e-power by means of the following very general lemma.
\begin{lemma}\label{PosEPower}
  Let $E$ be any real-valued non-negative random variable. Then, it holds
  \begin{align*}
    \mathbb{E}[E]>1 \iff \exists \lambda\in (0,1]:\mathbb{E}[\log(1-\lambda+\lambda E)]>0\, .
  \end{align*}
\end{lemma}
\begin{proof}
Assume first that $0<\mathbb{E}[\log((1-\lambda)+\lambda E)]$. Since $\log$ is concave, Jensen's inequality yields that $0<\mathbb{E}[\log((1-\lambda)+\lambda E)] \leq\log(1-\lambda +\lambda \mathbb{E}[E])$, which implies $1 <\mathbb{E}[E]$.

Now, assume $\mathbb{E}[E]>1$. Let $Y_n:=\min\{E,n\}$ for $n\in \N$. Then $(Y_n)_n$ is a monotonically increasing sequence of random variables with existing first moments pointwise converging to $E$. By monotone convergence, this implies $\mathbb{E}[Y_n] \to \mathbb{E}[E]>1$. Thus, we find an $N\in \N$ such that $\mathbb{E}[Y_N]>1$. Next, note that
\begin{align*}
  0<\mathbb{E}[Y_N-1]=\mathbb{E}[(Y_N-1)_+]-\mathbb{E}[(Y_N-1)_-] \, .
\end{align*}
Since this inequality is strict, there exists an $\epsilon\in (0,1)$ such that
\begin{align}
  0<\frac{\mathbb{E}[(Y_N-1)_+]}{1+\epsilon}-\frac{\mathbb{E}[(Y_N-1)_-]}{1-\epsilon} \, .\label{GleichungLem1.1}
\end{align}
Since, for $x\in [0,\epsilon)$, we have
\[
\frac{d}{dx}\log(1+x)=\frac{1}{1+x}>\frac{1}{1+\epsilon}=\frac{d}{dx} \frac{x}{1+\epsilon}
\]
and $\log(1)=0= 0 / (1+\epsilon)$, it holds $\log(1+x)\geq x / (1+\epsilon)$ for $x\in [0,\epsilon)$ and similarly $\log(1+x)\geq x / (1-\epsilon)$ for $x\in (-\epsilon,0]$. This implies
\begin{align}
  \log(1+x)\geq \frac{x_+}{1+\epsilon}-\frac{x_-}{1-\epsilon} \qquad \forall x \in (-\epsilon,\epsilon) \, .\label{GleichungLem1.2}
\end{align}
Now, let $\lambda \in (0,\epsilon/N)$, which yields $\lambda(Y_N-1)\in (-\epsilon,\epsilon)$. Then, we have by inequalities \eqref{GleichungLem1.2} and \eqref{GleichungLem1.1}, that
\begin{align*}
  \mathbb{E}[\log(1-\lambda+\lambda E)]&\geq \mathbb{E}[\log(1-\lambda+\lambda Y_N)] \\ 
  &= \mathbb{E}[\log(1+\lambda (Y_N-1)] \\ 
  &\geq \frac{\mathbb{E}[(Y_N-1)_+]}{1+\epsilon}-\frac{\mathbb{E}[(Y_N-1)_-]}{1-\epsilon}>0,
\end{align*}
completing the proof.
\end{proof}
\subsection{Betting processes based on growth rates}
 The equivalence  established in Lemma \ref{PosEPower} directly leads to the growth-rate optimal (GRO) criterion proposed by  \textcite{7852039d69554002af6904518991286e}.
\begin{definition}\label{GRO}
  The growth-rate optimal (GRO) betting process $\lambda^{\text{GRO}}:=\left(\lambda_n^\text{GRO}\right)_{n\in \N}$ for $(M_n)_{n\in \N}$ and a sequence $(Q_n)_{n\in \N}$ of probability distributions for $(X_n|\mathcal{F}_{n-1})_{n\in\N}$ under the alternative is given by
  \begin{align*}
  \forall n\in \N: \lambda^{\text{GRO}}_n(r_n,z_n):= \argmax_{\lambda \in [0,\gamma]}\mathbb{E}^{Q_n}[\log(1-\lambda+\lambda e(X_n,r_n,z_n))], 
  \end{align*} 
  for fixed $\gamma\in (0,1)$.
\end{definition}

\begin{bemerkung} $ $
\begin{itemize}    
\item[(i)] Notice that the function $\mathbb{E}^{Q_n}[\log(1-\lambda+\lambda e(X_n,r_n,z_n))]$ is concave in $\lambda$. Hence, the solution to the optimization problem in Definition \ref{GRO} can be calculated by a convex program. 
\item[(ii)] The purpose of $\gamma$ in Definition \ref{GRO} is to avoid the e-process $(M_n)_{n\in \N}$ becoming (nearly)  $0$. This problem does not affect \text{GRO} as much in theory, as the next lemma will show. However, it does affect all approximations of it.
\end{itemize}
\end{bemerkung}

\begin{lemma}\label{GROnieNull}
Let $n\in \N$, $Q_n$ be a probability distribution for $X_n|\mathcal{F}_{n-1}$ and $(r_n,z_n)\in \psi(\mathcal{M})$ such that $Q_n(e(X_n,r_n,z_n)=0)\geq \epsilon \in (0,1)$. Then, we have $\lambda^{\text{GRO}}_n(r_n,z_n)\leq 1-\epsilon$.
\end{lemma} 
\begin{proof}
First, notice that $\lambda^{\text{GRO}}_n(r_n,z_n)$ is increasing w.r.t. the random variable $e(X_n,r_n,z_n)$. Now, suppose $e(X_n,r_n,z_n)$ is upper bounded by $K+1$ for some $K>\epsilon / (1-\epsilon)$. Then, we have
\begin{align*}
\lambda^{\text{GRO}}_n(r_n,z_n)\leq \argmax_{\lambda \in [0,\gamma]} \epsilon \log(1-\lambda)+(1-\epsilon)\log(1+\lambda K).
\end{align*}
Setting the derivative of the previous expression to $0$ thus yields
\begin{align*}
 \lambda^{\textbf{GRO}}_n(r_n,z_n)\leq \min\left\{\gamma, \frac{K(1-\epsilon)-\epsilon}{K}\right\} \leq \frac{K(1-\epsilon)-\epsilon}{K} \xrightarrow{K\to \infty}1-\epsilon.
\end{align*}
\end{proof}
Thus, when backtesting the expected shortfall, the e-process $(M_n)_{n\in \N}$ will almost surely not become zero when using the GRO betting process under an alternative, where the probability of observing a loss below the estimated value-at-risk is not $0$. We can even precisely specify the conditions under which the GRO betting process becomes $0$ or $1$ by means of the following lemma.
\begin{lemma}
Let $n\in \N$, $Q_n$ be a probability distribution for $X_n|\mathcal{F}_{n-1}$ and $(r_n,z_n)\in \psi(\mathcal{M})$. Then, we have
\begin{enumerate}[(i)]
  \item 
  $\lambda_n^\text{GRO}(r_n,z_n)=0\iff \mathbb{E}^{Q_n}[e(X_n,r_n,z_n)]\leq 1$. 
  \item  
  If $\gamma=1$ in Definition \ref{GRO} and $\mathbb{E}^{Q_n}[1/e(X_n,r_n,z_n)]$ exists and is finite, then  $\lambda_n^\text{GRO}(r_n,z_n)=1\iff \mathbb{E}^{Q_n}[1/e(X_n,r_n,z_n)]\leq 1$.
\end{enumerate}
\end{lemma}
\begin{proof}
Throughout the proof, let $Y_n:=e(X_n,r_n,z_n)$.

For showing Assertion (i), notice that,  
    since $ \max_{\lambda\in [0,1]}\mathbb{E}^{Q_n}[\log(1-\lambda+\lambda Y_n)]\geq 0=\mathbb{E}^{Q_n}[\log(1-0+0 Y_n)]$ always, $\lambda_n^\text{GRO}(r_n,z_n)=0$ is equivalent to no $\lambda\in [0,\gamma]$  existing such that $\mathbb{E}^{Q_n}[\log(1-\lambda+\lambda Y_n)]>0$. By Lemma \ref{PosEPower} this is equivalent to $\mathbb{E}^{Q_n}[Y_n]\leq 1$.
    
For showing Assertion (ii), assume first that
    $\lambda_n^\text{GRO}(r_n,z_n)=1$. By Lemma \ref{GROnieNull}, we then have $Q_n(Y_n=0)=0$. Now, consider the function $f:[0,1]\to \R, \lambda\to \mathbb{E}^{Q_n}[\log(1-\lambda+\lambda Y_n)]$. Since $f$ is maximized for $\lambda=1$, we have for all $\lambda \in [0,1)$
    \begin{align*}
      0\leq \frac{f(1)-f(\lambda)}{1-\lambda}&=\mathbb{E}^{Q_n}\left[\frac{\log(Y_n)-\log(1-\lambda+\lambda Y_n)}{1-\lambda}\right]\\ 
      &=\mathbb{E}^{Q_n}\left[(Y_n-1)\frac{\log(Y_n)-\log(1-\lambda+\lambda Y_n)}{Y_n-(1-\lambda+\lambda Y_n)}\right] \, .
    \end{align*}
    By the intermediate value theorem, we have \[
    \frac{1}{a}\leq\frac{\log(a)-\log(b)}{a-b}\leq \frac{1}{b}
    \] for all $0<b<a$. We infer using  $Q_n(Y_n=0)=0$ that 
    \begin{align*}
     0\leq \mathbb{E}^{Q_n}\left[\frac{Y_n-1}{1-\lambda+\lambda Y_n}\right] \, .
    \end{align*}
  Now, notice that it holds for all $\lambda\in [0.5,1]$ and $y>0$ that $\frac{y-1}{y} \leq \frac{y-1}{1-\lambda+\lambda y}\leq \max\{0,2\frac{y-1}{y}\}$. Since $\mathbb{E}^{Q_n}[\frac{y-1}{y}]=\mathbb{E}^{Q_n}[1-\frac{1}{y}]\in \R$ exists by assumption, the dominated convergence theorem implies
  \begin{align*}
    0\leq \lim_{\lambda\to 1} \mathbb{E}^{Q_n}\left[\frac{Y_n-1}{1-\lambda+\lambda Y_n}\right] =\mathbb{E}^{Q_n}\left[\lim_{\lambda \to 1}\frac{Y_n-1}{1-\lambda+\lambda Y_n}\right] =\mathbb{E}^{Q_n}\left[\frac{Y_n-1}{Y_n}\right] \,  ,
  \end{align*}
  which directly yields $\mathbb{E}^{Q_n}\left[1 / Y_n\right]\leq \mathbb{E}^{Q_n}\left[1\right]=1$. 

To establish the reverse implication, assume 
$\mathbb{E}^{Q_n}[1/Y_n]\leq 1$. Then, we have for all $\lambda \in [0,1]$ that
  \begin{align*}
   \mathbb{E}^{Q_n}[\log(1-\lambda+\lambda Y_n)-\log(Y_n)]&=\mathbb{E}^{Q_n}\left[\log\left(\frac{1-\lambda}{Y_n}+\lambda\right)\right] \\ 
   &\leq \mathbb{E}^{Q_n}\left[\frac{1-\lambda}{Y_n}+\lambda-1\right] \\ 
   &=(1-\lambda)\left(\mathbb{E}^{Q_n}\left[\frac{1}{Y_n}\right]-1\right) \leq 0 \, .
  \end{align*}
  Therefore, the maximum of the function $f(\lambda)=\mathbb{E}^{Q_n}[\log(1-\lambda+\lambda Y_n)]$ is $\log(Y_n)$ which corresponds to $\lambda_n^\text{GRO}(r_n,z_n)=1$. 
\end{proof}

The obvious problem of the GRO betting process in practice is that one needs to specify the probability distribution $Q_n$ under the alterantive, which may often be infeasible. In the sequel, we thus introduce three alternative betting processes that try to approximate $Q_n$ in different ways. The first of these betting processes is the growth rate for empirical e-statistics (GREE) betting process, which tries to estimate $Q_n$ by the behaviour of the observed e-process up to time point $n-1$.
\begin{definition}\label{GREE}
 For each $n\in \N$, let $E_n$ be a random variable following the empirical distribution of $(e(X_m,r_m,z_m))_{1\leq m\leq n-1}$. Then, the growth rate for empirical e-statistics (GREE)  betting process 
 $\lambda^{\text{GREE}}:=\left(\lambda_n^\text{GREE}\right)_{n\in \N}$ for $(M_n)_{n\in \N}$ is given by
  \begin{align*}
  \lambda^{\text{GREE}}_n(r_n,z_n):&= \argmax_{\lambda \in [0,\gamma]}\mathbb{E}[\log(1-\lambda +\lambda E_n)]\\
  &= \argmax_{\lambda \in [0,\gamma]}\frac{1}{n-1}\sum_{m=1}^{n-1}\log(1-\lambda+\lambda e(X_m,r_m,z_m)), 
  \end{align*} 
  for fixed $\gamma\in (0,1)$. 
\end{definition}
Another approach to estimating $Q_n$ lies in observing the loss process $(X_n)_{n\in \N}$ up to a time point $n-1$. This approach leads to the growth rate for empirical losses (GREL) betting process, defined as follows.
\begin{definition}\label{GREL}
For each $n\in \N$, let $L_n$ be a random variable following the empirical distribution of $(X_m)_{1\leq m\leq n-1}$. Then, the growth rate for empirical losses (GREL)  betting process $\lambda^{\text{GREL}}:=\left(\lambda_n^\text{GREL}\right)_{n\in \N}$ for $(M_n)_{n\in \N}$ is given by
  \begin{align*}
  \lambda^{\text{GREL}}_n(r_n,z_n):&= \argmax_{\lambda \in [0,\gamma]}\mathbb{E}[\log(1-\lambda +\lambda e(L_n,r_n,z_n))]\\
  &= \argmax_{\lambda \in [0,\gamma]}\frac{1}{n-1}\sum_{m=1}^{n-1}\log(1-\lambda+\lambda e(X_m,r_n,z_n)), 
  \end{align*} 
  for fixed $\gamma\in (0,1)$.   
\end{definition}
For the GREL betting process, one thus applies the current risk estimates to the losses observed in the past and optimizes the betting process with respect to the implied distribution function. This approach can obviously only work well if the past loss distributions are ``similar'' to the current loss distribution at time point $n$. However, this method can detect singular underestimation of risk (i.\ e., the risk being only underestimated at a few time points, but estimated correctly or overestimated at most time points) remarkably well.

Finally, we consider the growth rate for empirical mixture (GREM)  betting process.
\begin{definition}\label{GREM}
 The growth rate for empirical mixture (GREM)  betting process $\left(\lambda_n^\text{GREM}\right)_{n\in \N}$ for $(M_n)_{n\in \N}$ is given by
  \begin{align*}
  \lambda^{\text{GREM}}_n(r_n,z_n):=\frac{M_{n-1}(\lambda^\text{GREE})\lambda_n^\text{GREE}+M_{n-1}(\lambda^\text{GREL})\lambda_n^\text{GREL}}{M_{n-1}(\lambda^\text{GREE})+M_{n-1}(\lambda^\text{GREL})}. 
  \end{align*} 
\end{definition}
The definition of this betting process is mainly due to the following lemma, which is a consequence of Lemma 1 of \textcite{VovkWangMerging}; see also Section 5.1 of \textcite{Hauptquelle}.
\begin{lemma}\label{GREMGREEGREL}
  For all $n\in \N_0$, it holds that
  \begin{align*}
M_n(\lambda^\text{GREM})=\frac{M_n(\lambda^\text{GREE})+M_n(\lambda^\text{GREL})}{2}.
  \end{align*}
\end{lemma}
\begin{bemerkung}
By \textcite{VovkWangMerging}, even any arbitrary convex combinations of test supermartingales defined as in Lemma \ref{ProcessisMartingale} will be of the same form for an appropriate betting process. Thus, one could also consider weighting the \text{GREE} and \text{GREL} processes differently or include other e-processes of the same form as in Lemma \ref{ProcessisMartingale}.    
\end{bemerkung}
Another possibility for defining the betting process $\lambda$ consists in only considering the last $T$ observed losses and/or e-statistics for calculating the GREE, GREL, and GREM betting processes, instead of all observations up to a time point $n\in \N$. This leads to the following definition.
\begin{definition}\label{FiniteHorizon}
   Let $n,T\in \N$ and define $n_0:=\max\{1,n-T\}$. Furthermore, let $E_n$ be a random variable following the empirical distribution of $(e(X_m,r_m,z_m)_{n_0\leq m\leq n-1}$ and $L_n$ be a random variable following the empirical distribution of $(X_m)_{n_0\leq m\leq n-1}$.
   \begin{enumerate}[(i)]
    \item 
    The $T$-finite horizon growth rate for empirical e-statistics ($T$-FH GREE)  betting process $\lambda^{\text{$T$-FH GREE}}:= \left(\lambda_n^\text{$T$-FH GREE}\right)_{n\in \N}$ for $(M_n)_{n\in \N}$ is given by
    \begin{align*}
  \lambda^{\text{$T$-FH GREE}}_n(r_n,z_n):&= \argmax_{\lambda \in [0,\gamma]}\mathbb{E}[\log(1-\lambda +\lambda E_n)]\\
  &= \argmax_{\lambda \in [0,\gamma]}\frac{1}{n-n_0}\sum_{m=n_0}^{n-1}\log(1-\lambda+\lambda e(X_m,r_m,z_m)), 
  \end{align*} 
  for fixed $\gamma\in (0,1)$.
  \item 
  The $T$-finite horizon growth rate for empirical losses ($T$-FH GREL)  betting process $\lambda^{\text{$T$-FH GREL}}:= \left(\lambda_n^\text{$T$-FH GREL}\right)_{n\in \N}$ for $(M_n)_{n\in \N}$ is given by
  \begin{align*}
  \lambda^{\text{$T$-FH GREL}}_n(r_n,z_n):&= \argmax_{\lambda \in [0,\gamma]}\mathbb{E}[\log(1-\lambda +\lambda e(L_n,r_n,z_n))]\\
  &= \argmax_{\lambda \in [0,\gamma]}\frac{1}{n-n_0}\sum_{m=n_0}^{n-1}\log(1-\lambda+\lambda e(X_m,r_n,z_n)), 
  \end{align*} 
  for fixed $\gamma\in (0,1)$. 
  \item 
  The $T$-finite horizon growth rate for empirical mixture ($T$-FH GREM)  betting process $\lambda^{\text{$T$-FH GREM}}:= \left(\lambda_n^\text{$T$-FH GREM}\right)_{n\in \N}$ for $(M_n)_{n\in \N}$ is given by
  \begin{align*}
  &\lambda^{\text{$T$-FH GREM}}_n(r_n,z_n) :=\\ 
  &~~~~\frac{M_{n-1}(\lambda^\text{$T$-FH GREE})\lambda_n^\text{$T$-FH GREE}+M_{n-1}(\lambda^\text{$T$-FH GREL})\lambda_n^\text{$T$-FH GREL}}{M_{n-1}(\lambda^\text{$T$-FH GREE})+M_{n-1}(\lambda^\text{$T$-FH GREL})}. 
  \end{align*} 
   \end{enumerate}
\end{definition}

\subsection{Calculating the betting process}\label{sec42}
As mentioned in the previous section, the GRO betting process, as well as the GREE, GREL, and GREM betting processes, can be calculated via a convex program. However, due to the number of past time steps to be considered often becoming large, the exact computation of these betting processes can become slow in practice, especially considering that the betting process needs to be calculated at every time step. 

A possible solution to this relies on a second order Taylor approximation of the function $\log(1+x)\approx x- x^2 / 2$. This leads to the following approximate formulas for the GREE and GREL betting processes.
\begin{align*}
  \lambda_n^\text{GREE}(r_n,z_n)&\approx \argmax_{\lambda \in [0,\gamma]} \frac{1}{n-1}\sum_{m=1}^{n-1} \left\{\lambda(e(X_m,r_m,z_m)-1)-\right.\\
  &~~~~~~~~~~\left.\frac{\lambda^2 (e(X_m,r_m,z_m)-1)^2}{2}\right\} \\ 
&=\max\left\{0,\min\left\{\gamma, \frac{\sum_{m=1}^{n-1} (e(X_m,r_m,z_m)-1)}{\sum_{m=1}^{n-1}(e(X_m,r_m,z_m)-1)^2}\right\}\right\}  \\ 
\lambda_n^\text{GREL}(r_n,z_n)&\approx \argmax_{\lambda \in [0,\gamma]} \frac{1}{n-1}\sum_{m=1}^{n-1} \left\{\lambda(e(X_m,r_n,z_n)-1)-\right.\\
&~~~~~~~~~~\left.\frac{\lambda^2 (e(X_m,r_n,z_n)-1)^2}{2}\right\} \\ 
&=\max\left\{0,\min\left\{\gamma, \frac{\sum_{m=1}^{n-1} (e(X_m,r_n,z_n)-1)}{\sum_{m=1}^{n-1}(e(X_m,r_n,z_n)-1)^2}\right\}\right\}  \, .
\end{align*}
Letting $n_0:=\max\{1,n-T\}$, the $T$-FH-GREE and $T$-FH-GREL betting processes can be approximated in a similar manner, leading to the formulas
{\small
\begin{align*}
  \lambda_n^\text{$T$-FH GREE}(r_n,z_n)&\approx \argmax_{\lambda \in [0,\gamma]} \frac{1}{n-n_0}\sum_{m=n_0}^{n-1} \left\{\lambda(e(X_m,r_m,z_m)-1)-\right.\\
  &~~~~~~~~~~~~~\left.\frac{\lambda^2 (e(X_m,r_m,z_m)-1)^2}{2}\right\} \\ 
&=\max\left\{0,\min\left\{\gamma, \frac{\sum_{m=n_0}^{n-1} (e(X_m,r_m,z_m)-1)}{\sum_{m=n_0}^{n-1}(e(X_m,r_m,z_m)-1)^2}\right\}\right\},  \\ 
\lambda_n^\text{$T$-FH GREL}(r_n,z_n)&\approx \argmax_{\lambda \in [0,\gamma]} \frac{1}{n-n_0}\sum_{m=n_0}^{n-1} \left\{\lambda(e(X_m,r_n,z_n)-1)-\right.\\
&~~~~~~~~~~~~~\left.\frac{\lambda^2 (e(X_m,r_n,z_n)-1)^2}{2}\right\} \\ 
&=\max\left\{0,\min\left\{\gamma, \frac{\sum_{m=n_0}^{n-1} (e(X_m,r_n,z_n)-1)}{\sum_{m=n_0}^{n-1}(e(X_m,r_n,z_n)-1)^2}\right\}\right\}.
\end{align*}}  

The GREM as well as the $T$-FH GREM betting process can then be approximated by plugging in the approximated GREE and GREL or $T$-FH GREE and $T$-FH GREL betting processes, respectively,  into the formulas from Lemma \ref{GREMGREEGREL} or Part (iii) of Definition \ref{FiniteHorizon}, respectively. 
\begin{bemerkung}
The approximations presented here perform well if the observed terms $\lambda_n (e(X_n,r_n,z_n)-1)$ are close to $0$. Since the betting process usually tends to be small, this is typically the case. One can even improve on this approximation by reducing the bound $\gamma$ on the search space for the $\lambda_n$. However, one should keep in mind that the approximations might be quite inaccurate if one observes high values of the e-statistic $e(X_n,r_n,z_n)$. One example of this problem can be inferred from considering the e-statistic for the $99\%$ expected shortfall derived from Lemma \ref{BayesPairESVaR} and imagining an observed excess over the projected value-at-risk twice as large as projected. Then, even when considering a relatively small $\lambda$ of $0.01$, we get $\log(1+\lambda (e(x,r,n)-1))\approx \log(3)\approx 1.1$, but 
\[
\lambda (e(x,r,n)-1)-\frac{(\lambda (e(x,r,n)-1))^2}{2}\approx 0. 
\]
Thus, calculation through the approximated formula tends to prefer smaller values of $\lambda$ than in the original formula. To remedy this, one might consider Taylor approximations of the function $\log(1+x)$ which are of higher order.
\end{bemerkung}
%
\subsection{Asymptotic optimality of betting processes}
In this section, we analyze the asymptotic optimality of our proposed betting processes. In this, we call a betting process asymptotically optimal if the asymptotic growth rate of the resulting e-process is close to the asymptotic growth rate of the e-process based on the GRO betting process. 
\begin{definition}\label{AsymEqui}
 Let $\lambda=(\lambda_n)_{n\in \N}$ and $\lambda'=(\lambda'_n)_{n\in \N}$ be betting processes under the conditions of Lemma \ref{ProcessisMartingale}. If, for the true probability measure $\mathbb{P}$, we have
 \begin{align*}
  \mathbb{E}^\mathbb{P}\left[\left| \frac{1}{n} (\log M_n(\lambda)-\log M_n(\lambda'))\right|\right] \xrightarrow{n\to \infty} 0,
 \end{align*}
 we call $\lambda$ and $\lambda'$ asymptotically equivalent and denote this by $\lambda \simeq \lambda'$.
\end{definition}

\begin{definition}\label{AsymOpti}
 Let $\lambda=(\lambda_n)_{n\in \N}$ be a betting process under the conditions of Lemma \ref{ProcessisMartingale}. If $\lambda \simeq \lambda^\text{GRO}$, we call $\lambda$ asymptotically optimal.
\end{definition}
Theorem \ref{AsymOptiofAlgs} shows that the proposed betting processes are asymptotically optimal in specific cases. For technical details, we refer to Section 5.2 of \textcite{Hauptquelle}.
\begin{theorem}\label{AsymOptiofAlgs}
  Under the conditions of Lemma \ref{ProcessisMartingale}, assume further that $\sup_{n\in \N}\mathbb{E}^{Q_n}[|\log(e(X_n,r,z))]<\infty$ for all $(r,z) \in \psi(\mathcal{P})$. Then, the following assertions hold true for any $\gamma\in (0,1)$.
  \begin{enumerate}[(i)]
    \item 
   If $(e(X_n,r_n,z_n))_{n\in \N}$ are i.i.d. and $e(X_n,r_n,z_n)$ is independent of $\mathcal{F}_{n-1}$ for all $n\in \N$, $\lambda^{\text{GREE}}$ is asymptotically optimal. 
   \item 
If $(X_n)_{n\in \N}$ are i.i.d., $X_n$ is independent of $\mathcal{F}_{n-1}$ for all $n\in \N$ and either 
\begin{enumerate}[a)]
  \item $\{(r_n,z_n),n\in\N\}$ is finite, or
  \item $\{(r_n,z_n),n\in\N\}$ is compact, $(r_n,z_n)\xrightarrow{\mathbb{P}}(r_0,z_0)$ for some $(r_0,z_0)\in \R^d$, and $e(x,r,z)$ is continuous in $(r,z)$ for all $x\in \R$, 
\end{enumerate}
then $\lambda^{\text{GREL}}$ is asymptotically optimal. 
\item 
If $\lambda^{\text{GREE}}$ or $\lambda^{\text{GREL}}$ are asymptotically optimal, then $\lambda^\text{GREM}$ is also asymptotically optimal.
  \end{enumerate}
\end{theorem}
%
\subsection{Practical comparison of betting processes}
In this section, we illustrate the concept of asymptotic optimality of the proposed betting processes for different scenarios by looking at appropriate examples.

\begin{beispiel}
First, we consider a sequence of independent, homoscedastic normally distributed random variables, the expectation of which changes periodically. To this end, let $(X_t)_{t\in \N}$ be independent and $X_t\sim \mathcal{N}(\sin(\frac{2 \pi t}{250}),1) $. Here, the period is chosen such that it represents a full banking year (in daily units). This behaviour of loss processes is common in the financial world where expected profits often fluctuate according to certain annual patterns. Now assume that the true values for the $99\%$ expected shortfall and value-at-risk at each time point $t\in \N$ are underestimated by a constant amount $c$, e.\ g., $0.5$, and the predictions of the expected shortfall are to be backtested by our e-backtesing procedure described in Section \ref{sec3}. By Theorem \ref{AsymOptiofAlgs}, the GREE as well as the GREM betting processes are asymptotically optimal, since 
\[
e^{\ES}_\alpha(X_t,r_t,z_t)=\frac{(Z_t-\Phi^{-1}(0.99)+c)_+}{\phi(\Phi^{-1}(0.99))-0.01\Phi^{-1}(0.99)},
\]
where $(Z_t)_{t\in \N}$ are i.i.d and standard normally distributed.
 Thus, we expect these processes to outperform the GREL betting process. To test this, we have run $10{,}000$ Monte Carlo simulations for each of the betting processes and applied our backtesting procedure using the respective betting process for a maximum amount of $1{,}000$ trading days. To calculate the betting processes, we use the approximate formulas introduced in Section \ref{sec42}. We have recorded for each run, whether the e-process exceeds the values $2.2$, $3.5$, $9.0$, and $20$, respectively. (For more details on why we chose these thresholds, see Section \ref{sec5}). We have also tracked the log-transformed e-process $(\log(M_t))_{1\leq t \leq 1000}$ during each Monte Carlo run.  

Surprisingly, the percentage of threshold violations using the GREL betting process does not differ much from the percentage of violations when using the GREE or GREM betting processes; see Table \ref{TabellePeriode}. However, looking at the log-transformed e-processes (see Figure \ref{periodic_log_processes_fig}), we see that on average the GREL betting process performs worst, which is in line with Theorem \ref{AsymOptiofAlgs}.
\end{beispiel}

\setlength{\tabcolsep}{1pt}
\begin{table}[!htb]
    \centering
    \resizebox{\textwidth}{!}{
    \begin{tabular}{c|c|c|c|c}
    Threshold & \textbf{GREE} & \textbf{GREL} & \textbf{GREM} & \textbf{$50$-FH GREM}  \\
         \hline
         2.2 & 99.78  (99.69, 99.87) & 99.83 (99.75, 99.91) & 99.84 (99.76, 99.92) & 99.72 (99.62,99.82) \\ 
         3.5 & 99.46 (99.32, 99.60) &  99.51 (99.37, 99.65) &99.67 (99.58, 99.78) & 99.25 (99.08,99.42)\\ 
         9 & 98.01 (97.74, 98.28) & 97.35 (97.04, 97.66) & 98.40 (98.15, 98.65) & 96.48 (96.12, 96.84)\\ 
         20 & 95.32 (94.91, 95.73) & 92.75 (92.24, 93.26) & 95.31 (94.90, 95.72) & 91.94 (91.41, 92.47)
    \end{tabular}
    }
    \caption[Percentages of threshold violations for forecasts underestimating $\ES_{0.99}$ and $\VaR_{0.99}$ by $0.5$ for normally distributed RVs with periodic mean]{Percentage of threshold violations for forecasts underestimating $\ES_{0.99}$ and $\VaR_{0.99}$ by $0.5$ for normally distributed RVs with periodic mean. All values are based on $10{,}000$ Monte Carlo simulation runs each and a maximum sample size of $1{,}000$. $95\%$ confidence intervals are given in brackets.}
    \label{TabellePeriode}
\end{table}
\setlength{\tabcolsep}{4pt}

\begin{figure}[!htb]
  \centering
  \includegraphics[width=0.8\textwidth]{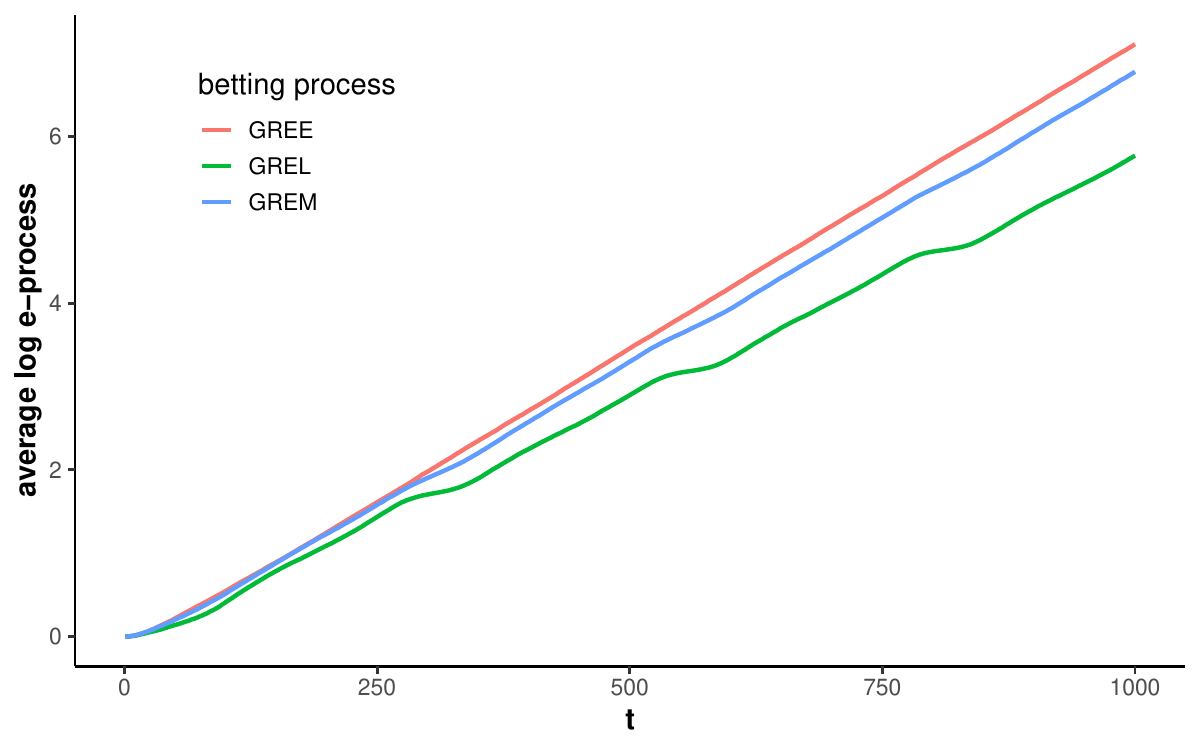}
  \caption[Average log-transformed e-processes for constant underestimation of $ES_{0.99}$ and $\VaR_{0.99}$ for normally distributed RVs with periodic means]{Average log-transformed e-processes for constant underestimation of $ES_{0.99}$ and $\VaR_{0.99}$ for normally distributed RVs with periodic means using different betting processes based on $10,000$ Monte Carlo simulations each.}
  \label{periodic_log_processes_fig}
\end{figure}

We also compared the GREM betting process with the $T$-FH GREM betting process, for $T=50$. The results for this comparison are provided in Table \ref{TabellePeriode} and Figure \ref{periodic_log_processes_FH_fig}. Here, we notice that the $T$-FH GREM betting process performs worse than all other considered betting processes. However, this underperformance is only marginal and, in particular, the e-process when using this betting process does not tend to degenerate (see Figure \ref{periodic_log_processes_FH_fig}). Using the $T$-FH GREM betting process thus might still yield acceptable performances under this scenario. 

\begin{figure}[!htb]
  \centering
  \includegraphics[width=0.8\textwidth]{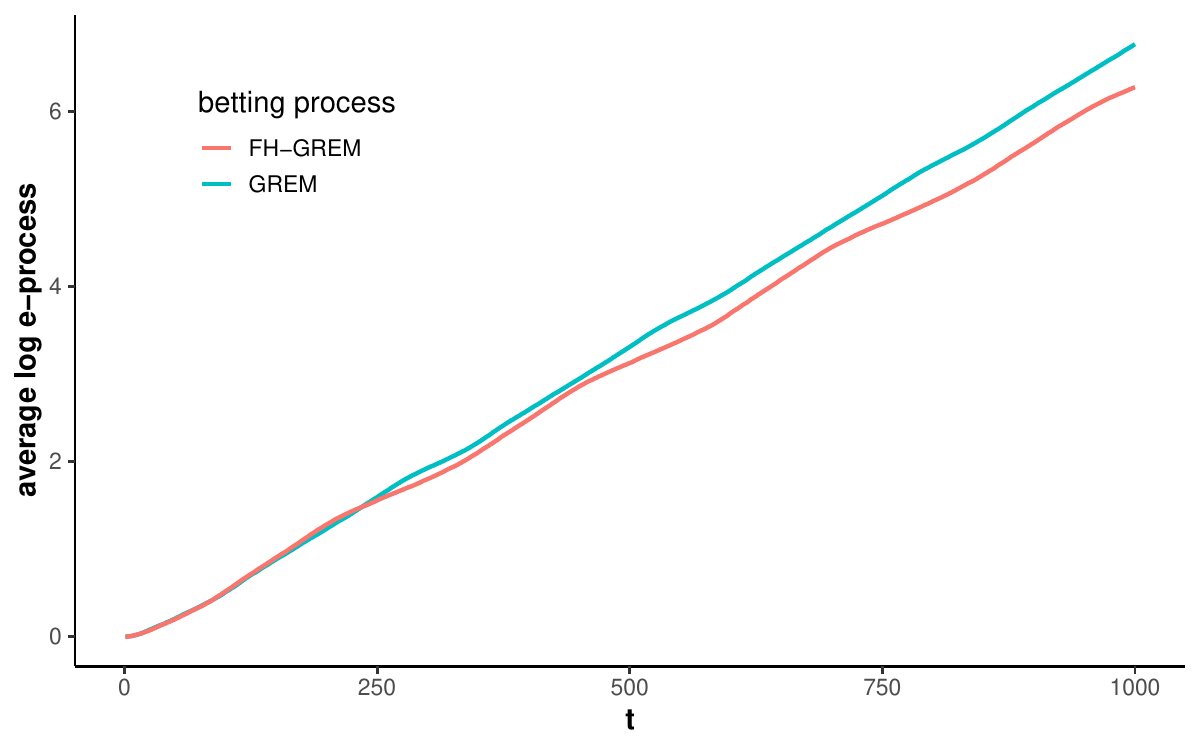}
  \caption[Average log-transformed e-processes using \textbf{GREM} or \textbf{$50$-FH GREM} betting processes for constant underestimation of $ES_{0.99}$ and $\VaR_{0.99}$ for normally distributed RVs with periodic means]{Average log-transformed e-processes using \textbf{GREM} or \textbf{$50$-FH GREM} betting processes for constant underestimation of $ES_{0.99}$ and $\VaR_{0.99}$ for normally distributed RVs with periodic means based on $10,000$ Monte Carlo simulations each.}
  \label{periodic_log_processes_FH_fig}
\end{figure}

\begin{beispiel}
Next, we consider a sequence of standard normally distributed i.i.d. loss variables $(X_t)_{t\in \N}$ and assume that the predictions $z_t$ and $r_t$ for the $99 \%$ value-at-risk and expected shortfall at each time point $t\in \N$ are drawn such that $z_t\sim 2+2\cdot\operatorname{Ber}(0.5)$ and $r_t\sim z_t+ 0.2+0.8\cdot \operatorname{Ber}(0.5)$, where $\operatorname{Ber}(p)$ denotes the Bernoulli distribution with success parameter $p$. Since $\VaR_{0.99}(X_1)\approx 2.33$ and $\ES_{0.99}(X_1)\approx 2.67$, the expected shortfall will be overestimated in most cases. When using the described e-backtesting procedure for this example, we thus suspect the GREE betting process to perform worse than the GREL and GREM betting processes. Namely, by Part (ii).a) of Theorem \ref{AsymOptiofAlgs}, the GREL and GREM betting processes are asymptotically optimal for this example. 

Indeed, our results demonstrate that the threshold violation percentages are substantially higher when using the GREL and GREM betting processes; see Table \ref{TabelleDiscrPreds}. The superiority of these betting processes is even more pronounced when looking at the average log-transformed e-processes; see Figure \ref{discr_preds_fig}. 

\setlength{\tabcolsep}{0.5pt}
\begin{table}[!htb]
    \centering
    \resizebox{\textwidth}{!}{
    \begin{tabular}{c|c|c|c|c}
    Threshold & \textbf{GREE} & \textbf{GREL} & \textbf{GREM} & \textbf{$50$-FH GREM}   \\
         \hline
         2.2 & 14.71  (14.02, 15.40) & 66.16(65.23, 67.09) & 48.11 (47.13, 49.09) & 43.24 (42.27, 44.21) \\ 
         3.5 & 7.48 (6.96, 8.00) &  47.15 (46.17, 48.13) & 29.70 (28.80, 30.60) & 27.64 (26.76, 28.52)\\ 
         9 & 1.71 (1.46, 1.96) & 17.89 (17.14, 18.64) & 10.04 (9.45, 10.63) & 9.39 (8.82, 9.96)\\ 
         20 & 0.36 (0.24, 0.48) & 6.39 (5.91, 6.87) & 3.03 (2.69, 3.37) & 3.52 (3.16, 3.88)
    \end{tabular}
    }
    \caption[Percentage of threshold violations for forecasts, where estimates for $\ES_{0.99}$ and $\VaR_{0.99}$ are drawn from discrete distributions]{Percentage of threshold violations for forecasts, where estimates for $\ES_{0.99}$ and $\VaR_{0.99}$ are drawn from discrete distributions for standard normally distributed RVs based on $10,000$ Monte Carlo simulations each and a maximum sample size of $1,000$. $95\%$ confidence intervals in brackets.}
    \label{TabelleDiscrPreds}
\end{table}
\setlength{\tabcolsep}{4pt}

\begin{figure}[!htb]
  \centering
  \includegraphics[width=0.8\textwidth]{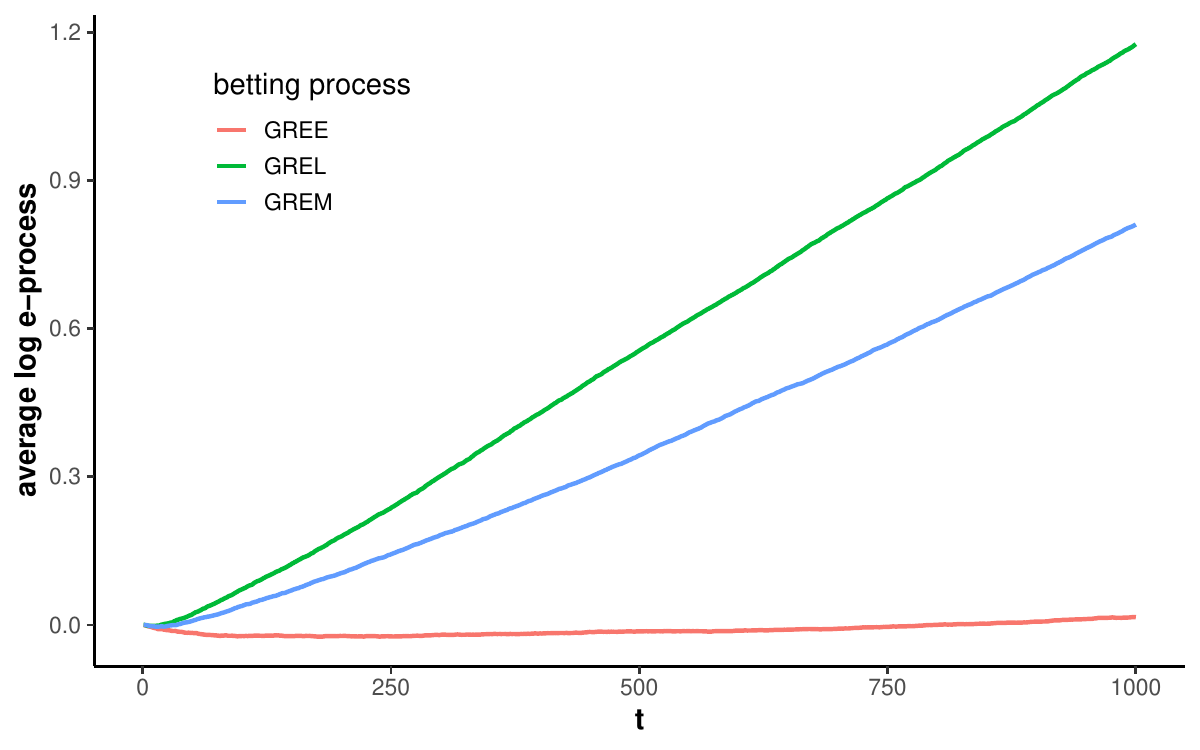}
  \caption[Average log-transformed e-processes where predictions for $ES_{0.99}$ and $\VaR_{0.99}$ are drawn from discrete distributions]{Average log-transformed e-processes where predictions for $ES_{0.99}$ and $\VaR_{0.99}$ are drawn from discrete distributions for standard normally distributed RVs using different betting processes based on $10,000$ Monte Carlo simulations each.}
  \label{discr_preds_fig}
\end{figure}

Again, we also compared the percentage of threshold violations as well as the average log-transformed e-processes when using the GREM betting process with those when using the $50$-FH GREM betting processes. These results are summarized in Table \ref{TabelleDiscrPreds} and Figure \ref{discr_preds_FH_fig}. Here, the $50$-FH GREM betting process performs very similar to the GREM betting process in terms of the threshold violation percentages, even exceeding those resulting from the GREM betting process for the threshold $20$. 

\begin{figure}[!htb]
  \centering
  \includegraphics[width=0.8\textwidth]{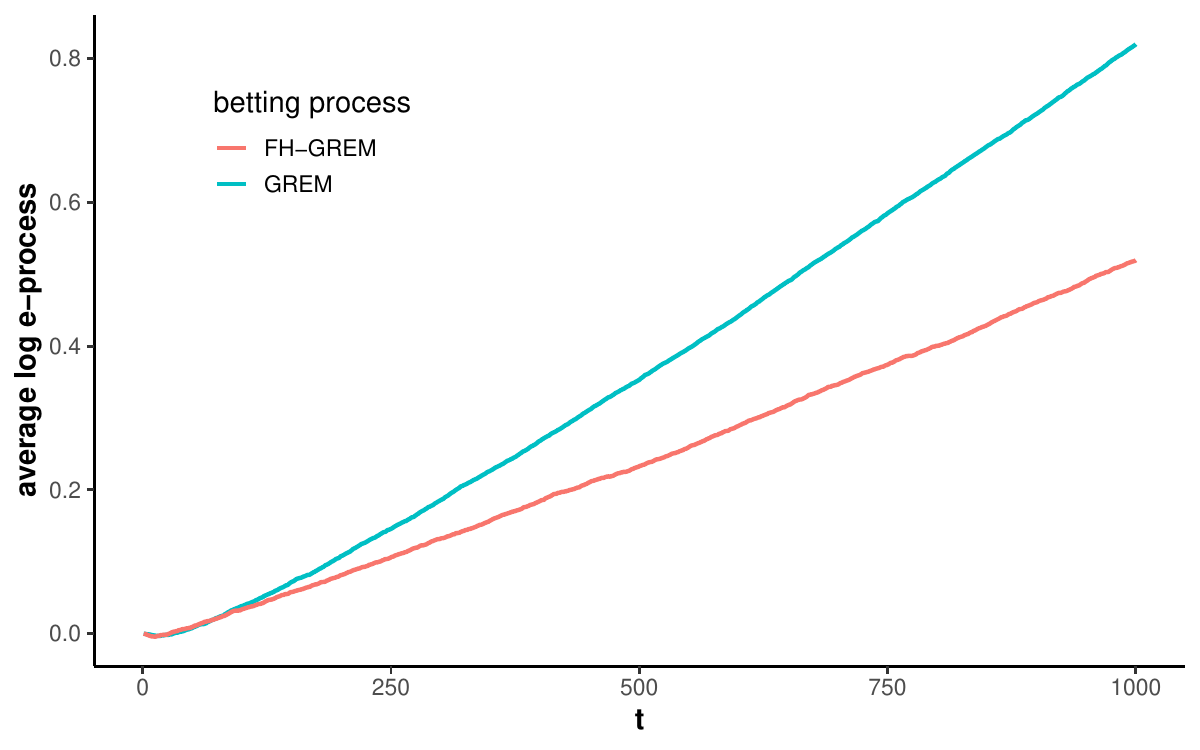}
  \caption[Average log-transformed e-processes using \textbf{GREM} or \textbf{$50$-FH GREM} betting processes where predictions for $ES_{0.99}$ and $\VaR_{0.99}$ are drawn from discrete distributions]{Average log-transformed e-processes using \textbf{GREM} or \textbf{$50$-FH GREM} betting processes where predictions for $ES_{0.99}$ and $\VaR_{0.99}$ are drawn from discrete distributions for standard normally distributed RVs based on $10,000$ Monte Carlo simulations each.}
  \label{discr_preds_FH_fig}
\end{figure}
\end{beispiel}

\begin{beispiel}
We conclude this section with an example illustrating a weakness of the GREE, GREL, and GREM betting processes that is addressed by their finite-horizon counterparts. For this, we consider a sequence $(X_t)_{t\in \N}$ of independent random variables such that $X_t\sim \mathcal{N}(0,\sin^2(\pi t/500))$ for all $t\in \N$. The value-at-risk and the expected shortfall at level $p=0.99$ are then estimated by assuming $X_t\sim\mathcal{N}(0,\sigma^2)$ with $\sigma=0.8$ for all $t\in \N$ and applying the well-known formulas for the value-at-risk and for the expected shortfall under a normal distribution (see, e.\ g., Example 2.18 in \cite{mcneil2005quantitative}). As the resulting estimates $r_t$ and $z_t$ for $\ES_{0.99}(X_t)$ and $\VaR_{0.99}(X_t)$ do not depend on $t$, the GREE, GREL, and GREM betting processes are all equal. We therefore backtest these predictions with our e-backtesting procedure using the GREM betting process and using the $50$-FH GREM betting process. Since, when considering this scenario, underestimations of the risk measures in question appear clustered in time, we expect the $50$-FH GREM betting process to outperform the GREM betting process, because  the former can optimize the bet over a small horizon where underestimations are more severe, while the latter always has to consider all observations up to time point $t\in \N$. Our results from our experiments based on $10{,}000$ Monte Carlo simulation runs are provided in Table \ref{TabelleVariance_FH} and Figure \ref{periodic_varaince_FH_fig}.

\begin{table}[!htb]
    \centering
    \begin{tabular}{c|c|c}
    Threshold & \textbf{GREM}/\textbf{GREE}/\textbf{GREL} & \textbf{$50$-FH GREM}   \\
         \hline
         2.2 & 55.58 (54.61, 56.55) & 78.36 (77.55, 79.17)\\ 
         3.5 & 35.66 (34.72, 36.60) & 65.37 (64.44, 66.30)\\ 
         9 &  10.47 (9.87, 11.07) & 37.88 (36.93, 38.83)\\ 
         20 & 2.78 (2.46, 3.10)& 18.80 (18.03, 19.57)
    \end{tabular}
    \caption[Percentage of threshold violations for forecasts calculating $ES_{0.99}$ and $\VaR_{0.99}$ under the assumption of homoscedasticity]{Percentage of threshold violations for forecasts calculating $ES_{0.99}$ and $\VaR_{0.99}$ under the assumption of homoscedasticity for normally distributed random variables with periodically changing variance based on $10{,}000$ Monte Carlo simulations each and a maximum sample size of $1{,}000$. $95\%$ confidence intervals are given in brackets.}
    \label{TabelleVariance_FH}
\end{table}

\begin{figure}[!htb]
  \centering
  \includegraphics[width=0.8\textwidth]{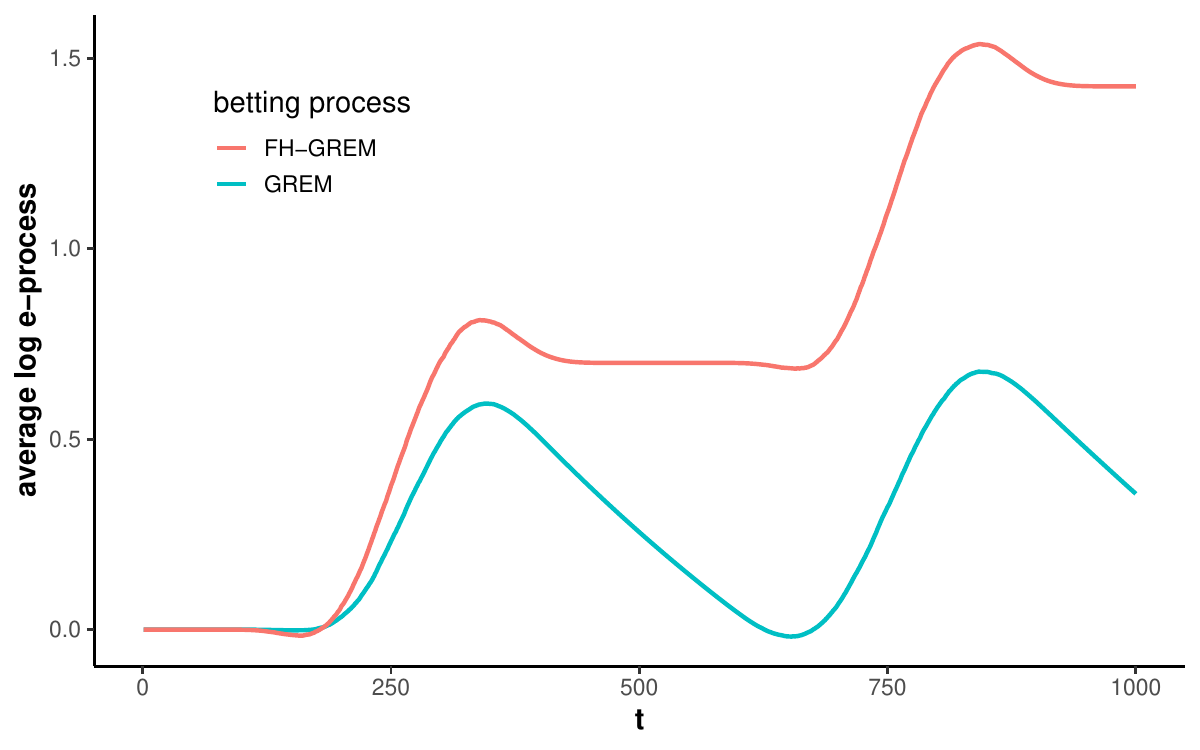}
  \caption[Average log-transformed e-processes using \textbf{GREM} or \textbf{$50$-FH GREM} betting processes for forecasts calculating $ES_{0.99}$ and $\VaR_{0.99}$ under the assumption of homoscedasticity]{Average log-transformed e-processes using \textbf{GREM} or \textbf{$50$-FH GREM} betting processes for forecasts calculating $ES_{0.99}$ and $\VaR_{0.99}$ under the assumption of homoscedasticity for normally distributed random variables with periodically changing variance based on $10,000$ Monte Carlo simulations each.}
  \label{periodic_varaince_FH_fig}
\end{figure}

As expected, here backtesting using the $50$-FH GREM betting process performed much better than using the infinite horizon betting process, especially when considering higher thresholds. When looking at the average log-transformed e-processes (see Figure \ref{periodic_varaince_FH_fig}), we see that the betting process when using the $50$-FH GREM e-process becomes mostly stable at time periods where the expected shortfall and value-at-risk are not being underestimated, while it decreases when using the GREM process, almost nullifying its increase when the respective risks are underestimated. This is due to the $50$-FH GREM betting process being able to optimize the bet based only on recent observations where no underestimations were present, while the GREM betting process always takes all observations into account, including those made under underestimations of the risk measures.
\end{beispiel}

Concluding from our experiments, we can identify scenarios where the GREE or the GREL  betting processes perform poorly. However, the GREM betting process, despite often yielding slightly worse results than the best process for a given scenario, has shown acceptable performances under all considered scenarios. Regarding the finite-horizon betting processes (in particular the $50$-FH GREM betting process), the performance using these processes is often slightly worse than their infinite-horizon counterparts. However, the usage of these betting processes can vastly improve the performance of the backtesting procedure if underestimations of the risk measures appear clustered in time. Thus, we recommend the usage of the GREM or $50$-FH GREM betting processes in practice, with the latter being suggested if one assumes underestimations of risk to be clustered, e.\ g., for autoregressive time series. One might also combine these two processes in a similar manner as in Lemma \ref{GREMGREEGREL}.

 \section{Choosing the significance threshold}\label{sec5}
An important question for our backtesting procedure and hypothesis testing with e-values in general concerns the choice of an appropriate threshold that the e-process needs to exceed in order to reject the associated null hypothesis. If one intends to limit the type-1 error probability of the resulting sequential test by some $\alpha \in (0,1)$, setting the threshold to $1/\alpha$ would guarantee this property by Corollary \ref{SupermartingaletoSequential}. However, the resulting sequential tests are often very conservative in most practical applications. This is due to these tests keeping their level under any possible distribution of the observed e-process that satisfies the null hypothesis. This also includes worst case distributions that are of high theoretical relevance but are unlikely to be observed in practice; cf. also \textcite{BlierWong2026} for a recent discussion on this topic. Thus, smaller thresholds than $1/\alpha$ are also considered. One approach to this, suggested by \cite{ShafThresh}, is to use the threshold $1 / \sqrt{\alpha}-1$ if one wants to construct an (approximate) level-$\alpha$ sequential test, leading, for example, to the (approximate) thresholds $2.2$, $3.5$, and $9.0$, respectively, for levels of $0.1$, $0.05$, and $0.01$, respectively.

We start our numerical investigation of the aforementioned suggestion of \textcite{ShafThresh} by considering a sequence of standard normally distributed i.i.d. random variables for which both the $99\%$ expected shortfall and value-at-risk are predicted correctly as described in Example 2.18 of \textcite{mcneil2005quantitative}. This prediction model is then backtested via our backtesting procedure detailed in Section \ref{sec3}. As the betting process, we use the GREM betting process, which in this case is equal to both the GREE and GREL betting processes due to the predictions not changing. We approximate the betting processes by the formulas derived in Section \ref{sec42}. We consider a maximum sample size of $1{,}000$ and have run $10{,}000$ Monte Carlo runs, examining how often the e-process exceeds the thresholds $2.2$, $3.5$, $9$, and $20$. 
The results for these experiments are detailed in Table \ref{TabelleStandNorm}. Our observed rejection percentages are roughly in line and usually slightly smaller than the levels proclaimed by \textcite{ShafThresh}. 

\begin{table}[!htb]
    \centering
    \begin{tabular}{c|c|c}
    Threshold &level according to \textcite{ShafThresh} in \% & percentage of rejections    \\
         \hline
         2.2 & 10 & 8.50  (7.95 , 9.05 ) \\ 
         3.5 & 5  & 3.70  (3.33 , 4.07 ) \\ 
         9 & 1  & 0.69  (0.53 , 0.85 ) \\ 
         20 & 0.23  & 0.16  (0.08 , 0.24 )
    \end{tabular}
    \caption[Percentage of rejections for true predictions of $\ES_{0.99}$ for normally distributed RVs]{Percentage of rejections for true predictions of $\ES_{0.99}$ for normally distributed random variables based on $10{,}000$ Monte Carlo simulation runs and a maximum sample size of $1{,}000$. $95\%$ confidence intervals are given in brackets.}
    \label{TabelleStandNorm}
\end{table}

Next, we performed the same experiments for a sequence of t-distributed i.i.d. random variables with 5 degrees of freedom. As argued by \cite{reverse-stress-testing} (see also the references therein), Student's $t$-distribution captures stylized facts about financial time series as, for instance, heavy tails. The $99\%$ expected shortfall and value-at-risk are again assumed to be predicted correctly (for their closed-form formulas, see, for example, \cite{norton2021calculating}). Our results for these experiments are summarized in Table \ref{TabelleTDistr}. Again, our observed rejection percentages are slightly smaller than the levels proclaimed by \textcite{ShafThresh}. They are also smaller than our observed rejection percentages in the standard normal case. 

\begin{table}[!htb]
    \centering
    \begin{tabular}{c|c|c}
    Threshold &level according to \textcite{ShafThresh} in \% & percentage of rejections    \\
         \hline
         2.2 & 10 & 7.12  (6.62 , 7.62 ) \\ 
         3.5 & 5  & 3.14  (2.80 , 3.48 ) \\ 
         9 & 1  & 0.50  (0.36 , 0.64 ) \\ 
         20 & 0.23  & 0.11  (0.05 , 0.17 )
    \end{tabular}
    \caption[Percentage of rejections for true predictions of $\ES_{0.99}$ for t-distributed RVs]{Percentage of rejections for true predictions of $\ES_{0.99}$ for $t$-distributed random variables with five degrees of freedom based on $10{,}000$ Monte Carlo simulation runs and a maximum sample size of $1{,}000$. $95\%$ confidence intervals are given in brackets.}
    \label{TabelleTDistr}
\end{table}

Finally, we have investigated generalized Pareto distributed excess distributions. For this, we let $p_1=0.95$ and assume a sequence $(L_t)_{t\in \N}$ of i.i.d loss variables such that $\VaR_{p_1}(L_1)=0$ and that $L_t|(L_t>0)$ is generalized Pareto distributed with shape parameter $\epsilon$ and scale parameter $\beta$. The value-at-risk and the expected shortfall at some level $p\geq p_1$ can then be calculated as described around Equations (7.18) and (7.19) in \textcite{mcneil2005quantitative}. We again conducted $10{,}000$ Monte Carlo simulation runs, testing our backtesting method using the GREM betting process for $\beta=1$ and different values for $\epsilon$. Our results are summarized in Table \ref{TabelleEVDThresh}.

\begin{table}[!htb]
    \centering
    \begin{tabular}{c|c|c}
    Threshold, $\epsilon$ &level according to \textcite{ShafThresh} in \% & percentage of rejections    \\
         \hline
         2.2, -0.2 & 10 & 8.67  (8.12 , 9.22 ) \\ 
         3.5, -0.2 & 5  & 3.72  (3.35 , 4.09 ) \\ 
         9, -0.2 & 1  & 0.65  (0.49 , 0.81 ) \\ 
         20, -0.2 & 0.23  & 0.16  (0.08 , 0.24 ) \\ 
         2.2, 0 & 10 & 7.74  (7.22 , 8.26 ) \\ 
         3.5, 0 & 5  & 3.44  (3.08 , 3.80 ) \\ 
         9, 0 & 1  & 0.55  (0.41 , 0.69 ) \\ 
         20, 0 & 0.23  & 0.18  (0.10 , 0.26 ) \\
         2.2, 0.2 & 10 & 6.50  (6.02 , 6.98 ) \\ 
         3.5, 0.2 & 5  & 2.88  (2.55 , 3.21 ) \\ 
         9, 0.2 & 1  & 0.51  (0.37 , 0.65 ) \\ 
         20, 0.2 & 0.23  & 0.07  (0.02 , 0.12 ) 
    \end{tabular}
    \caption[Percentage of rejections for true predictions of $\ES_{0.99}$ for GPD-distributed RVs]{Percentage of rejections for true predictions of $\ES_{0.99}$ for generalized Pareto (GPD)-distributed random variables with scale parameter $1$ and different shape parameters $\epsilon$ based on $10{,}000$ Monte Carlo simulations and a maximum sample size of $1{,}000$. $95\%$ confidence intervals are given in brackets.}
    \label{TabelleEVDThresh}
\end{table}

We notice again that our observed levels are close and slightly smaller compared to \textcite{ShafThresh}'s proclaimed levels. Considering the scale parameter $\epsilon$, the observed levels are smaller for smaller shapes. This is consistent with our previous observation about the cases of normally and t-distributed random variables that the observed levels are smaller for distributions with more pronounced right tails, as in the case of GPD-distributed variables with positive shape.   

Summarizing our observations, we conclude that the proposed thresholds by \textcite{ShafThresh} usually lead to slightly conservative sequential tests (keeping, but not exhausting the intended level). We thus recommend using these thresholds in practice and will also focus on these thresholds for the rest of this work.

\section{Sample size determination}\label{sec6}
In this section, we investigate methods for planning sample sizes for the e-backtesting procedure. In this, the goal is to achieve a specific power for specific regions in the space $\mathcal{P} \setminus \mathcal{P}_0$ of alternative distributions when applying a given significance threshold $\tau$ (cf.\  Section \ref{sec5}). For this, we focus on backtesting the expected shortfall and the value-at-risk by our procedure described in Section \ref{sec3}. In this, if not stated otherwise, we use the GREM betting process described in Section \ref{sec4}, which we approximate by the formulas derived in Section \ref{sec42}. 

In Sections \ref{sec61} - \ref{sec63}, we employ computer simulations for the determination of appropriate sample sizes. In Section \ref{sec64}, an analytical approach for determining sample size bounds is presented, which is applicable more generally than the simulations in Sections \ref{sec61} - \ref{sec63}.
\subsection{Underestimation of \texorpdfstring{ES$_p$}{the expected shortfall} for specific model classes}\label{sec61}
As a starting point, we assume a sequence of i.i.d loss variables $(X_t)_{t\in \N}$ where $X_1$ is distributed according to a stated distribution. For some $p \in (0,1)$ (e.\ g.,  $0.95$ or $0.99$), we assume the expected shortfall and possibly the value-at-risk at level $p$ to be misspecified. 
\subsubsection{Normal distribution}\label{sampsizesimulnormal}
First, let $X_1\sim \mathcal{N}(0,1)$. Using the formulas from Example 2.18 in \textcite{mcneil2005quantitative}, the true expected shortfall and value-at-risk at level $p$ are then readily available. Let the predictions for $\VaR_{p}$ and $\ES_{p}$ at each time point $t\in \N$ be given by $z_t:=d \cdot \VaR_{p}(X_1)$ and $r_t:=d \cdot \ES_{p}(X_1)$, where $d\in (0,1)$ is some factor. Under these specifications, we have applied our backtesting procedure for $1{,}000$ Monte-Carlo simulation runs each for $d \in \{0.5, 0.6, 0.7, 0.8, 0.9\}$ and a maximum sample size of $6{,}000$. Our simulated sample sizes under this scenario are summarized in Table \ref{TabelleNormSampleSizeFactor}.

\setlength{\tabcolsep}{3pt}
\begin{table}[!htb]
    \centering
    \begin{tabular}{c|c c c c|c c c c|c c c c}
    Power& \multicolumn{4}{c|}{70 \%} &  \multicolumn{4}{c|}{80 \%} & \multicolumn{4}{c}{90 \%}  \\
    \hline 
         \diagbox{d}{$\tau$} &2.2 & 3.5 & 9 & 20 & 2.2 & 3.5 & 9 & 20 & 2.2 & 3.5 & 9 & 20 \\ 
               \hline
               \multicolumn{13}{c}{\textbf{p=0.95}} \\
        0.5 & 22 & 28 & 41 & 51 & 27 & 33 & 47 & 59 & 35 & 42 & 56 & 72 \\ 
         0.6 & 33 & 43 & 62 & 78 & 41 & 51 & 73 & 92 & 53 & 66 & 91 & 111 \\ 
         0.7 & 58 & 73 & 107 & 135 & 75 & 92 & 130 & 163 & 103 & 124 & 164 & 201 \\ 
         0.8 & 119 & 156 & 240 & 302 & 155 & 197 & 285 & 355 & 211 & 267 & 362 & 436 \\ 
         0.9 & 559 & 762 & 1109 & 1348 & 742 & 974 & 1359 & 1610 & 1082 & 1280 & 1684 & 2072 \\
               \multicolumn{13}{c}{\textbf{p=0.99}} \\
         0.5 & 31 & 38 & 55 & 71 & 36 & 46 & 65 & 84 & 46 & 58 & 80 & 99 \\ 
         0.6 & 51 & 64 & 97 & 125 & 62 & 79 & 114 & 141 & 84 & 104 & 141 & 172 \\ 
         0.7 & 94 & 123 & 175 & 234 & 115 & 147 & 208 & 276 & 146 & 186 & 261 & 331 \\ 
         0.8 & 239 & 310 & 450 & 567 & 294 & 384 & 542 & 653 & 394 & 480 & 694 & 813 \\ 
         0.9 & 997 & 1323 & 1926 & 2465 & 1274 & 1708 & 2362 & 2861 & 1871 & 2287 & 2991 & 3617
    \end{tabular}
    \caption[Minimum sample sizes needed to achieve certain powers for forecasts 
    underestimating $\ES_{p}$ and $\VaR_{p}$ by a constant factor for normally distributed RVs]{Minimum sample sizes needed to achieve certain powers while using specific thresholds $\tau$ for forecasts 
    underestimating $\ES_{p}$ and $\VaR_{p}$ for levels $p \in \{0.95, 0.99\}$ by a constant factor $d$, for normally distributed random variables. All results are based on $1{,}000$ Monte Carlo simulation runs each and a maximum sample size of $6{,}000$.}
    \label{TabelleNormSampleSizeFactor}
\end{table}
\setlength{\tabcolsep}{4pt}

 The results presented in Table \ref{TabelleNormSampleSizeFactor} suggest that the required minimal sample size needed to achieve a certain power is roughly proportional to $(1-d)^{-2}$, which is in accordance to common sample size formulas. However, the required sample sizes might grow even faster for $d\to 1$. Focusing on the threshold $3.5$, we suggest a sample size of at least $14 / (1-d)^2$ for a power of $70\%$, of $18 / (1-d)^2$ for a power of $80 \%$,  and of $23 / (1-d)^2$ for a power of $90\%$.

Next, we look at the case of $z_t=\VaR_{p}(X_1)$, i.\ e., the value-at-risk being estimated correctly, and $r_t-z_t=d \cdot (\ES_{p}(X_1)-\VaR_{p}(X_1))$ for some $d\in (0,1)$, meaning that the difference between the expected shortfall and the value-at-risk is underestimated by a constant factor $d$. The results for this scenario are provided in Table \ref{TabelleNormSampleSizeVaRcorrect} and Figure \ref{fig_normal_VaR_correct}. 

\setlength{\tabcolsep}{1pt}
\begin{table}[!htb]
    \centering
    \resizebox{\textwidth}{!}{
    \begin{tabular}{c|c c c c|c c c c|c c c c}
    Power& \multicolumn{4}{c|}{70 \%} &  \multicolumn{4}{c|}{80 \%} & \multicolumn{4}{c}{90 \%}  \\
    \hline 
         \diagbox{d}{$\tau$} &2.2 & 3.5 & 9 & 20 & 2.2 & 3.5 & 9 & 20 & 2.2 & 3.5 & 9 & 20 \\ 
               \hline
         \multicolumn{13}{c}{\textbf{p=0.95}} \\      
        0.4 & 273 & 363 & 543 & 684 & 363 & 469 & 652 & 807 & 515 & 626 & 820 & 1006 \\ 
         0.5 & 408 & 550 & 791 & 1002 & 545 & 722 & 1005 & 1223 & 768 & 957 & 1321 & 1556 \\ 
         0.6 & 691 & 929 & 1382 & 1755 & 918 & 1205 & 1713 & 2089 & 1293 & 1631 & 2196 & 2579 \\ 
         0.7 & 1214 & 1671 & 2387 & 3016 & 1636 & 2204 & 3042 & 3856 & 2401 & 3090 & 4110 & 4950 \\ 
         0.8 & 2839 & 4037 & 4999 & >6000 & 4009 & 5412 & >6000 & >6000 & >6000 & >6000 & >6000 & >6000 \\
        \multicolumn{13}{c}{\textbf{p=0.99}} \\
         0.2 & 685 & 868 & 1311 & 1659 & 861 & 1069 & 1547 & 1949 & 1240 & 1464 & 1915 & 2402 \\ 
         0.3 & 1056 & 1353 & 1963 & 2494 & 1325 & 1639 & 2341 & 2865 & 1801 & 2158 & 2878 & 3546 \\ 
         0.4 & 1342 & 1847 & 2733 & 3494 & 1829 & 2307 & 3311 & 4194 & 2518 & 3122 & 4378 & 5230 \\ 
         0.5 & 2063 & 2880 & 4296 & 5347 & 2811 & 3670 & 5203 & >6000 & 4017 & 4967 & >6000 & >6000 \\ 
         0.6 & 3190 & 4569 & >6000 & >6000 & 4508 & 5896 & >6000 & >6000 & >6000 & >6000 & >6000 & >6000
    \end{tabular}
    }
    \caption[Minimum sample sizes needed to achieve certain powers for forecasts estimating $\VaR_{p}$ correctly and underestimating $\ES_{p}-\VaR_{p}$ by a constant factor for normally distributed RVs]{Minimum sample sizes needed to achieve certain powers while using specific thresholds $\tau$ for forecasts 
    estimating $\VaR_{p}$ correctly and underestimating $\ES_{p}-\VaR_{p}$ for levels $p=0.95,0.99$ by a constant factor $d$, for normally distributed random variables. All results are based on $1{,}000$ Monte Carlo simulation runs each and a maximum sample size of $6{,}000$.}
    \label{TabelleNormSampleSizeVaRcorrect}
\end{table}
\setlength{\tabcolsep}{4pt}

\begin{figure}[!htb]
  \centering
  \includegraphics[width=0.8\textwidth]{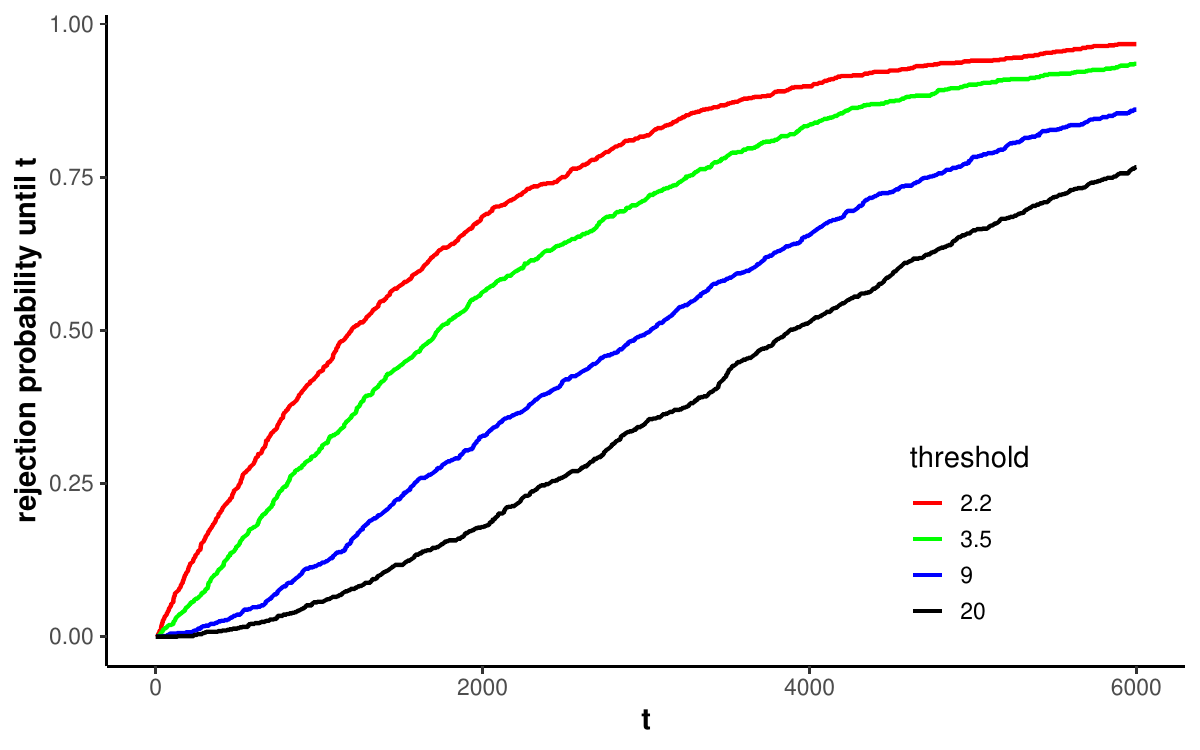}
  \caption[Cumulative rejection probabilities for forecasts estimating $\VaR_{0.99}$ correctly and underestimating $\ES_{0.99}-\VaR_{0.99}$ by a constant factor for normally distributed RVs]{Cumulative rejection probabilities for forecasts estimating $\VaR_{0.99}$ correctly and underestimating $\ES_{0.99}-\VaR_{0.99}$ by a constant factor $0.5$ for normally distributed random variables. The results are based on $1{,}000$ Monte Carlo simulations for different significance thresholds.}
  \label{fig_normal_VaR_correct}
\end{figure}

Comparing these results to those of our previous simulations, we notice that even seemingly large underestimations of the difference between the expected shortfall and the value-at-risk are hard to detect and require large sample sizes. However, when comparing the results, one needs to account for the true difference in this case being much smaller then the true value-at-risk and expected shortfall themselves. Thus, the actual estimates $r_t$ for the expected shortfall do not differ much from the true expected shortfall, even in the case of $d=0.2$. Accounting for this and evaluating the results under the different scenarios in dependence on the estimates $r_t$, we notice that the required sample sizes do not differ much. This suggests that the power of our backtesting procedure does not strongly depend on a possible misspecification of the value-at-risk.

Our results suggest that the required sample sizes needed to achieve specific powers are considerably smaller when backtesting the $95\%$ expected shortfall. This is due to the probability of some loss variable $X_t$ exceeding the projected value-at-risk being approximately five times as large as when backtesting at level $0.99$. Since our e-statistic is equal to zero when the projected value-at-risk is not exceeded, the loss variables exceeding their respective value-at-risk estimate $z_t$ are the only truly informative loss variables. Focussing again on the required sample sizes in dependence on the expected shortfall estimates $r_t$, we notice that the method performs similarly under all considered scenarios.

\subsubsection{Student's \texorpdfstring{$t$}{t}-distribution} 
Next, we investigate the case of $X_1\sim t(\nu)$ for some degrees of freedom $\nu >2$ (to ensure a finite variance of $X_1$). We first examine the scenario of the random variables being assumed as normally distributed with matching mean and variance which, using the formulas from Example 2.18 in \textcite{mcneil2005quantitative}, leads to the predictions $z_t:=\sqrt{\frac{\nu}{\nu-2}} \Phi^{-1}(p)$ for $\VaR_{p}(X_t)$ and $r_t:= \sqrt{\frac{\nu}{\nu-2}} \phi(\Phi^{-1}(p)) / (1-p)$ for $\ES_{p}(X_t)$. Again, we calculate the minimum required sample sizes to achieve certain powers when using the described e-backtesing procedure. Our results for the levels $p=0.99$ and $p=0.95$, degrees of freedom $\nu \in \{3, 5, 7, 10\}$, and common powers and thresholds are summarized in Table \ref{TabelletNorm99} and Figure \ref{fig_t_assumed_norm}.

\setlength{\tabcolsep}{1pt}
\begin{table}[!htb]
    \centering
    \resizebox{\textwidth}{!}{
    \begin{tabular}{c|c c c c|c c c c|c c c c}
    Power& \multicolumn{4}{c|}{70 \%} &  \multicolumn{4}{c|}{80 \%} & \multicolumn{4}{c}{90 \%}  \\
    \hline 
         \diagbox{$\nu$}{$\tau$} &2.2 & 3.5 & 9 & 20 & 2.2 & 3.5 & 9 & 20 & 2.2 & 3.5 & 9 & 20 \\ 
               \hline
        \multicolumn{13}{c}{\textbf{p=0.95}} \\
        3 & 2353 & 3246 & 4806 & 5925 & 3214 & 4163 & 5969 & >6000 & 4634 & 5494 & >6000 & >6000 \\ 
         5 & 1720 & 2370 & 3581 & 4491 & 2326 & 3058 & 4433 & 5407 & 3521 & 4269 & 5839 & >6000 \\ 
         7 & 2548 & 3678 & 5416 & >6000 & 3784 & 5066 & >6000 & >6000 & 5498 & >6000 & >6000 & >6000 \\ 
         10 & 3902 & 5527 & >6000 & >6000 & 5598 & >6000 & >6000 & >6000 & >6000 & >6000 & >6000 & >6000 \\ 
        \multicolumn{13}{c}{\textbf{p=0.99}} \\
         3 & 609 & 779 & 1139 & 1492 & 759 & 973 & 1387 & 1761 & 1062 & 1266 & 1801 & 2107 \\ 
         5 & 754 & 991 & 1448 & 1822 & 975 & 1218 & 1719 & 2114 & 1259 & 1615 & 2160 & 2610 \\ 
         7 & 992 & 1340 & 1966 & 2517 & 1328 & 1680 & 2408 & 2969 & 1927 & 2285 & 3086 & 3608 \\ 
         10 & 1589 & 2162 & 3167 & 3964 &2178 & 2712 & 3814 & 4840 & 3119 & 3747 & 4863 & 5901
    \end{tabular}
    }
    \caption[Minimum sample sizes needed to achieve certain powers for forecasts 
    estimating $\ES_{p}$ and $\VaR_{p}$ under the assumption of normality]{Minimum sample sizes needed to achieve certain powers while using specific thresholds $\tau$ for forecasts 
    estimating $\ES_{p}$ and $\VaR_{p}$ for levels $p \in \{0.95, 0.99\}$ under the assumption of normality for $t$-distributed random variables with $\nu$ degrees of freedom. All results are based on $1{,}000$ Monte Carlo simulation runs each and a maximum sample size of $6{,}000$.}
    \label{TabelletNorm99}
\end{table}
\setlength{\tabcolsep}{4pt}


\begin{figure}[!htb]
  \centering
  \includegraphics[width=0.8\textwidth]{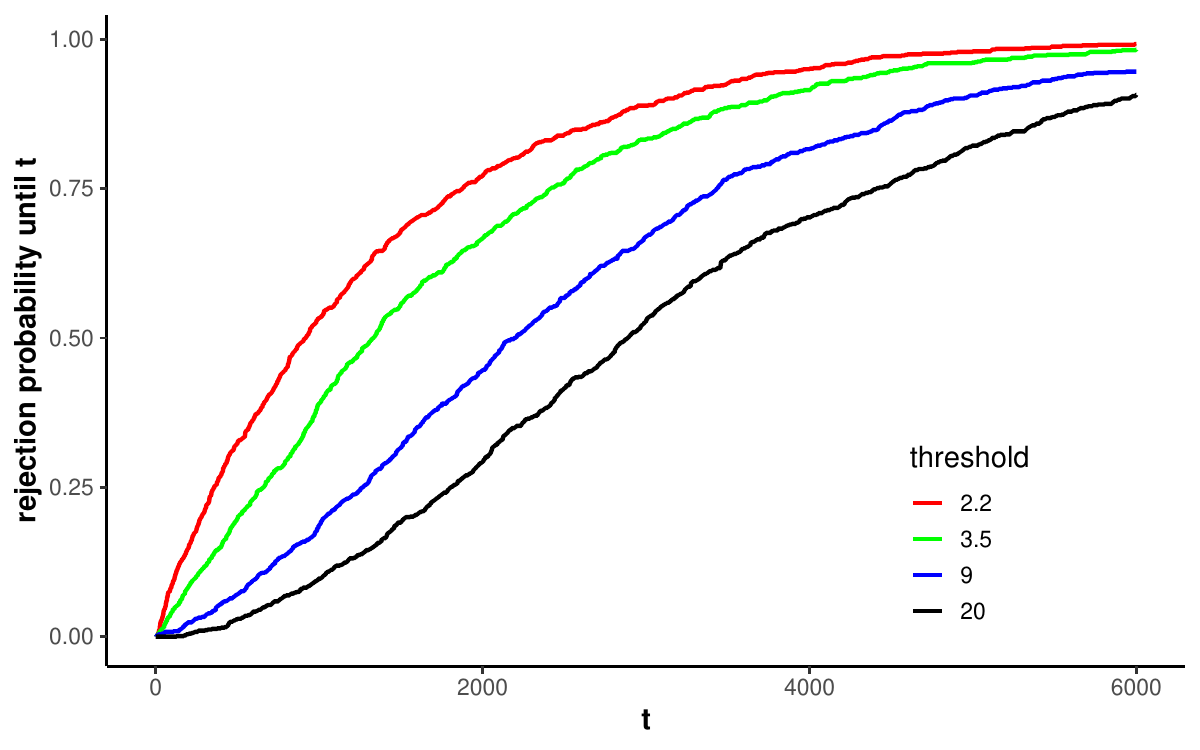}
  \caption[Cumulative rejection probabilities for forecasts 
    estimating $\ES_{0.99}$ and $\VaR_{0.99}$ under the assumption of normality for t-distributed RVs]{Cumulative rejection probabilities for forecasts 
    estimating $\ES_{0.99}$ and $\VaR_{0.99}$ under the assumption of normality for $t$-distributed random variables with $10$ degrees of freedom. All results are based on $1{,}000$ Monte Carlo simulation runs for different thresholds.}
  \label{fig_t_assumed_norm}
\end{figure}

These results suggest that the minimum required sample sizes are substantially higher when backtesting the expected shortfall at level $0.95$ than when testing at level $0.99$. This is likely due to the true expected shortfall of an accurately scaled normally distributed random variable not being much higher than the true expected shortfall of the respective t-distributed random variable. This is not the case when considering the $99\%$ expected shortfall; see Table \ref{TabelleESVaRtNormVergleich}. This table also explains the unexpected result of the required sample sizes being smaller for $\nu=5$ than for $\nu=3$ when considering the expected shortfall and value-at-risk at level $0.95$, since the true expected shortfall is about $8.65 \%$ percent higher than the estimated value for $\nu=5$, which is slightly larger than the ratio for $\nu=3$ ($8.40\%$). Apart from that, we observe (in general) a monotone relationship between the required sample sizes and $\nu$, $\tau$, and the targeted power. 

\setlength{\tabcolsep}{1pt}
\begin{table}[!htb]
    \centering
    \resizebox{\textwidth}{!}{
    \begin{tabular}{c|c c c c c c c c}
         $\nu$   & $\ES_{0.99}(T)$ & $\ES_{0.99}(N)$ & $\VaR_{0.99}(T)$ & $\VaR_{0.99}(N)$ & $\ES_{0.95}(T)$ & $\ES_{0.95}(N)$ & $\VaR_{0.95}(T)$ & $\VaR_{0.95}(N)$ \\ 
      \hline 
      3 & 7.00 & 4.62 & 4.54 & 4.03 & 3.87 & 3.57 & 2.35 & 2.85 \\ 
      5 & 4.45 & 3.44 & 3.36 & 3.00 & 2.89 & 2.66 & 2.02 & 2.12 \\ 
      7 & 3.77 & 3.15 & 3.00 & 2.75 & 2.59 & 2.44 & 1.89 & 1.95 \\ 
      10 & 3.36 & 2.98 & 2.76 & 2.60 & 2.41 & 2.31 & 1.81 & 1.84
    \end{tabular}
    }
    \caption[Expected shortfalls and value-at-risks for t-distributed RVs and normally distributed RVs with variance $\nu/(\nu-2)$]{Expected shortfalls and values-at-risks for a Student's $t$-distributed random variable $T$ with different  degrees of freedom $\nu$, and for a normally distributed random variable $N$ with variance $\nu/(\nu-2)$}
    \label{TabelleESVaRtNormVergleich}
\end{table}
\setlength{\tabcolsep}{4pt}

We have also run the same experiments as reported in Section \ref{sampsizesimulnormal}, i.\ e., underestimating the expected shortfall and the value-at-risk by a constant factor, as well as only underestimating the expected shortfall, for a sequence of $t$-distributed random variables with five degrees of freedom, in order to verify that our conclusions and approximate formulas for minimal required sample sizes  from Section \ref{sampsizesimulnormal} are comparable when considering other distributions. The corresponding results are provided in Tables \ref{TabelleTSampleSizeFactor} and \ref{TabelleTSampleSizeVaRcorrect99} as well as in Figure \ref{fig_T_VaR_correct}.

\setlength{\tabcolsep}{1pt}
\begin{table}[!htb]
    \centering
    \resizebox{\textwidth}{!}{
    \begin{tabular}{c|c c c c|c c c c|c c c c}
    Power& \multicolumn{4}{c|}{70 \%} &  \multicolumn{4}{c|}{80 \%} & \multicolumn{4}{c}{90 \%}  \\
    \hline 
         \diagbox{d}{$\tau$} &2.2 & 3.5 & 9 & 20 & 2.2 & 3.5 & 9 & 20 & 2.2 & 3.5 & 9 & 20 \\ 
               \hline
        \multicolumn{13}{c}{\textbf{p=0.95}} \\
        0.5 & 32 & 41 & 60 & 78 & 41 & 50 & 71 & 90 & 52 & 64 & 87 & 109 \\ 
         0.6 & 58 & 77 & 112 & 147 & 74 & 93 & 137 & 171 & 104 & 125 & 174 & 204 \\ 
         0.7 & 107 & 137 & 211 & 264 & 137 & 175 & 251 & 318 & 188 & 242 & 326 & 399 \\ 
         0.8 & 296 & 385 & 565 & 711 & 387 & 491 & 693 & 842 & 533 & 667 & 901 & 1063 \\ 
         0.9 & 1250 & 1831 & 2738 & 3400 & 1815 & 2446 & 3357 & 4118 & 2713 & 3303 & 4342 & 5217 \\
        \multicolumn{13}{c}{\textbf{p=0.99}} \\
         0.5 & 64 & 82 & 123 & 153 & 80 & 100 & 143 & 176 & 109 & 134 & 186 & 222 \\ 
         0.6 & 124 & 163 & 238 & 311 & 165 & 203 & 288 & 360 & 217 & 275 & 366 & 432 \\ 
         0.7 & 289 & 372 & 560 & 717 & 375 & 465 & 678 & 868 & 502 & 629 & 852 & 1060 \\ 
         0.8 & 779 & 1028 & 1605 & 2008 & 1018 & 1371 & 1929 & 2343 & 1453 & 1758 & 2450 & 2987 \\ 
         0.9 & 4035 & 5495 & >6000 & >6000 & 5540 & >6000 & >6000 & >6000 & >6000 & >6000 & >6000 & >6000
    \end{tabular}
    }
    \caption[Minimum sample sizes needed to achieve certain powers for forecasts 
    underestimating $\ES_{p}$ and $\VaR_{p}$ by a constant factor for t-distributed RVs]{Minimum sample sizes needed to achieve certain powers while using specific thresholds $\tau$ for forecasts 
    underestimating $\ES_{p}$ and $\VaR_{p}$ for levels $p \in \{0.95, 0.99\}$ by a constant factor $d$, for Student's $t$-distributed random variables with five degrees of freedom. All results are based on $1{,}000$ Monte Carlo simulations each and a maximum sample size of $6{,}000$.}
    \label{TabelleTSampleSizeFactor}
\end{table}
\setlength{\tabcolsep}{4pt}

\setlength{\tabcolsep}{1pt}
\begin{table}[!htb]
    \centering
    \resizebox{\textwidth}{!}{
    \begin{tabular}{c|c c c c|c c c c|c c c c}
    Power& \multicolumn{4}{c|}{70 \%} &  \multicolumn{4}{c|}{80 \%} & \multicolumn{4}{c}{90 \%}  \\
    \hline 
         \diagbox{d}{$\tau$} &2.2 & 3.5 & 9 & 20 & 2.2 & 3.5 & 9 & 20 & 2.2 & 3.5 & 9 & 20 \\ 
               \hline
        \multicolumn{13}{c}{\textbf{p=0.95}} \\
        0.2 & 157 & 201 & 292 & 370 & 200 & 248 & 346 & 433 & 265 & 319 & 433 & 531 \\ 
         0.3 & 201 & 269 & 407 & 519 & 260 & 351 & 496 & 623 & 386 & 464 & 640 & 811 \\ 
         0.4 & 341 & 443 & 654 & 820 & 438 & 547 & 790 & 966 & 585 & 749 & 1009 & 1255 \\ 
         0.5 & 477 & 637 & 1019 & 1279 & 649 & 867 & 1268 & 1525 & 920 & 1124 & 1565 & 1879 \\ 
         0.6 & 777 & 1063 & 1648 & 2063 & 1070 & 1382 & 2020 & 2498 & 1614 & 1950 & 2637 & 3120 \\ 
        \multicolumn{13}{c}{\textbf{p=0.99}} \\
         0.2 & 751 & 993 & 1444 & 1881 & 950 & 1209 & 1736 & 2169 & 1236 & 1580 & 2121 & 2648 \\ 
         0.3 & 1073 & 1479 & 2261 & 2817 & 1440 & 1924 & 2864 & 3311 & 1985 & 2433 & 3330 & 4096 \\ 
         0.4 & 1691 & 2240 & 3348 & 4154 & 2199 & 2807 & 3964 & 4818 & 3024 & 3667 & 4949 & 5993 \\ 
         0.5 & 2707 & 3488 & 5289 & >6000 & 3455 & 4406 & >6000 & >6000 & 4972 & >6000 & >6000 & >6000 \\ 
         0.6 & 4597 & >6000 & >6000 & >6000 & 5909 & >6000 & >6000 & >6000 & >6000 & >6000 & >6000 & >6000
    \end{tabular}
    }
    \caption[Minimum sample sizes needed to achieve certain powers for forecasts 
    estimating $\VaR_{p}$ correctly and underestimating $\ES_{p}-\VaR_{p}$ by a constant factor for t-distributed RVs]{Minimum sample sizes needed to achieve certain powers while using specific thresholds $\tau$ for forecasts 
    estimating $\VaR_{p}$ correctly and underestimating $\ES_{p}-\VaR_{p}$ for levels $p \in \{0.95,0.99\}$ by a constant factor $d$, for Student's $t$-distributed random variables with five degrees of freedom. All results are based on 
    $1{,}000$ Monte Carlo simulation runs each and a maximum sample size of $6{,}000$.}
    \label{TabelleTSampleSizeVaRcorrect99}
\end{table}
\setlength{\tabcolsep}{4pt}

\begin{figure}[!htb]
  \centering
  \includegraphics[width=0.8\textwidth]{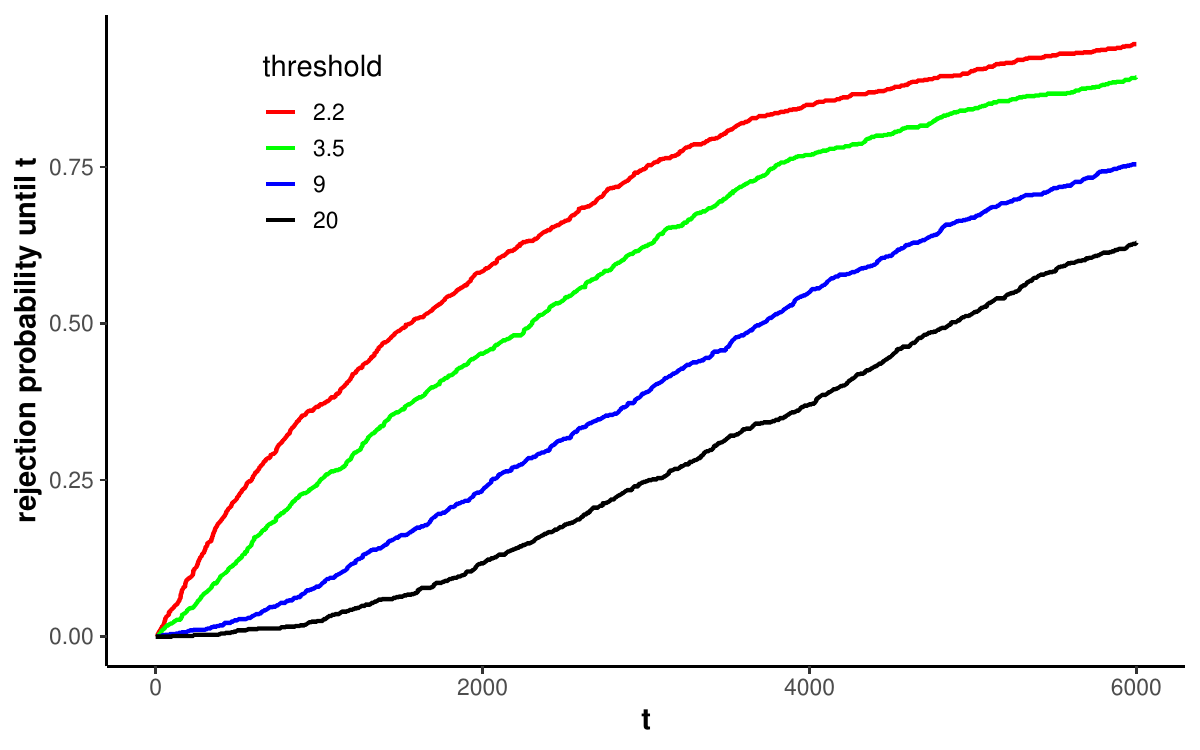}
  \caption[Cumulative rejection probabilities for forecasts estimating $\VaR_{0.99}$ correctly and underestimating $\ES_{0.99}-\VaR_{0.99}$ by a constant factor for t-distributed RVs]{Cumulative rejection probabilities for forecasts estimating $\VaR_{0.99}$ correctly and underestimating $\ES_{0.99}-\VaR_{0.99}$ by a constant factor $0.5$, for Student's $t$-distributed random variables with five degrees of freedom. All results are based on $1{,}000$ Monte Carlo simulations for different significance thresholds.}
  \label{fig_T_VaR_correct}
\end{figure}

Comparing these results to the ones for normally distributed random variables, we notice that the required sample sizes are often considerably larger. In many of our simulations, especially when considering small underestimations by  a factor of $0.9$, we were not able to exceed the examined thresholds even after a considerably large sample size of $6{,}000$, which, if every time point represented a banking day, would represent a time span of about $24$ years. 
The results also suggest that the required sample sizes in dependence on $d$ grow faster than in the order of $1/(1-d)^2$. One thus needs to be careful when trying to extrapolate our results to even smaller underestimations.  

\subsection{Extreme value distributions}\label{sec62}
In this section, we investigate sample sizes needed to reject estimation models for the expected shortfall and the value-at-risk in cases where extreme value theory is employed to model the extreme value behavior of the loss variables in question. 
To this end, let $p_1\in (0,1)$ (in this section we will consider $p_1=0.95$) and $(L_t)_{t\in \N}$ be an i.i.d. sequence of loss variables with $\VaR_{p_1}(L_1)=0$ such that $L_1|L_1>0$ is generalized Pareto distributed for some scale parameter $\beta$ and shape parameter $\epsilon$. Then, for $p\in [p_1,1)$, the true values of $\VaR_p(L_1)$ and $\ES_p(L_1)$ can be calculated via the formulas (7.18) - (7.19) in \textcite{mcneil2005quantitative}. 

\subsubsection{Correctly specified excess distribution}
We start by considering, in analogy to Section \ref{sec61}, a prediction model that underestimates either both the expected shortfall and the value-at-risk by a given constant factor $d$, or the difference between the expected shortfall and the value-at-risk for levels $0.99$ and $0.999$. (Here, we consider larger levels $p$ than in the previous section, because approximations based on extreme value theory are less biased when exploring more extreme tails.) 

First, we simulated the case of a forecaster underestimating both the expected shortfall and value-at-risk by a constant factor $d\in (0,1)$. Here, we considered the factors $d\in \{0.4,0.5,0.6,0.7,0.8\}$ for the level $p=0.99$ and $d\in \{0.4,0.5,0.6,0.7\}$ for the level $0.999$, for all investigated shape parameters $\epsilon$. Our results are summarized in Table \ref{TabelleGPDSampleSizeFactor}. These simulation results reveal that in most of our observed scenarios a large sample size is needed in order to detect even quite large underestimations of the expected shortfall and the value-at-risk. This is especially the case when considering these risk measures at level $0.999$. This effect is again due to the loss variables needing to exceed the projected value-at-risk in order for the e-statistic to be non-zero. Since such events are rare when considering levels of this magnitude, the e-processes grow only slowly. We thus recommend to consider even larger sample sizes as we did in this work (we only tested sample sizes up to $6{,}000$) when backtesting the expected shortfall and the value-at-risk at levels larger than $0.99$. 
We also noticed that the required sample sizes were considerably smaller when considering smaller shapes, which correspond to loss variables with less heavy tails. This is especially evident when backtesting $\ES_{0.999}$ and $\VaR_{0.999}$. Here, when letting $\epsilon=0.2$, even a quite large underestimation of the true risk measures by a factor of $0.6$ requires sample sizes of about $3{,}000-5{,}000$ observations in order to detect the underestimation for most investigated thresholds and powers. In the worst case ($\tau=20$ and 90\% power), the required sample size is even larger than the maximum considered value of $6{,}000$. However, when letting $\epsilon=-0.2$, the required sample sizes are considerably smaller.
\setlength{\tabcolsep}{1pt}
\begin{table}[!htb]
    \centering
    \resizebox{\textwidth}{!}{
    \begin{tabular}{c|c c c c|c c c c|c c c c}
    Power& \multicolumn{4}{c|}{70 \%} &  \multicolumn{4}{c|}{80 \%} & \multicolumn{4}{c}{90 \%}  \\
    \hline 
         \diagbox{d, $\epsilon$}{$\tau$} &2.2 & 3.5 & 9 & 20 & 2.2 & 3.5 & 9 & 20 & 2.2 & 3.5 & 9 & 20 \\ 
               \hline
        \multicolumn{13}{c}{\textbf{p=0.99}} \\
         0.4, 0.2 & 282 & 360 & 525 & 688 & 358 & 448 & 656 & 792 & 486 & 616 & 826 & 974 \\ 
         0.5, 0.2 & 462 & 627 & 855 & 1100 & 582 & 751 & 1059 & 1314 & 765 & 941 & 1296 & 1582 \\ 
         0.6, 0.2 & 779 & 999 & 1479 & 1801 & 1005 & 1258 & 1747 & 2090 & 1371 & 1646 & 2170 & 2609 \\ 
         0.7, 0.2 & 1421 & 1965 & 2855 & 3458 & 1870 & 2409 & 3333 & 4080 & 2754 & 3267 & 4200 & 4932 \\ 
         0.8, 0.2 & 3780 & 4932 & >6000 & >6000 & 4887 & >6000 & >6000 & >6000 & >6000 & >6000 & >6000 & >6000 \\ 
         0.4, 0 & 238 & 309 & 452 & 575 & 305 & 382 & 537 & 661 & 398 & 496 & 664 & 809 \\ 
         0.5, 0 & 345 & 458 & 672 & 845 & 442 & 565 & 808 & 997 & 574 & 722 & 994 & 1211 \\ 
         0.6, 0 & 555 & 743 & 1057 & 1368 & 696 & 900 & 1249 & 1619 & 954 & 1189 & 1607 & 2033 \\ 
         0.7, 0 & 1025 & 1380 & 2043 & 2537 & 1376 & 1756 & 2463 & 3006 & 1940 & 2347 & 3202 & 3759 \\ 
         0.8, 0 & 2483 & 3224 & 4732 & 5872 & 3342 & 4286 & 5761 & >6000 & 4818 & 5573 & >6000 & >6000 \\ 
         0.4, -0.2 & 208 & 267 & 368 & 475 & 259 & 313 & 449 & 553 & 330 & 401 & 559 & 672 \\ 
         0.5, -0.2 & 290 & 398 & 564 & 713 & 386 & 481 & 678 & 847 & 507 & 634 & 866 & 1045 \\ 
         0.6, -0.2 & 470 & 611 & 876 & 1122 & 590 & 749 & 1050 & 1297 & 767 & 950 & 1297 & 1556 \\ 
         0.7, -0.2 & 865 & 1091 & 1592 & 1989 & 1114 & 1414 & 1871 & 2382 & 1536 & 1853 & 2543 & 3051 \\ 
         0.8, -0.2 & 1676 & 2369 & 3392 & 4352 & 2344 & 3036 & 4226 & 5270 & 3282 & 4137 & 5483 & >6000 \\
        \multicolumn{13}{c}{\textbf{p=0.999}} \\
          0.4, 0.2 & 668 & 838 & 1192 & 1564 & 829 & 1044 & 1410 & 1828 & 1116 & 1358 & 1835 & 2279 \\ 
         0.5, 0.2 & 1193 & 1512 & 2228 & 2829 & 1498 & 1896 & 2687 & 3373 & 2069 & 2577 & 3319 & 4132 \\ 
         0.6, 0.2 & 2174 & 2792 & 3994 & 5089 & 2810 & 3455 & 4677 & 5873 & 3614 & 4455 & 5837 & >6000 \\ 
         0.7, 0.2 & 4463 & >6000 & >6000 & >6000 & 5830 & >6000 & >6000 & >6000 & >6000 & >6000 & >6000 & >6000 \\ 
         0.4, 0 & 377 & 491 & 708 & 890 & 475 & 599 & 832 & 1044 & 618 & 767 & 1036 & 1290 \\ 
         0.5, 0 & 650 & 845 & 1211 & 1564 & 837 & 1025 & 1459 & 1845 & 1111 & 1360 & 1798 & 2185 \\ 
         0.6, 0 & 1070 & 1409 & 2148 & 2736 & 1345 & 1786 & 2510 & 3196 & 1898 & 2266 & 3192 & 3925 \\ 
         0.7, 0 & 2183 & 3015 & 4526 & 5701 & 2885 & 3686 & 5304 & >6000 & 3895 & 4851 & >6000 & >6000 \\
         0.4, -0.2 & 257 & 331 & 477 & 626 & 316 & 407 & 578 & 708 & 433 & 528 & 689 & 867 \\ 
         0.5, -0.2 & 408 & 517 & 733 & 948 & 496 & 618 & 855 & 1069 & 626 & 752 & 1055 & 1317 \\ 
         0.6, -0.2 & 661 & 832 & 1251 & 1592 & 806 & 1059 & 1469 & 1896 & 1156 & 1363 & 1834 & 2259 \\ 
         0.7, -0.2 & 1184 & 1513 & 2207 & 2948 & 1547 & 1983 & 2721 & 3388 & 2076 & 2551 & 3397 & 4124 
    \end{tabular}
    }
    \caption[Minimum sample sizes needed to achieve certain powers for forecasts 
    underestimating $\ES_{p}$ and $\VaR_{p}$ by a constant factor for RVs with GPD distributed excesses]{Minimum sample sizes needed to achieve certain powers when using specific thresholds $\tau$ for forecasts 
    underestimating $\ES_{p}$ and $\VaR_{p}$ for levels $p\in \{0.99,0.999\}$ by a constant factor $d$, for random variables with generalized Pareto (GPD) excess distributions with scale $1$ and different shapes $\epsilon$. All results are based on $1{,}000$ Monte Carlo simulation runs each and a maximum sample size of $6{,}000$.}
    \label{TabelleGPDSampleSizeFactor}
\end{table}
\setlength{\tabcolsep}{4pt}
Second, we investigated predictions where the value-at-risk is estimated correctly, but the difference between the expected shortfall and the value-at-risk is underestimated by a constant factor. For this, we considered the factors $d\in \{0.2,0.3,0.4,0.5,0.6\}$ for the level $p=0.99$ and $d\in \{0.25,0.5,0.1\}$ for the level $p=0.999$, for all considered shape parameters $\epsilon$. Our results are summarized in Table \ref{TabelleGPDSampleSizeVaRCorrect}  as well as in Figure \ref{fig_GPD_VaR_correct}. In this scenario, rejections were vastly more difficult than in previously considered experiments, in particular when backtesting at level $0.999$. This is likely due to the value-at-risk not being underestimated leading to a mean amount of six loss variables exceeding the predicted (true) value-at-risk for the maximum considered sample size of $6{,}000$ and the level $0.999$. Since our e-statistic is only non-zero for these exceedance events, our e-process can only increase in the presence of such events, which explains the high sample sizes required to achieve certain powers. In the future, we recommend repeating our experiments for higher maximum sample sizes than those considered here, especially if one is interested in detecting underestimations of smaller magnitude, since we were only able to make meaningful statements for considerably large underestimations. 
For the level $0.99$ we noticed again that rejections were more difficult for larger shapes, which corresponds again to loss variables with heavier right tails, although this is less pronounced as for predictions that also underestimate the value-at-risk.
\setlength{\tabcolsep}{1pt}
\begin{table}[!htb]
    \centering
    \resizebox{\textwidth}{!}{
    \begin{tabular}{c|c c c c|c c c c|c c c c}
    Power& \multicolumn{4}{c|}{70 \%} &  \multicolumn{4}{c|}{80 \%} & \multicolumn{4}{c}{90 \%}  \\
    \hline 
         \diagbox{d, $\epsilon$}{$\tau$} &2.2 & 3.5 & 9 & 20 & 2.2 & 3.5 & 9 & 20 & 2.2 & 3.5 & 9 & 20 \\ 
               \hline
        \multicolumn{13}{c}{\textbf{p=0.99}} \\
         0.2, 0.2 & 827 & 1040 & 1542 & 2049 & 1021 & 1268 & 1871 & 2348 & 1305 & 1637 & 2311 & 2827 \\ 
         0.3, 0.2 & 1227 & 1644 & 2360 & 2950 & 1632 & 2008 & 2819 & 3486 & 2140 & 2651 & 3458 & 4218 \\ 
         0.4, 0.2 & 1812 & 2395 & 3570 & 4419 & 2368 & 3105 & 4200 & 5156 & 3378 & 4061 & 5343 & >6000 \\ 
         0.5, 0.2 & 2710 & 3678 & 5371 & >6000 & 3756 & 4754 & >6000 & >6000 & 4923 & >6000 & >6000 & >6000 \\ 
         0.6, 0.2 & 4623 & >6000 & >6000 & >6000 & >6000 & >6000 & >6000 & >6000 & >6000 & >6000 & >6000 & >6000 \\ 
         0.2, 0 & 701 & 942 & 1385 & 1739 & 899 & 1137 & 1630 & 2050 & 1213 & 1486 & 2002 & 2475 \\ 
         0.3, 0 & 1048 & 1342 & 2030 & 2525 & 1348 & 1700 & 2403 & 2915 & 1956 & 2201 & 3083 & 3716 \\ 
         0.4, 0 & 1531 & 1976 & 3009 & 3758 & 1962 & 2545 & 3606 & 4413 & 2727 & 3402 & 4633 & 5595 \\ 
         0.5, 0 & 2269 & 2935 & 4329 & 5413 & 2884 & 3826 & 5278 & >6000 & 3911 & 5199 & >6000 & >6000 \\ 
         0.6, 0 & 3606 & 5024 & >6000 & >6000 & 5040 & >6000  & >6000  & >6000 & >6000  & >6000  & >6000 & >6000 \\ 
         0.2, -0.2 & 639 & 885 & 1259 & 1624 & 850 & 1077 & 1543 & 1880 & 1103 & 1418 & 1911 & 2300 \\ 
         0.3, -0.2 & 873 & 1202 & 1773 & 2231 & 1167 & 1510 & 2164 & 2596 & 1530 & 1983 & 2693 & 3187 \\ 
         0.4, -0.2 & 1265 & 1691 & 2483 & 3204 & 1661 & 2183 & 2991 & 3780 & 2401 & 2916 & 3909 & 4698 \\ 
         0.5, -0.2 & 2010 & 2667 & 3998 & 5081 & 2744 & 3487 & 5056 & >6000 & 3996 & 4986 & >6000 & >6000 \\ 
         0.6, -0.2 & 3419 & 4499 & >6000 & >6000 & 4432 & 5930 & >6000 & >6000 & >6000 & >6000 & >6000 & >6000 \\ 
        \multicolumn{13}{c}{\textbf{p=0.999}} \\
         0.025, 0.2 & 3917 & 4952 & >6000 & >6000 & 4770 & >6000 & >6000 & >6000 & >6000 & >6000 & >6000 & >6000 \\ 
         0.05, 0.2 & 4601 & 5651 & >6000 & >6000 & 5789 & >6000 & >6000 & >6000 & >6000 & >6000 & >6000 & >6000 \\ 
         0.1, 0.2 & 5487 & >6000 & >6000 & >6000 & >6000 & >6000 & >6000 & >6000 & >6000 & >6000 & >6000 & >6000 \\  
         0.025, 0 & 3823 & 4839 & >6000 & >6000 & 4673 & 5947 & >6000 & >6000 & >6000 & >6000 & >6000 & >6000 \\ 
         0.05, 0 & 4345 & 5616 & >6000 & >6000 & 5485 & >6000 & >6000 & >6000 & >6000 & >6000 & >6000 & >6000 \\ 
         0.1, 0 & 4979 & >6000 & >6000 & >6000 & >6000 & >6000 & >6000 & >6000 & >6000 & >6000 & >6000 & >6000 \\ 
         0.025, -0.2 & 3792 & 4933 & >6000 & >6000 & 4744 & 5985 & >6000 & >6000 & >6000 & >6000 & >6000 & >6000 \\ 
         0.05, -0.2 & 4151 & 5200 & >6000 & >6000 & 5053 & >6000 & >6000 & >6000 & >6000 & >6000 & >6000 & >6000 \\ 
         0.1, -0.2 & 4780 & >6000 & >6000 & >6000 & 5921 & >6000 & >6000 & >6000 & >6000 & >6000 & >6000 & >6000
    \end{tabular}
    }
    \caption[Minimum sample sizes needed to achieve certain powers for forecasts estimating $\VaR_{p}$ correctly and underestimating $\ES_{p}-\VaR_{p}$ by a constant factor for RVs with GPD distributed excesses]{Minimum sample sizes needed to achieve certain powers while using specific thresholds $\tau$ for forecasts estimating $\VaR_{p}$ correctly and underestimating $\ES_{p}-\VaR_{p}$ for levels $p \in \{0.99,0.999\}$ by a constant factor $d$, for random variables with generelized Pareto (GPD) excess distributions with scale $1$ and different shapes $\epsilon$. All results are based on $1{,}000$ Monte Carlo simulation runs each and a maximum sample size of $6{,}000$.}
    \label{TabelleGPDSampleSizeVaRCorrect}
\end{table}
\setlength{\tabcolsep}{4pt}

\begin{figure}[!htb]
  \centering
  \includegraphics[width=0.8\textwidth]{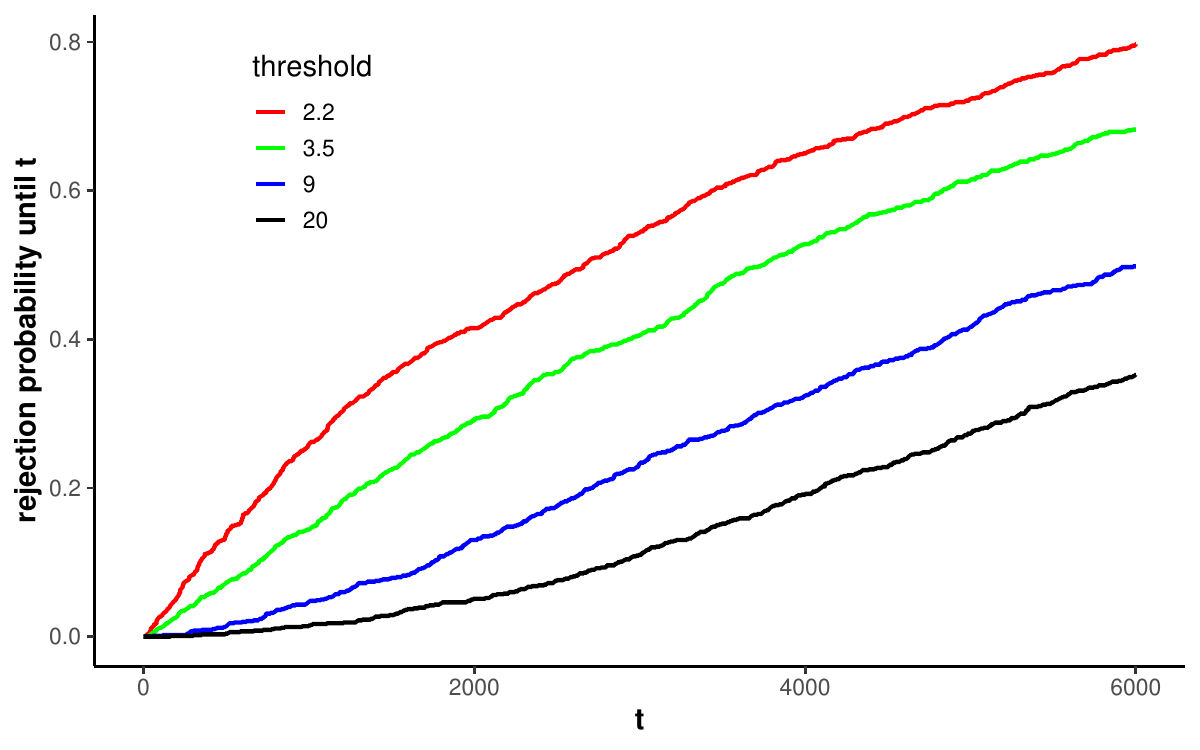}
  \caption[Cumulative rejection probabilities for forecasts estimating $\VaR_{0.99}$ correctly and underestimating $\ES_{0.99}-\VaR_{0.99}$ by a constant factor for RVs with GPD distributed excesses]{Cumulative rejection probabilities for forecasts estimating $\VaR_{0.99}$ correctly and underestimating $\ES_{0.99}-\VaR_{0.99}$ by a constant factor $0.6$, for random variables with generalized Pareto 
  (GPD)-distributed excesses with scale $1$ and shape $0.2$. All results are based on $1{,}000$ Monte Carlo simulations for different significance thresholds.}
  \label{fig_GPD_VaR_correct}
\end{figure}

\subsubsection{Misspecified excess distribution}
Lastly, we investigated forecasts misspecifying the extreme value behaviour of the involved random variables. For this, we let $(L_t)_{t\in \N}$ be an i.i.d. sequence of loss variables with $\VaR_{0.95}(L_1)=0$ such that $L_1|L_1>0$ is generalized Pareto distributed with scale parameter $1$ and shape parameter $0.2$. It is assumed that the forecaster misspecifies these parameters as $\hat{\beta}$ and $\hat{\epsilon}$ and calculates the predicted value-at-risk and expected shortfall according to the formulas (7.18) - (7.19) in \textcite{mcneil2005quantitative} using the misspecified parameters. We then backtest these predictions using the discussed backtesting procedure. Our results are summarized in Table \ref{TabelleGPDMisspecification99} and in Figure \ref{fig_GPD_misspecified}. Our results suggest that the backtesting method is more likely to detect misspecifications of the shape parameter $\epsilon$ than to detect misspecifications of the scale parameter $\beta$. This is especially the case when considering backtesting at level $0.999$. Here, a misspecification of the shape parameter as $\hat{\epsilon}=0.0$ with the scale parameter being estimated correctly can be detected quite reliably, while single underestimations of the scale parameter $\beta$ need to be quite severe in order to be detectable.
\setlength{\tabcolsep}{1pt}
\begin{table}[!htb]
    \centering
    \resizebox{\textwidth}{!}{
    \begin{tabular}{c|c c c c|c c c c|c c c c}
    Power& \multicolumn{4}{c|}{70 \%} &  \multicolumn{4}{c|}{80 \%} & \multicolumn{4}{c}{90 \%}  \\
    \hline 
         \diagbox{$(\hat{\beta},\hat{\epsilon})$}{$\tau$} &2.2 & 3.5 & 9 & 20 & 2.2 & 3.5 & 9 & 20 & 2.2 & 3.5 & 9 & 20 \\ 
         \hline
         \multicolumn{13}{c}{\textbf{p=0.99}} \\
         (0.5,0.2) & 493 & 653 & 980 & 1295 & 651 & 809 & 1157 & 1515 & 875 & 1059 & 1503 & 1866 \\
         (0.6,0.2) & 846 & 1130 & 1692 & 2166 & 1172 & 1477 & 2053 & 2587 & 1558 & 1944 & 2653 & 3233 \\
         (0.7,0.2) & 1764 & 2278 & 3444 & 4365 & 2300 & 2865 & 4210 & 5184 & 3173 & 3766 & 5325 & >6000 \\
         (0.8,0.2) & 4595 & >6000 & >6000 & >6000 & >6000 & >6000 & >6000 & >6000 & >6000 & >6000 & >6000 & >6000 \\ 
         (0.7,0.0) & 552 & 703 & 1027 & 1321 & 711 & 874 & 1254 & 1566 & 974 & 1176 & 1590 & 1943 \\
         (0.8,0.0) & 778 & 1049 & 1551 & 1995 & 1004 & 1357 & 1885 & 2338 & 1418 & 1805 & 2356 & 2902 \\ 
         (0.9,0.0) & 1285 & 1717 & 2502 & 3082 & 1699 & 2109 & 3010 & 3753 & 2357 & 2922 & 3871 & 4553 \\ 
         (1.0,0.0) & 2084 & 2836 & 4196 & 5311 & 2786 & 3563 & 5101 & >6000 & 3843 & 4742 & >6000 & >6000 \\
         (1.0,-0.2) & 708 & 919 & 1305 & 1681 & 901 & 1112 & 1539 & 1935 & 1156 & 1422 & 1936 & 2411 \\
         (1.1,-0.2) & 906 & 1186 & 1736 & 2155 & 1140 & 1446 & 2076 & 2527 & 1571 & 1943 & 2651 & 3220 \\
         (1.2,-0.2) & 1230 & 1638 & 2372 & 3018 & 1577 & 2021 & 2829 & 3651 & 2252 & 2753 & 3705 & 4632 \\
         (1.3,-0.2) & 1763 & 2376 & 3587 & 4532 & 2283 & 3153 & 4362 & 5385 & 3357 & 4193 & 5617 & >6000 \\ 
    \multicolumn{13}{c}{\textbf{p=0.999}} \\
    (0.3,0.2) & 409 & 515 & 732 & 973 & 503 & 628 & 903 & 1177 & 666 & 829 & 1123 & 1393 \\
         (0.4,0.2) & 687 & 880 & 1292 & 1688 & 885 & 1108 & 1533 & 1973 & 1130 & 1391 & 1894 & 2415 \\
         (0.5,0.2) & 1300 & 1655 & 2513 & 3196 & 1589 & 2056 & 2993 & 3781 & 2181 & 2666 & 3710 & 4678 \\
         (0.6,0.2) & 2451 & 3294 & 4946 & >6000 & 3295 & 4249 & 5894 & >6000 & 4497 & 5489 & >6000 & >6000 \\
         (0.7,0.0) & 862 & 1098 & 1593 & 2011 & 1101 & 1363 & 1897 & 2396 & 1465 & 1813 & 2357 & 2981 \\
         (0.8,0.0) & 1155 & 1464 & 2185 & 2899 & 1474 & 1769 & 2682 & 3345 & 1929 & 2372 & 3381 & 4122 \\
         (0.9,0.0) & 1716 & 2157 & 3240 & 4193 & 2150 & 2747 & 3791 & 4865 & 3023 & 3687 & 4815 & 5912 \\
         (1.0,0.0) & 2482 & 3223 & 4795 & >6000 & 3203 & 4064 & 5749 & >6000 & 4370 & 5403 & >6000 & >6000 \\
         (1.0,-0.2) & 722 & 930 & 1316 & 1724 & 909 & 1131 & 1608 & 2016 & 1206 & 1497 & 1963 & 2450 \\
         (1.1,-0.2) & 919 & 1113 & 1654 & 2118 & 1154 & 1371 & 1956 & 2481 & 1523 & 1788 & 2509 & 3037 \\
         (1.2,-0.2) & 1157 & 1443 & 2062 & 2734 & 1467 & 1775 & 2490 & 3238 & 1927 & 2281 & 3122 & 3958 \\
         (1.3,-0.2) & 1358 & 1780 & 2632 & 3416 & 1753 & 2161 & 3149 & 3905 & 2281 & 2881 & 3922 & 4726
    \end{tabular}
    }
    \caption[Minimum sample sizes needed to achieve certain powers for forecasts of $\ES_{p}$ and $\VaR_{p}$ while misspecifying the excess distributions of RVs]{Minimum sample sizes needed to achieve certain powers when using specific thresholds $\tau$ for forecasts of $\ES_{p}$ and $\VaR_{p}$ for levels $p \in \{0.99,0.999\}$ while misspecifying the excess distributions of random variables with GPD($1,0.2$)-distributed excesses as GPD($\hat{\beta},\hat{\epsilon}$). All results are based on $1{,}000$ Monte Carlo simulation runs each and a maximum sample size of $6{,}000$.}
    \label{TabelleGPDMisspecification99}
\end{table}
\setlength{\tabcolsep}{4pt}

\begin{figure}[!htb]
  \centering
  \includegraphics[width=0.8\textwidth]{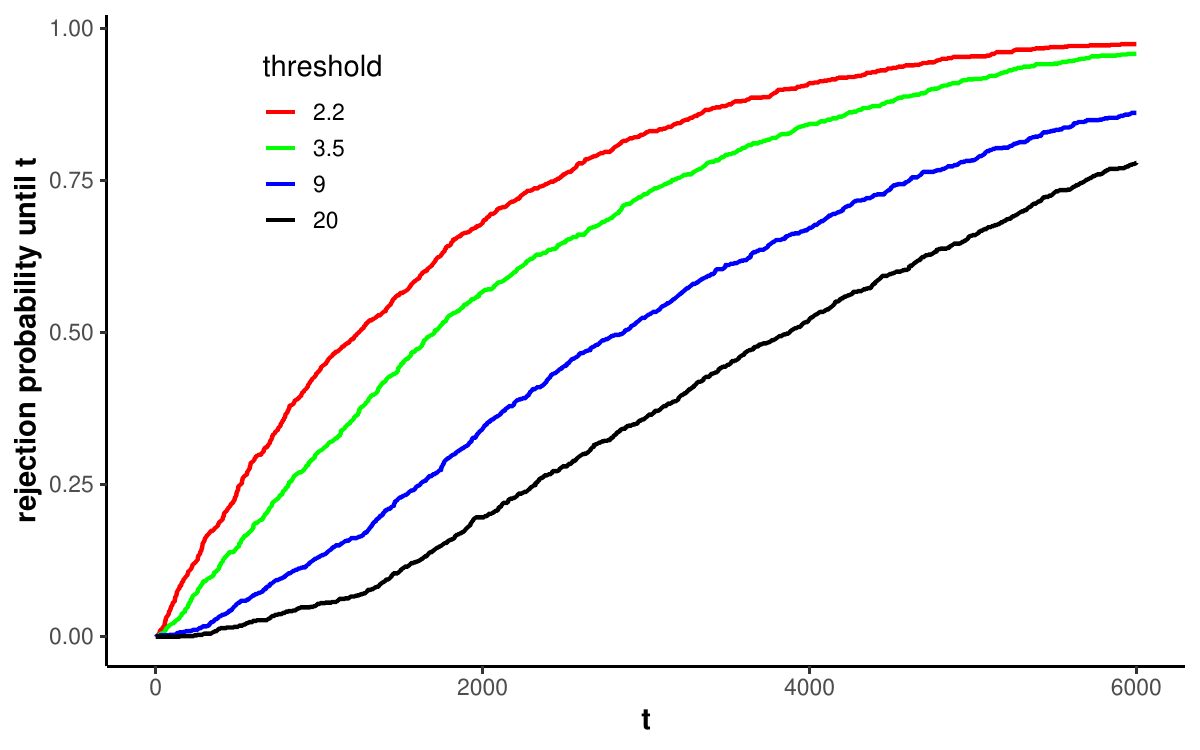}
  \caption[Cumulative rejection probabilities for forecasts estimating $\ES_{0.99}$ and $\VaR_{0.99}$ assuming GPD($1,0.0$)-distributed excesses for random variables with GPD$(1,0.2)$-distributed excesses]{Cumulative rejection probabilities for forecasts estimating $\ES_{0.99}$ and $\VaR_{0.99}$ assuming GPD($1,0.0$)-distributed excesses for random variables with GPD$(1,0.2)$-distributed excesses. All results are based on $1{,}000$ Monte Carlo simulation runs for different thresholds.}
  \label{fig_GPD_misspecified}
\end{figure}

\subsection{Misspecified temporal dynamics}\label{sec63}
In this section, we investigate sample sizes needed in order to detect misspefications of the temporal dynamics underlying the considered sequence of loss variables. For this purpose, we consider particular mixed ARMA-GARCH time series, which are defined as follows.
\begin{model}\label{arma-garch}
We assume a sequence of loss variables $(L_t)_{t\in \N}$ following an AR(1)-GARCH(1,1) model, meaning that
\begin{align*}
  L_t&=\mu_t+\sigma_t Z_t, \quad \mu_t=\alpha'_0+\alpha_1' L_{t-1}, \\ 
  \sigma^2_t&=\alpha_0+\alpha_1 (L_{t-1}-\mu_{t-1})^2+\beta_1 \sigma^2_{t-1} \, ,
\end{align*}
for parameters $\alpha'_0,\alpha_1',\alpha_0,\alpha_1,\beta_1 \in \R$ (in our case, $\alpha'_0=-0.05,\alpha_1'=0.3,\alpha_0=0.01,\alpha_1=0.1$ and $\beta_1=0.85$) and an i.i.d. sequence of centered and normalized innovation variables $(Z_t)_{t\in \N}$.
\end{model}
Under Model \ref{arma-garch}, we consider misspecifications that lead to underestimations of the expected shortfall. Since the true values for the value-at-risk and the expected shortfall are mostly dependent on the distribution of the innovation variables, and we only scale and shift them by a time-specific standard deviation $\sigma_t$ and a time-specific mean $\mu_t$, we will not consider underestimation by a given amount or factor here. Instead, we consider (i) a misspecification of the innovation distribution, (ii) a misspecification of the time series model, and (iii) constant (in time) forecasts of the value-at-risk and the expected shortfall, while their actual values change over time.
\subsubsection{Misspecified innovation distribution}
First, we focus on a misspecification of the distribution of the innovation variables. For this, we have generated $500$ data points assuming a certain distribution of $(Z_t)_{t\in \N}$, and we simulated that the forecaster fits a time series assuming another class of distribution for these innovation variables. We then backtest the forecaster's predictions for $(\ES_p(L_t))_{t\in \N}$ and  $(\VaR_p(L_t))_{t\in \N}$ using our described backtesting procedure with the GREM betting process. We begin with forecaster assuming normally distributed innovations, while the true innovations are Student's $t$-distributed with $\nu=5$ degrees of freedom. Our corresponding results, based on $1,000$ Monte-Carlo simulation runs, are provided in Table \ref{TabelleTimeSeriesTNorm99}.

\setlength{\tabcolsep}{1pt}
\begin{table}[!htb]
    \centering
    \begin{tabular}{c|c c c c|c c c c|c c c c}
    Power& \multicolumn{4}{c|}{70 \%} &  \multicolumn{4}{c|}{80 \%} & \multicolumn{4}{c}{90 \%}  \\
    \hline 
         \diagbox{p}{$\tau$} &2.2 & 3.5 & 9 & 20 & 2.2 & 3.5 & 9 & 20 & 2.2 & 3.5 & 9 & 20 \\ 
               \hline
         0.99 & 577 & 818 & 1215 & 1552 & 814 & 1030 & 1477 & 1873 & 1106 & 1407 & 2014 & 2528 \\ 
      0.95 & 1300 & 1924 & 2956 & 3922 & 2119 & 2839 & 4456 & >6000 & 4679 & >6000 & >6000 & >6000
    \end{tabular}
    \caption[Minimum sample sizes needed to achieve certain powers for forecasts 
    estimating $\ES_{p}$ and $\VaR_{p}$ of an AR-GARCH time series with Student's $t$-distributed innovations under the assumption of normally distributed innovations]{Minimum sample sizes needed to achieve certain powers for forecasts 
    estimating $\ES_{p}$ and $\VaR_{p}$ of an AR(1)-GARCH(1,1) time series with Student's $t$-distributed innovations under the working assumption of normally distributed innovations, for different thresholds $\tau$ and levels $p$. All results are based on $1{,}000$ Monte Carlo simulation runs each and a maximum sample size of $6{,}000$.}
    \label{TabelleTimeSeriesTNorm99}
\end{table}
\setlength{\tabcolsep}{4pt}

%
The results in Table \ref{TabelleTimeSeriesTNorm99} reveal that detections are considerably easier when backtesting the expected shortfall at level $0.99$ than at level $0.95$. This again illustrates that the expected shortfall and the value-at-risk at the latter level for $t$-distributed random variables do not differ much from the corresponding values for normally distributed random variables (see also Table \ref{TabelleESVaRtNormVergleich}). Examining the results for $p=0.99$, it can be observed that a sample size of about $2{,}000$ leads to very high detection probabilities when considering the thresholds $2$ and $3.5$ and reasonable high detection probabilities when applying the thresholds $9$ and $20$. We therefore recommend this sample size to detect misspecifications of this type, in particular when applying large thresholds. If one considers a smaller threshold, a sample size of $1{,}000-1{,}500$ might also be considered. 

Next, we assume the innovation variables to be skewed $t$-distributed with five degrees of freedom and skewness parameter $1.5$; see, e.\ g., \textcite{reverse-stress-testing} for details regarding the skewed $t$-distribution and its applications in risk management. Under this specification, we have simulated the forecaster assuming the innovation variables to be normally distributed or $t$-distributed, respectively. Our corresponding results are provided in Table \ref{TabelleTimeSeriesSkewedTNorm99} and Figure \ref{fig_Time_Series_skewed_T_T}. 

\begin{table}[!htb]
    \centering
    \resizebox{\textwidth}{!}{
    \begin{tabular}{c|c c c c|c c c c|c c c c}
    Power& \multicolumn{4}{c|}{70 \%} &  \multicolumn{4}{c|}{80 \%} & \multicolumn{4}{c}{90 \%}  \\
    \hline 
         \diagbox{p}{$\tau$} &2.2 & 3.5 & 9 & 20 & 2.2 & 3.5 & 9 & 20 & 2.2 & 3.5 & 9 & 20 \\ 
               \hline
          \multicolumn{13}{c}{\textbf{normal distribution}} \\
         0.99 & 218 & 295 & 426 & 563 & 278 & 367 & 515 & 655 & 388 & 477 & 667 & 846 \\ 
      0.95 & 256 & 348 & 516 & 669 & 343 & 461 & 659 & 837 & 527 & 653 & 882 & 1085 \\ 
  \multicolumn{13}{c}{\textbf{(unskewed) t-distribution}} \\
  0.99 & 853 & 1161 & 2014 & 2617 & 1232 & 1703 & 2696 & 3525 & 2261 & 2802 & 4723 & >6000 \\ 
      0.95 & 319 & 462 & 712 & 902 & 434 & 607 & 869 & 1113 & 635 & 869 & 1263 & 1508  
    \end{tabular}
    }
    \caption[Minimum sample sizes needed to achieve certain powers for forecasts 
    estimating $\ES_{p}$ and $\VaR_{p}$ of an AR-GARCH time series with skewed t-distributed innovations under the assumption of normally or t-distributed innovations]{Minimum sample sizes needed to achieve certain powers for forecasts 
    estimating $\ES_{p}$ and $\VaR_{p}$ of an AR(1)-GARCH(1,1) time series with skewed t-distributed innovations under the working assumption of normally or (unskewed) $t$-distributed innovations, respectively, for different thresholds $\tau$ and levels $p$. All results are based on $1{,}000$ Monte Carlo simulation runs each and a maximum sample size of $6{,}000$.}
    \label{TabelleTimeSeriesSkewedTNorm99}
\end{table}

\begin{figure}[!htb]
  \centering
  \includegraphics[width=0.8\textwidth]{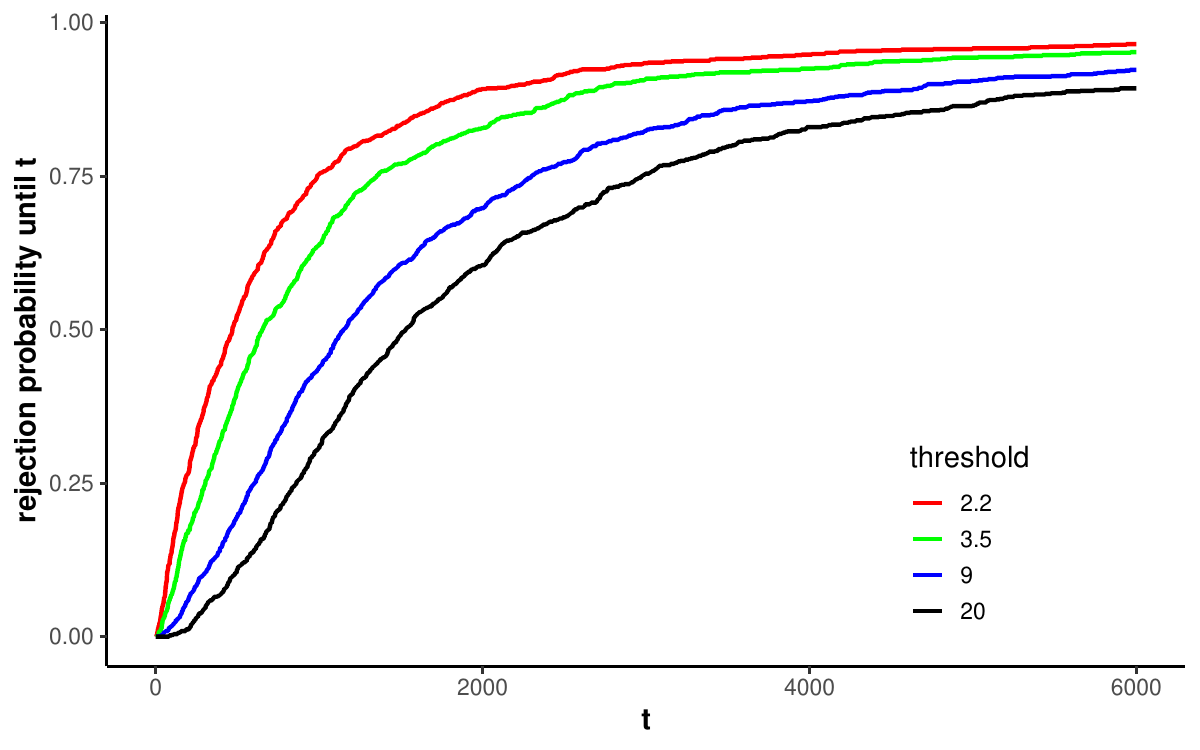}
  \caption[Cumulative rejection probabilities for forecasts estimating $\ES_{0.99}$ and $\VaR_{0.99}$ of an AR(1)-GARCH(1,1) time series with skewed t-distributed innovations under the working assumption of $t$-distributed innovations]{Cumulative rejection probabilities for forecasts estimating $\ES_{0.99}$ and $\VaR_{0.99}$ of an AR(1)-GARCH(1,1) time series with skewed $t$-distributed innovations under the working assumption of (unskewed) t-distributed innovations. All results are based on $1{,}000$ Monte Carlo simulation runs, for different thresholds.}
  \label{fig_Time_Series_skewed_T_T}
\end{figure}

From these results, we conclude that forecasts assuming normally distributed innovation variables are considerably easy to detect under all considered scenarios. When backtesting estimates at level $0.95$, this also holds true (although to a lesser degree) for forecasts assuming (unskewed) $t$-distributed innovations. When considering the level $p=0.99$, this is not the case. Here, sample sizes of at least $2{,}000$ or even higher are required in order to achieve certain detection probabilities.   
\subsubsection{Misspecified time series model}
We have also considered scenarios, where the structure of the underlying ARMA-GARCH model is estimated incorrectly. For this, we first consider an AR(1)-ARCH(1,1) time series with skewed t-distributed innovations and the same parameters as described in the previous scenario, and we have simulated that the forecaster fits a GARCH(1,1), AR(1)-ARCH(0,1),  or AR(1)-ARCH(1,0) model, respectively, and calculates estimates for the value-at-risk and the expected shortfall using this misspecified time series.

\begin{figure}[!htb]
  \centering
  \includegraphics[width=0.8\textwidth]{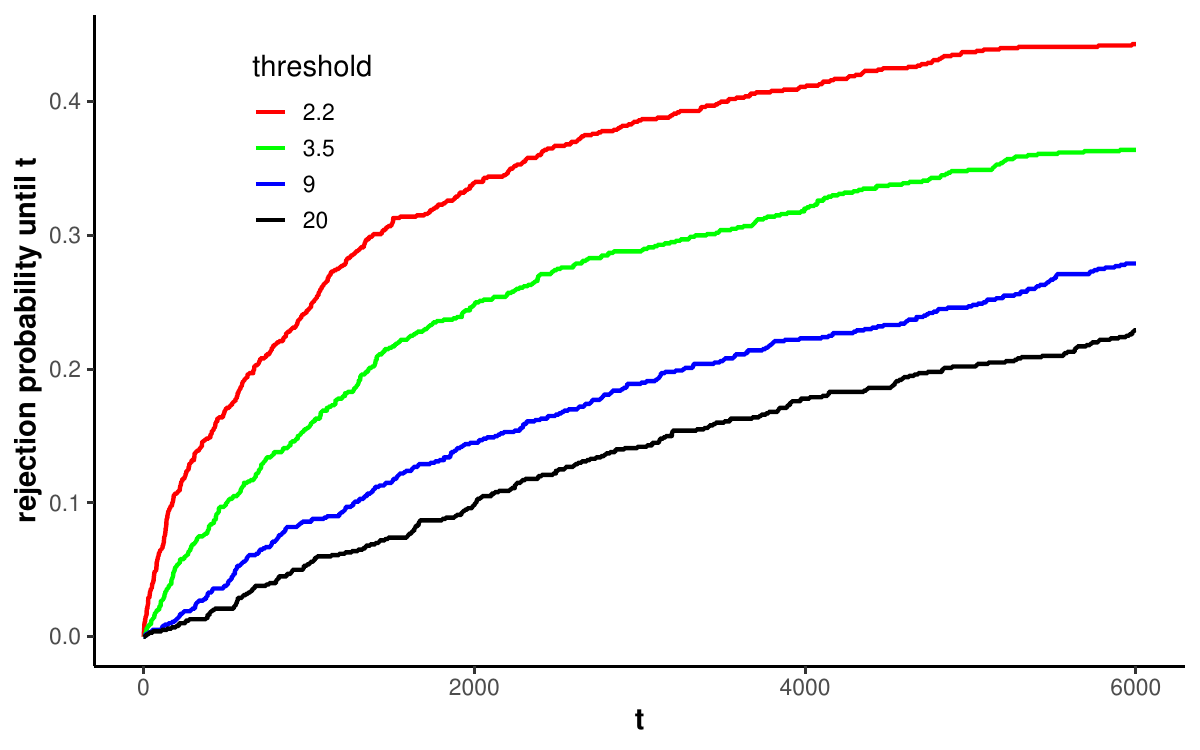}
  \caption[Cumulative rejection probabilities for forecasts estimating $\ES_{0.95}$ and $\VaR_{0.95}$ of an AR(1)-GARCH(1,1) time series by fitting a GARCH(1,1) time series]{Cumulative rejection probabilities for forecasts estimating $\ES_{0.95}$ and $\VaR_{0.95}$ of an AR(1)-GARCH(1,1) time series with skewed $t$-distributed innovations by fitting a GARCH(1,1) time series. All results are based on $1{,}000$ Monte Carlo simulation runs for different thresholds.}
  \label{fig_without_AR}
\end{figure}
For the case of fitting a GARCH(1,1) model, we notice that detections are very difficult; see Figure \ref{fig_without_AR} for the case of $p=0.95$. In the case of $p=0.99$, rejections were even less frequent. This is likely due to the average projected risk measures not differing much from the average risk measures when considering the true autoregressive model. Since both the GREE and the GREL betting processes do not increase bets if e-statistics were large only recently, but consider the whole time horizon, these betting processes do not utilize the autoregressive structure of the true time series optimally.

\begin{figure}[!htb]
  \centering
  \includegraphics[width=0.8\textwidth]{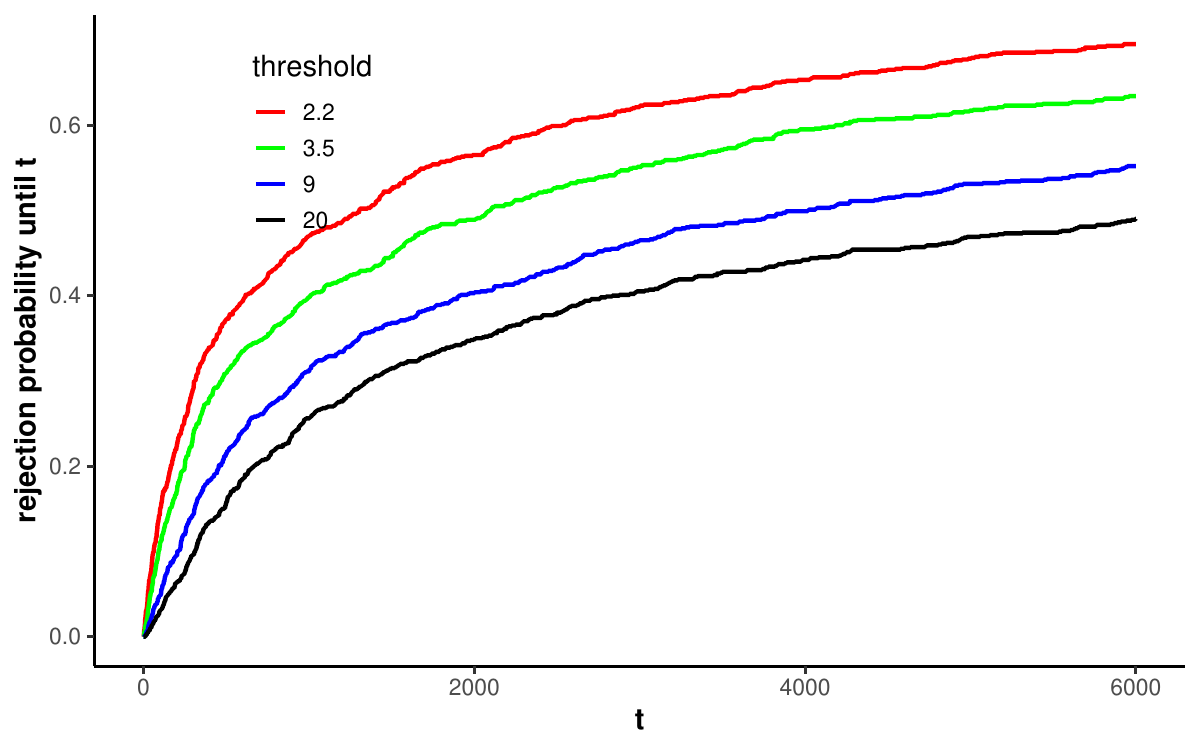}
  \caption[Cumulative rejection probabilities for forecasts estimating $\ES_{0.95}$ and $\VaR_{0.95}$ of an AR(1)-GARCH(1,1) time series derived by fitting an AR(1)-GARCH(0,1) time series]{Cumulative rejection probabilities for forecasts estimating $\ES_{0.95}$ and $\VaR_{0.95}$ of an AR(1)-GARCH(1,1) time series with skewed t-distributed innovations derived by fitting an AR(1)-GARCH(0,1) time series. All results are based on $1{,}000$ Monte Carlo simulations for different thresholds.}
  \label{fig_without_ARCH}
\end{figure}

\begin{figure}[!htb]
  \centering
  \includegraphics[width=0.8\textwidth]{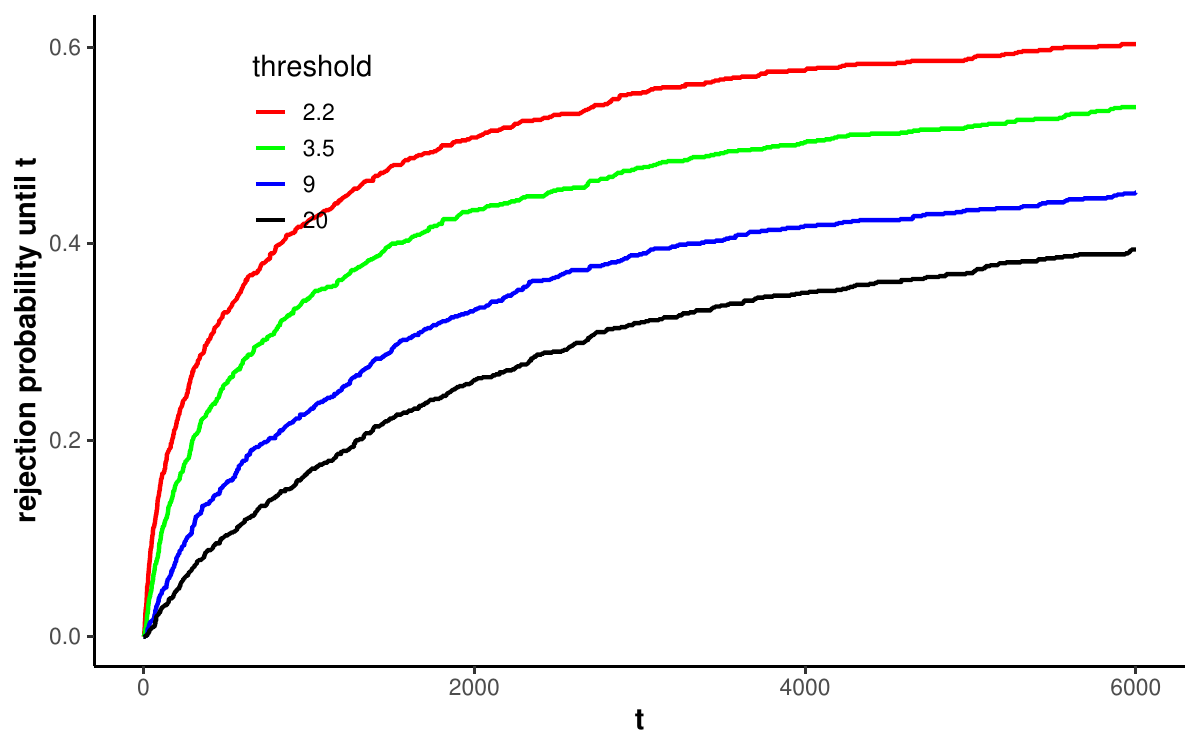}
  \caption[Cumulative rejection probabilities for forecasts estimating $\ES_{0.95}$ and $\VaR_{0.95}$ of an AR(1)-GARCH(1,1) time series derived by fitting an AR(1)-GARCH(1,0) time series]{Cumulative rejection probabilities for forecasts estimating $\ES_{0.95}$ and $\VaR_{0.95}$ of an AR(1)-GARCH(1,1) time series with skewed $t$-distributed innovations derived by fitting an AR(1)-GARCH(1,0) time series. All results are based on $1{,}000$ Monte Carlo simulation runs for different thresholds.}
  \label{fig_without_G}
\end{figure}

%
Misspecification of the GARCH part of the time series also led to surprisingly small detection probabilities, although to a lesser degree, when compared to detection probabilities under changes in the AR part. The largest rejection probabilities occured when omitting the coefficient $\alpha_1$ (see Figure \ref{fig_without_ARCH}) or both the $\alpha_1$ and the $\beta_1$ coefficient compared to just omitting $\beta_1$ (see Figure \ref{fig_without_G}). This again illustrates that the GREE and the GREL betting processes are unable to handle autoregressive structures well, also with regard to the GARCH part of the time series.

Therefore, we also tried using the $T$-FH GREM betting process with $T=50$ for the aforementioned scenarios of misspecification of the true AR-GARCH structure. The corresponding results are provided in Figures \ref{fig_without_AR_FH}, \ref{fig_without_ARCH_FH}, and \ref{fig_without_G_FH}, as well as in Table \ref{TabelleTimeSeriesFH}. We notice a substantial increase in the rejection probabilities when using the $50$-FH GREM betting process, especially in the presence of misspecification of the GARCH part of the time series. For practice, we thus recommend the usage of a $T$-FH GREM betting process for these types of scenarios. 

\begin{figure}[!htb]
  \centering
  \includegraphics[width=0.8\textwidth]{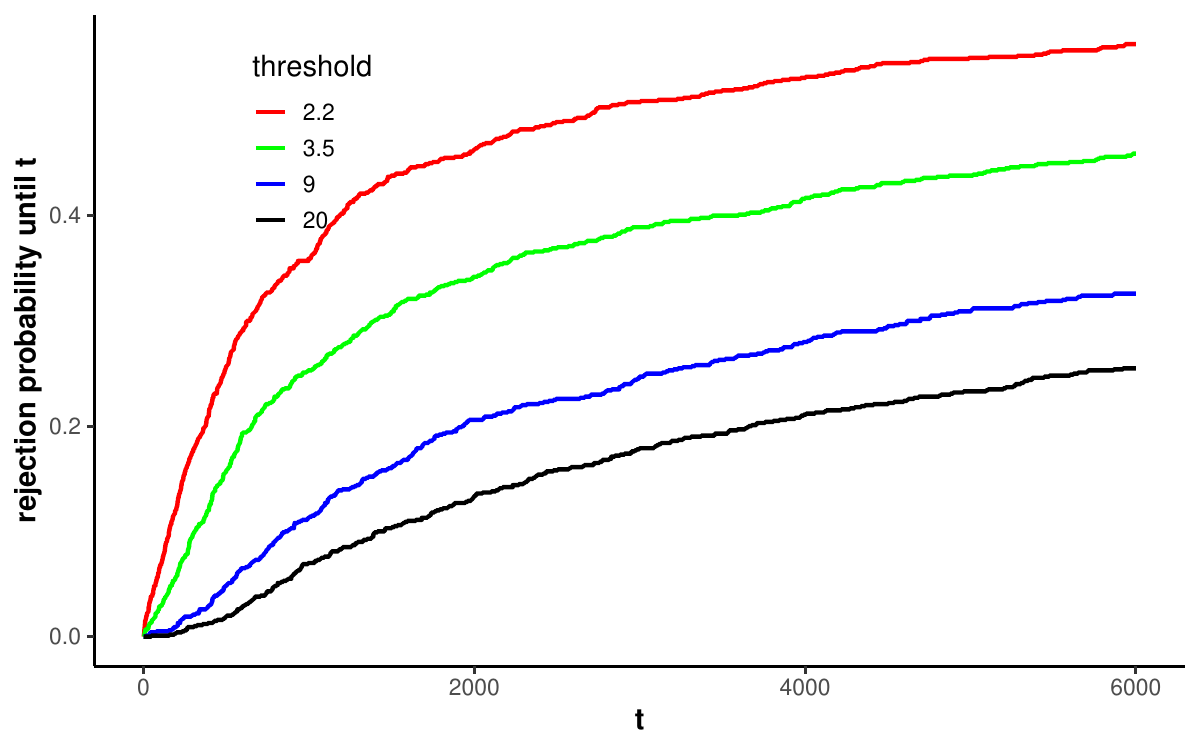}
  \caption[Cumulative rejection probabilities for forecasts estimating $\ES_{0.95}$ and $\VaR_{0.95}$ of an AR(1)-GARCH(1,1) time series derived by fitting a GARCH(1,1) time series using the \textbf{$50$-FH GREM} betting process]{Cumulative rejection probabilities for forecasts estimating $\ES_{0.95}$ and $\VaR_{0.95}$ of an AR(1)-GARCH(1,1) time series with skewed $t$-distributed innovations derived by fitting a GARCH(1,1) time series using the $50$-FH GREM betting process. All results are based on $1{,}000$ Monte Carlo simulation runs for different thresholds.}
  \label{fig_without_AR_FH}
\end{figure}

\begin{figure}[!htb]
  \centering
  \includegraphics[width=0.8\textwidth]{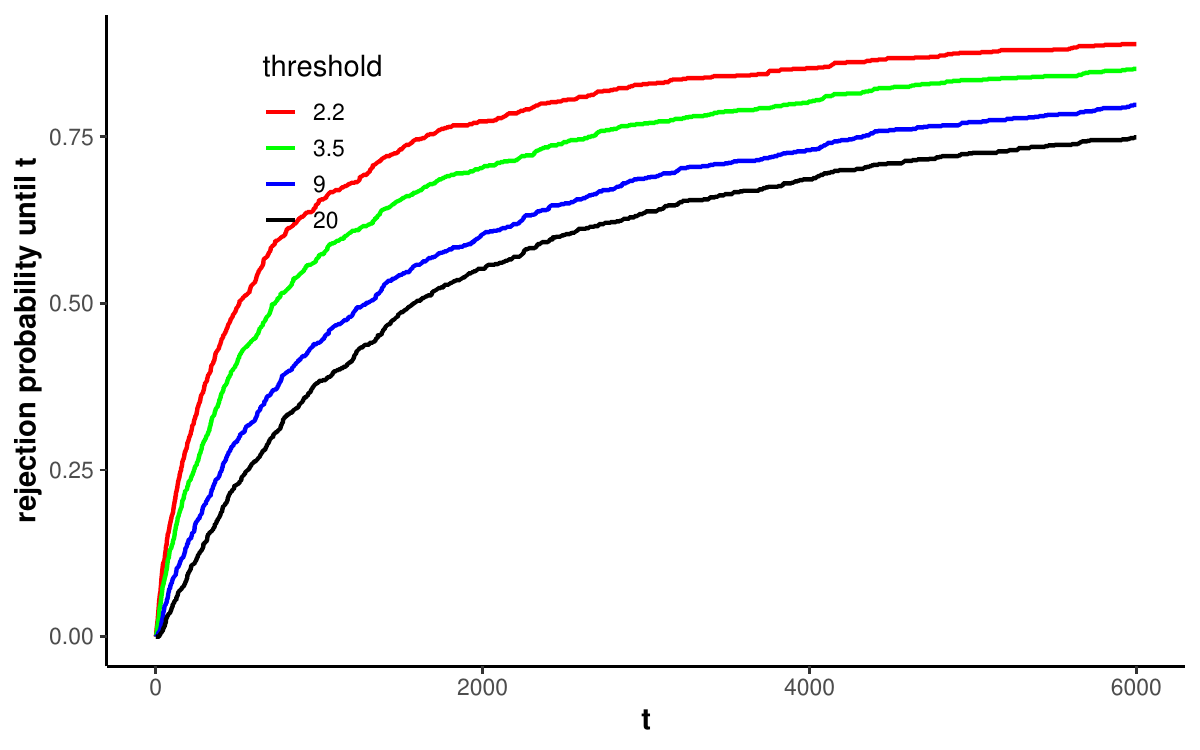}
  \caption[Cumulative rejection probabilities for forecasts estimating $\ES_{0.95}$ and $\VaR_{0.95}$ of an AR(1)-GARCH(1,1) time series derived by fitting an AR(1)-GARCH(0,1) time series using the \textbf{$50$-FH GREM} betting process]{Cumulative rejection probabilities for forecasts estimating $\ES_{0.95}$ and $\VaR_{0.95}$ of an AR(1)-GARCH(1,1) time series with skewed $t$-distributed innovations derived by fitting an AR(1)-GARCH(0,1) time series using the $50$-FH GREM betting process. All results are based on $1{,}000$ Monte Carlo simulation runs for different thresholds.}
  \label{fig_without_ARCH_FH}
\end{figure}

\begin{figure}[!htb]
  \centering
  \includegraphics[width=0.8\textwidth]{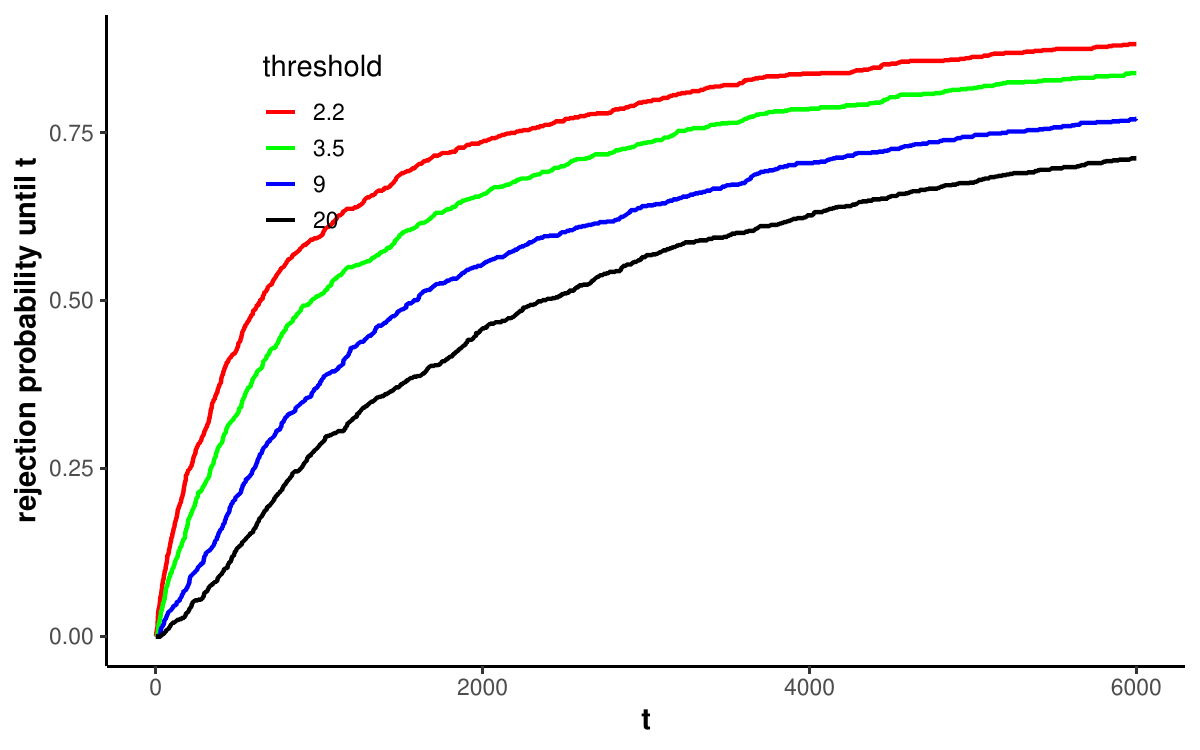}
  \caption[Cumulative rejection probabilities for forecasts estimating $\ES_{0.95}$ and $\VaR_{0.95}$ of an AR(1)-GARCH(1,1) time series derived by fitting an AR(1)-GARCH(1,0) time series using the \textbf{$50$-FH GREM} betting process]{Cumulative rejection probabilities for forecasts estimating $\ES_{0.95}$ and $\VaR_{0.95}$ of an AR(1)-GARCH(1,1) time series with skewed $t$-distributed innovations derived by fitting an AR(1)-GARCH(1,0) time series using the $50$-FH GREM betting process. All results are based on $1{,}000$ Monte Carlo simulation runs for different thresholds.}
  \label{fig_without_G_FH}
\end{figure}


\setlength{\tabcolsep}{1pt}
\begin{table}[!htb]
    \centering
    \resizebox{\textwidth}{!}{
    \begin{tabular}{c|c c c c|c c c c|c c c c}
    Power& \multicolumn{4}{c|}{70 \%} &  \multicolumn{4}{c|}{80 \%} & \multicolumn{4}{c}{90 \%}  \\
    \hline 
         \diagbox{(p,q)}{$\tau$} &2.2 & 3.5 & 9 & 20 & 2.2 & 3.5 & 9 & 20 & 2.2 & 3.5 & 9 & 20 \\ 
               \hline
               (0,1) & 1312 & 1949 & 3198 & 4181 &2355 & 3914 & >6000 & >6000 & >6000 & >6000 & >6000 & >6000 \\ 
               (1,0) &1593 & 2475 & 3843 & 5634 &3086 & 4470 & >6000 & >6000 & >6000 & >6000 & >6000 & >6000 \\ 
               (0,0) & 1097 & 1527 & 2502 & 3388 & 1696 & 2529 & 4282 & 5148 & 3886 & 5581 & >6000 & >6000
    \end{tabular}
    }
    \caption[Minimum sample sizes needed to achieve certain powers for estimating $\ES_{0.95}$ and $\VaR_{0.95}$ of an AR(1)-GARCH(1,1) time series derived by fitting a misspecified AR(1)-GARCH(p,q) model using the \textbf{$50$-FH GREM} betting process]{Minimum sample sizes needed to achieve certain powers when using specific thresholds $\tau$ for estimating $\ES_{0.95}$ and $\VaR_{0.95}$ of an AR(1)-GARCH(1,1) time series with skewed $t$-distributed innovations derived by fitting a misspecified AR(1)-GARCH(p,q) model using the $50$-FH GREM betting process. All results are based on $1{,}000$ Monte Carlo simulation runs each.}
    \label{TabelleTimeSeriesFH}
\end{table}
\setlength{\tabcolsep}{4pt}
\subsubsection{Constant forecasts of the value-at-risk and the expected shortfall}
As our final scenario, we consider forecasts, where the estimates for the value-at-risk $z_t$ and expected shortfall $r_t$ are constant for every $t\in \N$ and calculated using the unconditional mean $\mu$ and the unconditional variance $\sigma^2$ of the considered AR(1)-GARCH(1,1) time series. The latter quantities are given by
\[
\mu=\frac{\alpha'_0}{1-\alpha_1'} \text{~~and~~} 
\sigma^2=\frac{\alpha_0}{1-\alpha_1-\beta_1}.
\]
Since the actual values $\mu_t$ and $\sigma_t$ vary over time, the resulting risk estimates can severely underestimate the true expected shortfall or value-at-risk at certain time points. However, the reported risk will be roughly the same as the true risk averaged over time. We thus expect our backtesting procedure using the GREM betting process to perform poorly, since both the GREE and the GREL processes (which are equal under this scenario) do not increase bets if the true risk of some loss variable $L_t$ at time point $t\in \N$ is large. Our corresponding simulation results are summarized in Table \ref{TabelleTimeSeriesConstant} and Figure \ref{fig_Time_Series_Skewed_T_constant}.

\setlength{\tabcolsep}{1pt}
\begin{table}[!htb]
    \centering
    \resizebox{\textwidth}{!}{
    \begin{tabular}{c|c c c c|c c c c|c c c c}
    Power& \multicolumn{4}{c|}{70 \%} &  \multicolumn{4}{c|}{80 \%} & \multicolumn{4}{c}{90 \%}  \\
    \hline 
         \diagbox{p}{$\tau$} &2.2 & 3.5 & 9 & 20 & 2.2 & 3.5 & 9 & 20 & 2.2 & 3.5 & 9 & 20 \\ 
               \hline
          \multicolumn{13}{c}{\textbf{normal distribution}} \\
         0.99 & 977 & 1271 & 1850 & 2305 & 1403 & 1841 & 2459 & 2995 & 2295 & 2735 & 3455 & 4083 \\ 
      0.95 & 826 & 1103 & 1763 & 2195 & 1327 & 1684 & 2462 & 3022 & 2166 & 2695 & 3675 & 4309 \\ 
      \multicolumn{13}{c}{\textbf{t-distribution}} \\
         0.99 & 4002 & 5499 & >6000 & >6000 & >6000 & >6000 & >6000 & >6000 & >6000 & >6000 & >6000 & >6000 \\ 
      0.95 & 2141 & 3213 & 5388 & >6000 & 3917 & 5397 & >6000 & >6000 & >6000 & >6000 & >6000 & >6000 \\
      \multicolumn{13}{c}{\textbf{skewed t-distribution}} \\
         0.99 & >6000 & >6000 & >6000 & >6000 & >6000 & >6000 & >6000 & >6000 & >6000 & >6000 & >6000 & >6000 \\ 
      0.95 & 5658 & >6000 & >6000 & >6000 & >6000 & >6000 & >6000 & >6000 & >6000 & >6000 & >6000 & >6000 \\ 
    \end{tabular}
    }
    \caption[Minimum sample sizes needed to achieve certain powers for forecasts 
    estimating $\ES_{p}$ and $\VaR_{p}$ of an AR-GARCH time series as constant]{Minimum sample sizes needed to achieve certain powers for constant forecasts 
    of $\ES_{p}$ and $\VaR_{p}$ of an AR(1)-GARCH(1,1) time series with different innovation distributions, for different thresholds $\tau$ and levels $p$. All results are based on $1{,}000$ Monte Carlo simulation runs each and a maximum sample size of $6{,}000$.}
    \label{TabelleTimeSeriesConstant}
\end{table}
\setlength{\tabcolsep}{4pt}



\begin{figure}[!htb]
  \centering
  \includegraphics[width=0.8\textwidth]{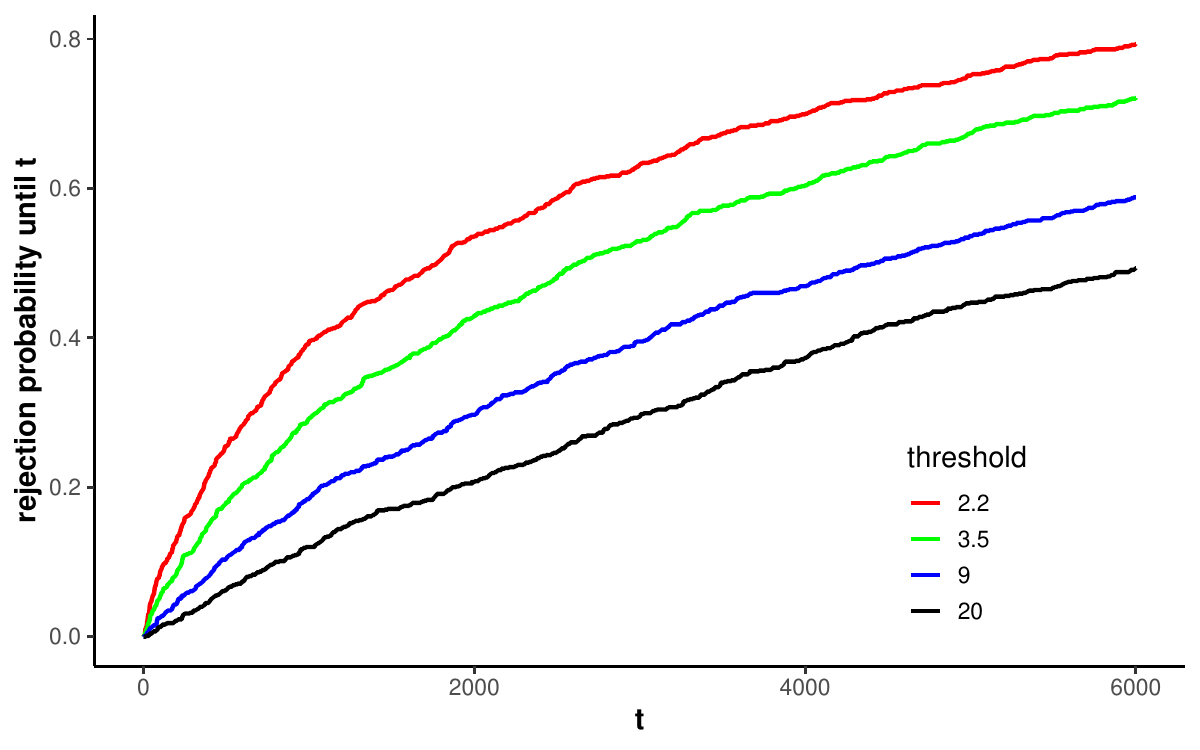}
  \caption[Cumulative rejection probabilities for forecasts estimating $\ES_{0.99}$ and $\VaR_{0.99}$ of an AR-GARCH time series with skewed t-distributed innovations as constant]{Cumulative rejection probabilities for constant forecasts of $\ES_{0.99}$ and $\VaR_{0.99}$ of an AR(1)-GARCH(1,1) time series with skewed $t$-distributed innovations with five degrees of freedom and skewness parameter $1.5$. All results are based on $1{,}000$ Monte Carlo simulation runs for different thresholds.}
  \label{fig_Time_Series_Skewed_T_constant}
\end{figure}

Indeed, we find detections to be considerably difficult. In the case of normally distributed innovations, this is not as pronounced as in the case of the other considered innovation distributions. For the normal distribution case, a sample size of $2{,}000-2{,}500$ will yield acceptable rejection probabilities for most of the considered scenarios. However, considering $t$- and especially skewed $t$-distributed innovation variables, detections often do not occur even after $6{,}000$ observations. Thus, one might consider other betting processes than the GREM betting process, for example a $T$-FH GREM betting process, in order to detect (more reliably) forecasting models for $\ES_p$, where the average risk forecasts are accurate (or overestimating the true risks), but large risks at certain time points are underestimated.  One might also combine the finite-horizon betting processes with the GREE and the GREL betting processes in a similar manner as suggested by Lemma \ref{GREMGREEGREL}, in order to keep the asymptotic optimality of these betting processes in the scenarios mentioned in Theorem \ref{AsymOptiofAlgs}. 
\subsection{Analytic derivation of sample size bounds}\label{sec64}
In the previous sections, we mostly derived the required sample sizes to achieve certain powers by means of numerical computer simulations. In most practical applications, this is necessary since the models are usually too complex to be precisely tractable. In this section, however, we propose an analytical approach for determining sample size bounds. 

For this, let $p\in (0,1)$ and $(X_t)_{t\in \N}$ be a sequence of i.i.d. Bernoulli-distributed loss variables with success parameter $1-p$, defined on some probability space $(\Omega^\prime,\mathcal{F}^\prime,\mathbb{Q})$. Obviously, $\VaR_{p}(X_1)=0$ and $\ES_p(X_1)=1$. Furthermore, as before, denote by
$(\Omega,\mathcal{F},(\mathcal{F}_n)_{n\in T},\mathbb{P})$ the filtered probability space
on which the observed losses $(L_t)_{t\in T}$ are defined. 
Now, consider $r\in (0,1)$ to be a constant (under-)estimate of $\ES_p(X_t)$, where we assume that $\VaR_{p}(X_t)$ is estimated correctly. Our proposed sample size bounds are then based on the following result.
\begin{theorem}\label{SSBBBestCase}
  Under the aforementioned specifications, let $(L_t)_{t\in \N}$ be a sequence of loss variables and $(r_t,z_t)_{t\in \N}$ be a sequence of estimates for $\ES_p(L_t|\mathcal{F}_{t-1})$ and $\VaR_p(L_t|\mathcal{F}_{t-1})$ such that $z_t=\VaR_p(L_t|\mathcal{F}_{t-1})$ and $(r_t-z_t)\geq r \cdot \mathbb{E}[L_t-z_t|L_t>z_t, \mathcal{F}_{t-1}]$ for all $t\in \N$. Then, it holds for all $t\in \N$, that
  \begin{align*}
    \max_{\lambda\in [0,1]} \mathbb{E}^\mathbb{P}[\log(1-\lambda+\lambda e^{\ES}_p(L_t,r_t,z_t))|\mathcal{F}_{t-1}] \leq \max_{\lambda \in [0,1]} \mathbb{E}^\mathbb{Q}[\log(1-\lambda+\lambda e^{\ES}_p(X_t,r,0))].
  \end{align*}
\end{theorem}
\begin{proof}
  Let $t\in \N$ and $\lambda \in [0,1]$. By splitting the expectation pertaining to $\mathbb{P}$ and using the definition of $e^{\ES}_p$ as well as $z_t=\VaR_p(L_t|\mathcal{F}_{t-1})$, we derive that
  \begin{align}
 &\mathbb{E}^\mathbb{P}[\log(1-\lambda+\lambda e^{\ES}_p(L_t,r_t,z_t))|\mathcal{F}_{t-1}] \nonumber \\ 
 \leq & p \log(1-\lambda) + (1-p)\mathbb{E}^\mathbb{P}[\log(1-\lambda+\lambda e^{\ES}_p(L_t,r_t,z_t))|L_t> z_t,\mathcal{F}_{t-1}] \label{Gleichung91}\, .
  \end{align}
Since $\log$ is concave, we get by Jensen's inequality, that
  \begin{align}
    &\mathbb{E}^\mathbb{P}[\log(1-\lambda+\lambda e^{\ES}_p(L_t,r_t,z_t))|L_t> z_t, \mathcal{F}_{t-1}] \nonumber \\ 
    \leq& \log\left( 1-\lambda +\lambda \frac{\mathbb{E}^\mathbb{P}[(L_t-z_t)_+|L_t>z_t,\mathcal{F}_{t-1}]}{(1-p)(r_t-z_t)}\right) \nonumber\\ 
    \leq& \log\left( 1-\lambda +\lambda \frac{1}{(1-p)r}\right). \label{Gleichung92}
  \end{align}
  Considering $X_t$, we calculate straightforwardly, that
  \begin{align*}
  \mathbb{E}^\mathbb{Q}[\log(1-\lambda+\lambda e^{\ES}_p(X_t,r,0))]=p \log(1-\lambda)+(1-p)\log\left( 1-\lambda +\lambda \frac{1}{(1-p)r}\right) \, ,
  \end{align*}
  which proves the result in conjunction with inequalities \eqref{Gleichung91} and \eqref{Gleichung92}. 
\end{proof}
Theorem \ref{SSBBBestCase} shows that when using the described backtesting procedure with the GRO betting process, the e-process testing the $\ES_p$ forecasts $r$ for $(X_t)_{t\in \N}$ has no smaller e-power than the e-process testing a sequence of loss variables where the value-at-risk is always estimated correctly and the estimated expected excess $r_t-z_t$ equals $r$ times the true difference between the conditional value-at-risk (which equals $\ES_p(L_t|\mathcal{F}_{t-1})$ if $\mathbb{P}(L_t \leq\VaR_p(L_t)|\mathcal{F}_{t-1})=p$) and the value-at-risk. Thus, when determining sample sizes required to achieve certain powers, the ones for backtesting the forecasts for $(X_t)_{t\in \N}$ should be seen as a lower bound for the ones needed when backtesting the forecasts for $(L_t)_{t\in \N}$. 

 Furthermore, one can calculate the GRO betting process for backtesting in the case of $(X_t)_{t\in \N}$ exactly, leading to the following result. 
 \begin{corollary}\label{GROTheoretical}
  Under the aforementioned conditions, it holds for all $t\in \N$
  \begin{align*}
    \argmax_{\lambda \in [0,1]} \mathbb{E}^\mathbb{Q}[\log(1-\lambda+\lambda e^{\ES}_p(X_t,r,0))]= \frac{(1-p)(1-r)}{1-(1-p)r}.
  \end{align*}
 \end{corollary}
Corollary \ref{GROTheoretical} shows in particular that the GRO betting process when backtesting these risk forecasts is constant in $t$. Since, at each time point $t\in \N$, $e^{\ES}_p(X_t,r,0)$ can only attain two values, either $[(1-p)r]^{-1}$  if the (predicted) value-at-risk is exceeded or zero otherwise, the terms $\left(1-\lambda^\text{GRO}_t+\lambda^\text{GRO}_t e^{\ES}_p(X_t,r,0)\right)_t$ can also only attain two values. Thus, the value of our e-process at some point $n\in \N$ is only dependent on the number of value-at-risk exceedances of $(X_t)_{t\in \N}$ up to $n$. Since we are mostly interested in the e-process $(M_n)_{n\in \N}$ exceeding certain thresholds $\tau$, we state the following definition.
\begin{definition}
Let $p,r\in (0,1)$, $\tau \in (1,\infty)$, and $n \in \N$. If $\left(1+ \frac{1-r}{r}\right)^{n}\geq \tau$, the number of required value-at-risk exceedances, denoted by NORVE$(p,r,\tau,n)$ for level $p$, relative underestimation $r$, threshold $\tau$ and sample size $n$, is the smallest $n_0\in \N$ such that 
\begin{align}
  \left(1+ \frac{1-r}{r}\right)^{n_0}\left(1- \frac{(1-p)(1-r)}{1-(1-p)r}\right)^{n-n_0}\geq \tau,\label{GleichungTheoretical}
\end{align}
and $\infty$ otherwise.
\end{definition}
\begin{corollary}\label{coro-norve}
  Under the aforementioned condition, NORVE$(p,r,\tau,n)$ is the minimum number of value-at-risk exceedances up to time point $n$, when e-backtesting the estimate $r$ for $\ES_{p}(X_t)$ using the GRO betting process, such that the corresponding e-process $(M_n)_{n\in \N}$ exceeds the threshold $\tau$ at $n$.  
\end{corollary}
\begin{proof}
As argued previously, $M_n$ can only attain values of the form
\begin{align*}
  \left(1+\lambda^\text{GRO}_t \left(\frac{1}{(1-p)r}-1\right)\right)^m \left(1- \lambda^\text{GRO}_t\right) ^{n-m}
\end{align*}
for $m\in \N_0$ with $m\leq n$. Since $\lambda^\text{GRO}_t= [(1-p)(1-r)] / [1-(1-p)r]$ for all $t\in \N$ by Corollary \ref{GROTheoretical}, straightforward calculation leads to the representation on the left-hand side of \eqref{GleichungTheoretical}. 
\end{proof}
Corollary \ref{coro-norve} enables a simple way of approximating the sample sizes required to reject an estimate $r$ for $\ES_p(X_t)$ with a certain power $q$ when applying the described e-backtesting procedure employing a threshold $\tau$: One determines the smallest $n\in \N$ such that the probability of the observed  number of value-at-risk exceedances being larger than or equal to NORVE$(p,r,\tau,n)$ is larger than $q$. Since we assumed $(X_t)_{t\in \N}$ to be i.i.d., this probability can be calculated using the binomial distribution. We summarize our findings in the following definition.
\begin{definition}
  Let $p,r,q\in (0,1)$ and $\tau \in (1,\infty)$. Then, the sample size bound for backtesting, denoted SSBB$(p,r,\tau,q)$, for forecasts of $\ES_{p}$ such that $z_t=\VaR_p(L_t|\mathcal{F}_{t-1})$ and $r_t-z_t=r \cdot (\ES_p(L_t|\mathcal{F}_{t-1})-\VaR_p(L_t|\mathcal{F}_{t-1}))$ for all $t\in \N$ is given by
  \begin{align*}
  \text{SSBB}(p,r,q,\tau):=\min\{n\in \N: F_{n,p}(n-\text{NORVE}(p,r,\tau,n))>q\},
  \end{align*} 
  where $F_{n,p}$ denotes the cumulative distribution function of the binomial distribution with parameters $n$ and $p$,  which we denote by $\operatorname{Binom}(n,p)$. 
\end{definition}
\setlength{\tabcolsep}{1pt}
\begin{table}[!htb]
    \centering
    \resizebox{\textwidth}{!}{
    \begin{tabular}{c|c c c c|c c c c|c c c c}
    Power& \multicolumn{4}{c|}{70 \%} &  \multicolumn{4}{c|}{80 \%} & \multicolumn{4}{c}{90 \%}  \\
    \hline 
         \diagbox{r}{$\tau$} &2.2 & 3.5 & 9 & 20 & 2.2 & 3.5 & 9 & 20 & 2.2 & 3.5 & 9 & 20 \\ 
               \hline
               \multicolumn{13}{c}{\textbf{p=0.95}} \\
               0.2 & 49 & 49 & 95 & 117 & 59 & 85 & 134 & 157 & 132 & 158 & 184 & 234 \\ 
               0.3 & 72 & 95 & 140 & 184 & 110 & 134 & 204 & 249 & 209 & 234 & 282 & 353 \\ 
        0.4 & 95 & 140 & 227 & 292 & 157 & 226 & 294 & 383 & 306 & 353 & 446 & 538 \\ 
         0.5 & 162 & 227 & 356 & 462 & 272 & 339 & 493 & 601 & 492 & 583 & 718 & 851 \\ 
         0.6 & 292 & 398 & 609 & 797 & 471 & 601 & 837 & 1029 & 829 & 984 & 1246 & 1462 \\ 
         0.7 & 567 & 797 & 1212 & 1542 & 923 & 1177 & 1640 & 1995 & 1613 & 1891 & 2402 & 2825 \\ 
         0.8 & 1419 & 1954 & 2957 & 3773 & 2266 & 2847 & 3962 & 4867 & 3918 & 4587 & 5835 & 6830 \\ 
         0.9 & 6213 & 8463 & 12771 & 16244 & 9774 & 12298 & 17091 & 20942 & 16833 & 19631 & 24951 & 29223 \\
               \multicolumn{13}{c}{\textbf{p=0.99}} \\
               0.2 & 244 & 244 & 476 & 589 & 299 & 427 & 671 & 790 & 667 & 798 & 926 & 1175 \\
               0.3 & 361 & 476 & 700 & 920 & 551 & 671 & 1022 & 1251 & 1051 & 1175 & 1538 & 1776 \\
               0.4 & 476 & 700 & 1138 & 1462 & 906 & 1137 & 1588 & 1922 & 1538 & 1776 & 2358 & 2701 \\
         0.5 & 811 & 1246 & 1889 & 2313 & 1364 & 1811 & 2471 & 3123 & 2473 & 2929 & 3716 & 4383 \\ 
         0.6 & 1462 & 2101 & 3155 & 4096 & 2471 & 3123 & 4304 & 5368 & 4272 & 5045 & 6464 & 7547 \\ 
         0.7 & 2945 & 4096 & 6275 & 8030 & 4731 & 6004 & 8423 & 10409 & 8409 & 9802 & 12463 & 14579 \\ 
         0.8 & 7411 & 10086 & 15308 & 19593 & 11763 & 14773 & 20659 & 25287 & 20359 & 23807 & 30258 & 35440 \\ 
         0.9 & 32300 & 43956 & 66306 & 84479 & 50843 & 63871 & 88851 & 108818 & 87614 & 102115 & 129733 & 151910 \\
        \multicolumn{13}{c}{\textbf{p=0.999}} \\
         0.1 & 1204 & 2439 & 2439 & 3615 & 1609 & 2994 & 4278 & 4278 & 3889 & 5321 & 6679 & 6679 \\ 
        0.2 & 2439 & 2439 & 4762 & 5890 & 2994 & 4278 & 6720 & 7905 & 6679 & 7992 & 9273 & 11769 \\ 
        0.3 & 3615 & 4762 & 7005 & 9208 & 5514 & 6720 & 10231 & 12517 & 10530 & 11769 & 15404 & 17779 \\ 
        0.4 & 4762 & 7005 & 11387 & 14622 & 9074 & 11379 & 15896 & 19232 & 15404 & 18955 & 23603 & 27042 \\ 
         0.5 & 8110 & 12469 & 18897 & 24197 & 13649 & 18124 & 25818 & 31246 & 24753 & 29317 & 37195 & 43868 \\ 
         0.6 & 14622 & 21022 & 32612 & 40974 & 24726 & 31246 & 44133 & 53710 & 43868 & 50486 & 64688 & 75517 \\ 
         0.7 & 29464 & 40974 & 62772 & 80317 & 48397 & 61122 & 85316 & 104130 & 84136 & 98072 & 125752 & 146914 \\ 
         0.8 & 74132 & 101911 & 154127 & 196976 & 117671 & 148808 & 207669 & 254979 & 205763 & 240243 & 304751 & 357613
    \end{tabular}
    }
    \caption[\textbf{SSBB}s for levels $p=0.95,0.99,0.999$ and different powers, thresholds and proportional underestimations]{Sample size bounds for backtesting (SSBBs) for levels $p \in \{0.95,0.99,0.999\}$ and different powers $q$, thresholds $\tau$, and proportional underestimations $r$.}
    \label{SSBBs99}
\end{table}
\setlength{\tabcolsep}{4pt}

SSBBs for common choices for the required parameters are provided in Table \ref{SSBBs99}. Like the required sample sizes from our previous simulations, SSBBs for rejecting underestimates of $\ES_{0.99}$ or $\ES_{0.999}$, respectively,  are considerably higher than those needed to reject underestimates of $\ES_{0.95}$, often more than five times as high in the case of  $\ES_{0.99}$ and more than $50$ times as high in the case of $\ES_{0.999}$. These factors can also be approximately inferred from the formulas for the NORVE: Assuming $1-(1-p)r\approx 1$, the NORVE becomes the smallest integer $n_0\in \N$ such that
\begin{align*}
  \left(1+ \frac{1-r}{r}\right)^{n_0}\left(1-(1-p)(1-r)\right)^{n-n_0}\geq \tau.
\end{align*}
Furthermore, assuming that $(1-p)(1-r)$ is small, we can (for every factor $d\in (0,1)$) 
approximate $\left(1-d(1-p)(1-r)\right)^{n-n_0}\approx (\left(1-(1-p)(1-r)\right)^{n-n_0})^d$. Since $\left(1+ (1-r)/r\right)^{n_0}$ does not depend on $p$, this suggests that the NORVE, and therefore the SSBB, grows approximately linearily in $1/(1-p)$. 

Regarding the underestimation parameter $r$, we notice that the SSBBs are almost perfectly proportional to $(1-r)^{-2}$; see also Figure \ref{fig_r_to_SSBBs}. If one considers $1-r$ to be interpretable as a relative effect size for the underestimation of the expected shortfall, this is in line with other sample size formulas,  where the required sample size is usually proportional to the squared inverted effect size.

\begin{figure}[!htb]
  \centering
  \includegraphics[width=0.8\textwidth]{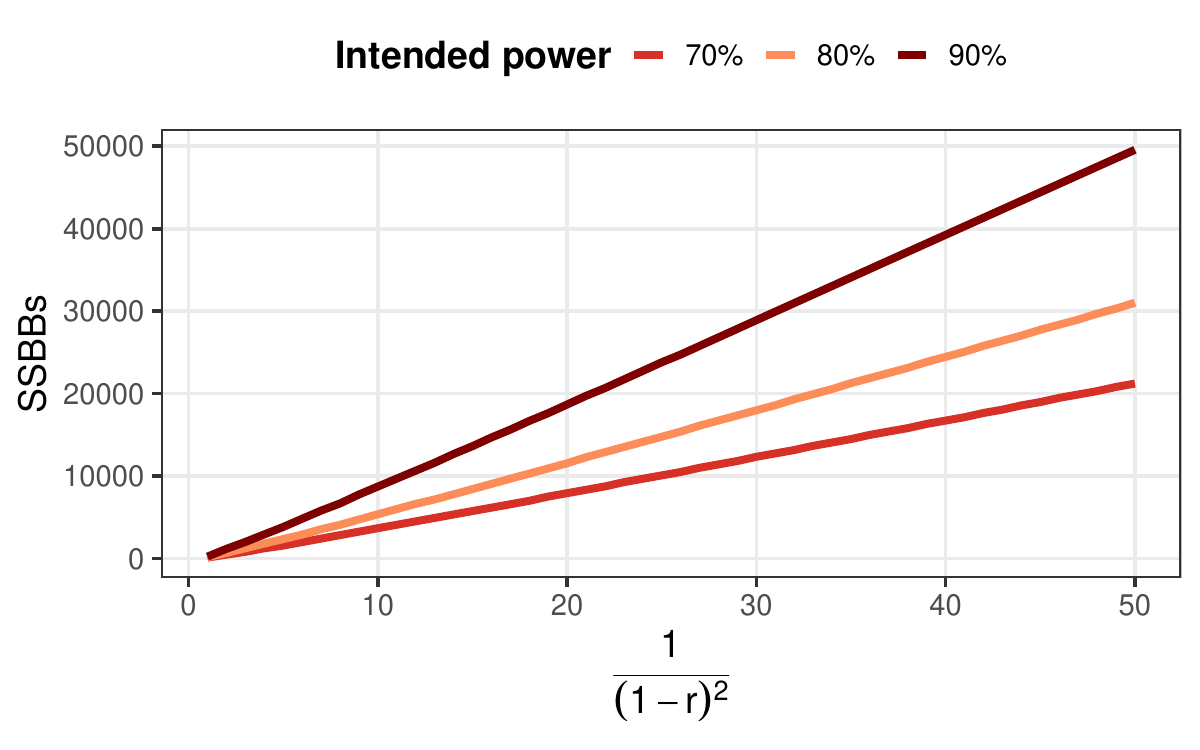}
  \caption{Sample size bounds for backtesting (SSBBs) as a function of $(1-r)^{-2}$ for the threshold $\tau=3.5$ and the level $p=0.99$.}
  \label{fig_r_to_SSBBs}
\end{figure}

One drawback of the proposed SSBBs is that they examine whether the threshold $\tau$ is exceeded precisely at the final considered time point $n$. While this is in line with conventional hypothesis testing, the backtesting procedure described in this work also allows for a rejection of the null hypothesis if the threshold in question is exceeded at some earlier time point $t < n$. Thus, the actually required sample sizes to detect underestimations of relative magnitude $r$ for $\ES_p(X_t)$ with a prescribed power are smaller than the corresponding SSBBs. We thus also conducted numerical experiments to approximate the actual required sample size. These results are summarized in Table \ref{TabelleSampleSizesX99}. We notice that the actual required sample sizes, while being of the same order of magnitude, are often considerably smaller than the corresponding SSBBs. This again illustrates the strength of anytime-valid inference. We notice that this difference is more pronounced when applying smaller thresholds. This is likely due to the probability of the value-at-risk exceedance events all occurring early being larger when considering smaller thresholds, as there are fewer exceedance events needed to pass these thresholds. Regarding the other parameters, no clear relationship is identifiable.

\setlength{\tabcolsep}{0.5pt}
\begin{table}[!htb]
    \centering
    \begin{tabular}{c|c c c c|c c c c|c c c c}
    Power& \multicolumn{4}{c|}{70 \%} &  \multicolumn{4}{c|}{80 \%} & \multicolumn{4}{c}{90 \%}  \\
    \hline 
         \diagbox{r}{$\tau$} &2.2 & 3.5 & 9 & 20 & 2.2 & 3.5 & 9 & 20 & 2.2 & 3.5 & 9 & 20 \\ 
               \hline
        \multicolumn{13}{c}{\textbf{p=0.95}} \\
        0.2 & 35 & 42 & 70 & 97 & 51 & 66 & 96 & 123 & 87 & 105 & 144 & 190  \\        
         0.3 & 44 & 66 & 108 & 152 & 70 & 95 & 144 & 205 & 112 & 149 & 231 & 283  \\  
         0.4 & 60 & 99 & 167 & 228 & 92 & 134 & 228 & 296 & 164 & 218 & 337 & 405  \\        
         0.5 & 99 & 160 & 265 & 367 & 150 & 216 & 354 & 462 & 246 & 349 & 519 & 671 \\ 
         0.6 & 161 & 266 & 459 & 644 & 241 & 363 & 609 & 822 & 410 & 569 & 962 & 1155 \\ 
         0.7 & 306 & 499 & 899 & 1219 & 454 & 728 & 1232 & 1625 & 801 & 1223 & 1760 & 2279 \\ 
         0.8 & 647 & 1118 & 2131 & 3025 & 956 & 1656 & 2843 & 3793 & 1767 & 2771 & 4118 & 5311 \\ 
         0.9 & 2651 & 4607 & 8723 & 11850 & 3829 & 6410 & 11378 & 14850 & 6528 & 9798 & 16789 & 21058 \\ 
        \multicolumn{13}{c}{\textbf{p=0.99}} \\
         0.2 & 163 & 221 & 383 & 519 & 232 & 332 & 506 & 688 & 377 & 535 & 753 & 930  \\        
         0.3 & 210 & 330 & 583 & 762 & 361 & 504 & 790 & 1003 & 567 & 821 & 1129 & 1384  \\  
         0.4 & 316 & 517 & 910 & 1182 & 546 & 773 & 1177 & 1592 & 901 & 1201 & 1875 & 2218  \\      
         0.5 & 500 & 820 & 1346 & 1877 & 797 & 1137 & 1843 & 2394 & 1298 & 1789 & 2582 & 3387 \\ 
         0.6 & 792 & 1307 & 2443 & 3296 & 1298 & 1957 & 3347 & 4328 & 2468 & 3231 & 4787 & 5985 \\ 
         0.7 & 1463 & 2362 & 4473 & 6113 & 2211 & 3427 & 5905 & 7549 & 3887 & 5351 & 8110 & 10228 \\ 
         0.8 & 3285 & 5983 & 10469 & 15148 & 4957 & 8718 & 13824 & 19323 & 8737 & 13589 & 20930 & 26876 \\ 
         0.9 & 14667 & 26224 & 49976 & 68633 & 22938 & 36946 & 66933 & 86391 & 42909 & 60961 & 94641 & >$10^5$ \\
        \multicolumn{13}{c}{\textbf{p=0.999}} \\
       0.1 & 1205 & 1470 & 2330 & 3310 & 1611 & 2358 & 3508 & 4041 & 3233 & 3535 & 4994 & 6175  \\ 
         0.2 & 1544 & 2057 & 3249 & 4708 & 2312 & 3173 & 4679 & 6217 & 4115 & 5394 & 7536 & 9083  \\        
         0.3 & 2211 & 3958 & 6385 & 8345 & 3764 & 5762 & 8370 & 10636 & 7047 & 9242 & 12324 & 15300  \\  
         0.4 & 2951 & 4931 & 8753 & 11576 & 4616 & 7006 & 11384 & 14814 & 7804 & 11955 & 16166 & 21373  \\    
         0.5 & 5228 & 8238 & 14291 & 19693 & 7929 & 12139 & 19111 & 25024 & 14499 & 20599 & 26960 & 34839 \\ 
         0.6 & 8743 & 13376 & 23382 & 32323 & 13175 & 19459 & 32206 & 42082 & 25693 & 32136 & 48002 & 58161 \\ 
         0.7 & 13652 & 23075 & 43559 & 62498 & 20554 & 32618 & 59237 & 80089 & 38636 & 55151 & 85336 & >$10^5$   
    \end{tabular}
    \caption[Minimum sample sizes needed to achieve certain powers for forecasts 
    underestimating $\ES_{p}(X_t)$]{Minimum sample sizes needed to achieve certain powers when using specific thresholds $\tau$ for forecasts 
    underestimating $\ES_{p}(X_t)$ for levels $p \in \{0.95,0.99,0.999\}$ by different relative magnitudes $r$. All results are based on $1{,}000$ Monte Carlo simulations each and a maximum sample size of $100{,}000$.}
    \label{TabelleSampleSizesX99}
\end{table}
\setlength{\tabcolsep}{4pt}

Comparing with our numerical experiments from the previous sections, we infer that the actually required sample sizes needed to achieve certain powers when estimating the value-at-risk correctly and underestimating the difference between the expected shortfall and value-at-risk by a given factor $r$ are usually larger than the corresponding SSBBs (see Tables \ref{TabelleNormSampleSizeVaRcorrect}, \ref{TabelleTSampleSizeVaRcorrect99} and \ref{TabelleGPDSampleSizeVaRCorrect} and compare them with Table \ref{SSBBs99}). Therefore, we recommend using the SSBBs as lower bounds for the sample size one needs to consider in order to reject underestimations of this type. We also infer that the ratio between the actually required sample sizes and SSBBs stays approximately constant when focusing on specific distributions over different parameters $p$, $r$, $\tau$, and $q$. Using these ratios as a factor, one might also use the SSBBs to extrapolate the actual required sample sizes to scenarios with a less pronounced underestimation of the expected shortfall, where we were unable to approximate the actually required sample sizes due to them becoming too large. 

SSBBs are derived under the assumption that the value-at-risk is estimated correctly. We thus also investigated how our findings from the study of SSBSs can be applied to situations where the value-at-risk is also misspecified. To assess this, we considered a sequence $(X_t)_{t\in \N}$ of i.i.d. random variables where the expected shortfall at some level $\alpha\in (0,1)$ is always underestimated as $r_t<\ES_{\alpha}(X_1)$ and the value-at-risk is estimated as $\VaR_{p}(X_1)$, where $p\in (0,1)$ is varied over different runs. As distributions, we considered $X_1$ to be  standard Student's $t$-distributed with five degrees of freedom (see Table \ref{TabelleESconstantT95}) and $X_1$ to have GPD(1,0.2)-distributed excess distributions over the threshold $\VaR_{0.95}(X_1)$ (see Table \ref{TabelleESconstantGPD99}). We infer that the required sample sizes do not change much if the estimate of the value-at-risk is reasonable, with small underestimations of the value-at-risk often requiring slightly larger sample sizes than when estimating it correctly. Nevertheless, these experiments indicate that the assumption of the value-at-risk being estimated correctly made when deriving the SSBBs will likely not be of high relevance in practice. Using the SSBBs to lower-bound the required sample sizes will most likely lead to an overestimation of the required sample size even if the value-at-risk is misspecified, assuming a fixed estimate for the expected shortfall.

\setlength{\tabcolsep}{3pt}
\begin{table}[!htb]
    \centering
    \begin{tabular}{c|c c c c|c c c c|c c c c}
    Power& \multicolumn{4}{c|}{70 \%} &  \multicolumn{4}{c|}{80 \%} & \multicolumn{4}{c}{90 \%}  \\
    \hline 
         \diagbox{p'}{$\tau$} &2.2 & 3.5 & 9 & 20 & 2.2 & 3.5 & 9 & 20 & 2.2 & 3.5 & 9 & 20 \\ 
          \hline
          \multicolumn{13}{c}{\textbf{p=0.95}} \\
          0.92 & 1048 & 1459 & 2113 & 2585 & 1512 & 1906 & 2608 & 3272 & 2340 & 2681 & 3378 & 4106 \\ 
          0.93 & 1120 & 1578 & 2439 & 3037 & 1574 & 2183 & 2982 & 3597 & 2365 & 2949 & 3782 & 4481 \\ 
          0.94 & 1372 & 1977 & 2880 & 3638 & 1968 & 2589 & 3626 & 4336 & 2898 & 3514 & 4619 & 5511 \\ 
          0.95 & 1252 & 1693 & 2561 & 3185 & 1747 & 2172 & 3151 & 3793 & 2425 & 3038 & 4179 & 4891 \\ 
          0.96 & 893 & 1212 & 1769 & 2320 & 1197 & 1561 & 2245 & 2712 & 1809 & 2225 & 2796 & 3434 \\ 
          \multicolumn{13}{c}{\textbf{p=0.99}} \\
          0.96 & 386 & 519 & 754 & 936 & 500 & 626 & 899 & 1120 & 666 & 845 & 1162 & 1406 \\ 
          0.97 & 582 & 744 & 1129 & 1433 & 767 & 973 & 1369 & 1713 & 1094 & 1391 & 1843 & 2193 \\       
          0.98 & 861 & 1163 & 1698 & 2124 & 1149 & 1470 & 2023 & 2485 & 1516 & 1876 & 2564 & 3078 \\ 
          0.99 & 749 & 930 & 1370 & 1804 & 953 & 1173 & 1688 & 2120 & 1243 & 1540 & 2166 & 2639 \\ 
          0.991 & 622 & 791 & 1154 & 1506 & 776 & 961 & 1362 & 1795 & 1043 & 1243 & 1787 & 2173 
    \end{tabular}
    \caption[Minimum sample sizes needed to achieve certain powers for forecasts 
    underestimating $\ES_{p}$ by a factor $0.9$ and estimating $\VaR_{p}$ as $\VaR_{p'}$ for standard t-distributed random variables]{Minimum sample sizes needed to achieve certain powers for forecasts 
    underestimating $\ES_{p}$ by a factor $0.9$ and estimating $\VaR_{p}$ as $\VaR_{p'}$ for levels $p \in \{0.95,0.99\}$ and  Student's $t$-distributed random variables with five degrees of freedom, for different thresholds $\tau$ and levels $p'$. All results are based on $1{,}000$ Monte Carlo simulation runs each and a maximum sample size of $6{,}000$.}
    \label{TabelleESconstantT95}
\end{table}
\setlength{\tabcolsep}{4pt}


\setlength{\tabcolsep}{2pt}
\begin{table}[!htb]
    \centering
    \begin{tabular}{c|c c c c|c c c c|c c c c}
    Power& \multicolumn{4}{c|}{70 \%} &  \multicolumn{4}{c|}{80 \%} & \multicolumn{4}{c}{90 \%}  \\
    \hline 
         \diagbox{p'}{$\tau$} &2.2 & 3.5 & 9 & 20 & 2.2 & 3.5 & 9 & 20 & 2.2 & 3.5 & 9 & 20 \\ 
          \hline
          \multicolumn{13}{c}{\textbf{p=0.99}} \\
          0.96 & 454 & 612 & 913 & 1178 & 598 & 762 & 1122 & 1397 & 868 & 1081 & 1404 & 1715 \\ 
          0.97 & 756 & 1006 & 1504 & 1865 & 971 & 1240 & 1785 & 2292 & 1385 & 1679 & 2315 & 2869 \\       
          0.98 & 1172 & 1607 & 2378 & 3016 & 1581 & 2076 & 2881 & 3623 & 2208 & 2718 & 3744 & 4483 \\ 
          0.99 & 1305 & 1647 & 2414 & 3034 & 1626 & 2030 & 2903 & 3598 & 2326 & 2718 & 3809 & 4513 \\ 
          0.992 & 836 & 1080 & 1594 & 2021 & 1049 & 1344 & 1901 & 2343 & 1396 & 1710 & 2274 & 2949 \\ 
          \multicolumn{13}{c}{\textbf{p=0.999}} \\
          0.994 & 1280 & 1679 & 2516 & 3247 & 1664 & 2118 & 2963 & 3821 & 2240 & 2674 & 3767 & 4821 \\ 
          0.995 & 1514 & 1944 & 2943 & 3789 & 1905 & 2428 & 3470 & 4454 & 2635 & 3182 & 4370 & 5321 \\       
          0.996 & 1994 & 2560 & 3788 & 4768 & 2484 & 3188 & 4472 & 5698 & 3506 & 4202 & 5864 & >6000 \\ 
          0.997 & 2412 & 3025 & 4534 & 5940 & 2930 & 3822 & 5495 & >6000 & 3865 & 4763 & >6000 & >6000 \\ 
          0.998 & 2316 & 3002 & 4287 & 5772 & 2942 & 3735 & 5263 & >6000 & 3951 & 4849 & >6000 & >6000 
    \end{tabular}
    \caption[Minimum sample sizes needed to achieve certain powers for forecasts 
    underestimating $\ES_{p}$ by a factor $0.7$ and estimating $\VaR_{p}$ as $\VaR_{p'}$ for random variables with GPD-distributed excesses]{Minimum sample sizes needed to achieve certain powers for forecasts 
    underestimating $\ES_{p}$ by a factor $0.7$ and estimating $\VaR_{p}$ as $\VaR_{p'}$ for levels $p \in \{0.99,0.999\}$ and random variables with GPD(1,0.2)-distributed excesses, for different thresholds $\tau$ and levels $p'$. All results are based on $1{,}000$ Monte Carlo simulations each and a maximum sample size of $6{,}000$.}
    \label{TabelleESconstantGPD99}
\end{table}
\setlength{\tabcolsep}{4pt}


As detailed in Lemma \ref{SSBBBestCase}, the SSBBs are derived from a ``best-case scenario'', i.\ e., by assuming that the distribution of the loss variables is such that the e-power of the considered e-backtesting procedure is maximized among distributions with common value-at-risk and expected shortfall. One might also ask whether there are minimum sample sizes that guarantee a prescribed rejection probability even under worst case distributions. However, this is not possible, due to the following lemma.
\begin{lemma}\label{SampleSizeWorstCase}
  Let $c, \tau>1$, $n\in \N$, $p,q\in (0,1)$ and $(z_t,r_t)_{t\in \N}$ be a sequence of $(\VaR_p,\ES_p)$ forecasts such that $r_t\geq z_t$ for all $t\in \N$. Then, there exists a sequence of loss variables $(L_t)_{t\in \N}$ with $\ES_p(L_t)\geq c r_t$ for all $t\in \N$ such that
\begin{align*}
  \mathbb{P}\left(\max_{1\leq t\leq n}M_t(\mathbf{\lambda})>\tau\right)<q \, ,
\end{align*}  
for every betting process $\lambda=(\lambda_t)_{t\in \N}$, where $(M_t)_{t\in \N}$ is defined as in Section \ref{sec33}.
\end{lemma}
\begin{proof}
  For $t\in \N$, let $L_t$ such that $\mathbb{P}(L_t=z_t)=\sqrt[n+1]{q}$ and 
  \begin{align*}
    \mathbb{P}\left(L_t=\frac{c r_t-\sqrt[n+1]{q}z_t}{1-\sqrt[n+1]{q}}\right)=1-\sqrt[n+1]{q} \, .
  \end{align*}
  Then, $\ES_p(L_t)\geq \mathbb{E}[L_t]=c r_t$, but $\mathbb{P}(\forall 0\leq t\leq n: e_p^{\ES}(L_t,r_t,z_t)=0)=q^{n/(n+1)}>q$, where $e_p^{\ES}$ is the backtest e-statistic as defined in Section \ref{sec3}. Since under this event all $M_t(\mathbf{\lambda})$ must be smaller than or equal to $1$ for all $t\leq n$ regardless of the chosen betting process $\mathbf{\lambda}$, the result follows immediately.
\end{proof}

\section{Conclusion and Outlook}\label{sec8} 
In this work, we have analyzed the e-backtesting procedure both from the theoretical and from the practical perspective. 

In terms of theoretical aspects, we have generalized the construction of monotone backtest e-statistics to the class of Bayes pairs, and we have provided some results on the existence and the non-existence of such statistics. In future research, the relationship of the value-at-risk and the expected shortfall constituting a Bayes pair might also lead to a method for estimating the value-at-risk and the expected shortfall with neural networks, namely by using the involved loss function as a loss function for the network. This would, of course, also be applicable to other Bayes pairs as well. 

In terms of practical aspects, we have explored the problems of choosing the betting process as well as the significance threshold, and we have considered the planning of appropriate sample sizes. While the exact numbers provided in the various tables in Sections \ref{sec61} - \ref{sec63} refer to specific scenarios and are not applicable per se to other ones, they exhibit clear and interpretable patterns. We have analyzed these patterns and deduced from this analysis recommendations which are likely to generalize to other, related scenarios. For instance, under the scope of Section \ref{sec61}, it is near at hand to assume that at least the orders of magnitude of the required sample sizes will be transferable to other elliptical distributions, too, and analogously this will most likely be true for AR-GARCH processes of higher orders under the scope of Section \ref{sec63}. For sample size determination when backtesting the expected shortfall, we have furthermore worked out a generally applicable method (termed SSBBs) to come up with reasonable sample size bounds.  

Still, these topics would likely benefit from further investigation: For example, while we discussed in Theorem \ref{AsymOptiofAlgs} certain scenarios under which the GREE, GREL, or GREM betting processes are asymptotically optimal, these scenarios are highly theoretical and unlikely to apply well to practical applications. In particular, assumptions on the e-statistics or the loss variables being i.i.d. over time are highly unlikely to hold in practice. It thus seems reasonable to search for betting processes that are asymptotically optimal under more general and/or practically more relevant scenarios. 
One type of betting process we suggest considering for this is constituted by the finite-horizon counterparts of the aforementioned processes. While we could not make assertions about these processes being asymptotically optimal, they are likely to better detect underestimations for scenarios where high risks and underestimations of these risks occur clustered in time. Since this behavior is common for many models considered in (financial) risk management, e.\ g., for ARIMA and GARCH time series, these processes are likely to perform well in practical applications.

In the computer simulations reported in Section \ref{sec6}, we were for most considered scenarios only able to consider sample sizes up to $6{,}000$. As mentioned when discussing the specific scenarios, this often does not suffice to detect less severe underestimations of the expected shortfall, especially at higher levels $p$. In these situations, we recommend using the corresponding SSBBs to approximate the sample size needed. One might also extend our simulations to higher maximum sample sizes in order to verify that our statements made for more severe underestimations also carry over to these scenarios. Several R worksheets (cf.\ \cite{rbase}) for this are available from the authors upon request.

While this work, as well as the one by \cite{Hauptquelle}, has focused on financial risk management, one might also apply the e-backtesting method to other fields of risk management. Especially, the areas of environmental risk management and meteorology offer possible further applications, as climate models are also often backtested. However, in this context backtesting is oftentimes referred to as hindcasting; see \textcite{huijnen2012hindcast}, for example. 


\nocite{artzner1999coherent}
\nocite{Shaf2025}
\printbibliography

@article{Hauptquelle, 
title = {E-backtesting},
 author = {Qiuqi Wang and Ruodu Wang and Johanna Ziegel}, 
 publisher = {Institute for Operations Research and the Management Sciences (INFORMS)}, 
 journal = {Management Science}, 
 year = {2026}, 
 volume={72},
number={6},
pages={4952-4973} 
}

@article{RamdasBuch,
year = {2025},
volume = {1},
journal = {Foundations and Trends in Statistics},
title = {Hypothesis Testing with E-values},
number = {1-2},
pages = {1-390},
author = {Aaditya Ramdas and Ruodu Wang}
}

@book{Ville1939,
  title={Etude critique de la notion de collectif},
  author={Ville, Jean},
  year={1939},
  publisher={Gauthier-Villars Paris}
}

@unpublished{Shaf2025,
  author = {Shafer, Glenn},
  title  = {An Introduction to Game-Theoretic Statistics},
  month  = {April},
  year   = {2025},
  note   = {Unpublished manuscript}
}

@article{ROCKAFELLAR20021443,
title = {Conditional value-at-risk for general loss distributions},
journal = {Journal of Banking \& Finance},
volume = {26},
number = {7},
pages = {1443-1471},
year = {2002},
author = {R.Tyrrell Rockafellar and Stanislav Uryasev}
}

@article{7852039d69554002af6904518991286e,
title = {Safe Testing},
author = {Peter Gr{\"u}nwald and {de Heide}, Rianne and Wouter Koolen},
publisher = {The Royal Statistical Society},
year = {2024},
journal = {Journal of the Royal Statistical Society. Series B: Statistical Methodology},
volume = {86},
number = {5},
pages = {1091--1128},
publisher = {Oxford University Press}
}

@article{VovkWangMerging,
author = {Vovk, Vladimir and Wang, Ruodu},
year = {2024},
pages = {1185-1205},
title = {Merging sequential e-values via martingales},
volume = {18},
journal = {Electronic Journal of Statistics}
}

@book{mcneil2015quantitative,
  title={Quantitative {R}isk {M}anagement: 
  {C}oncepts, {T}echniques and {T}ools - {R}evised {E}dition},
  author={McNeil, Alexander J and Frey, R{\"u}diger and Embrechts, Paul},
  year={2015},
  publisher={Princeton {U}niversity {P}ress, Princeton, NJ}
}

@book{mcneil2005quantitative,
  title={Quantitative {R}isk {M}anagement: 
  {C}oncepts, {T}echniques and {T}ools},
  author={McNeil, Alexander J and Frey, R{\"u}diger and Embrechts, Paul},
  year={2005},
  publisher={Princeton {U}niversity {P}ress, Princeton, NJ}
}

@techreport{BCBS2019d457,
title = {Minimum Capital Requirements for Market Risk},
author = {{Basel Committee on Banking Supervision}},
institution = {Bank for International Settlements},
year = {2019},
month = jan,
note = {Revised February 2019},
}

@techreport{BCBS2013,
title = {Consultative document:
Fundamental review of the 
trading book },
author = {{Basel Committee on Banking Supervision}},
institution = {Bank for International Settlements},
year = {2013},
month = oct
}

@article{artzner1999coherent,
  title={Coherent measures of risk},
  author={Artzner, Philippe and Delbaen, Freddy and Eber, Jean-Marc and Heath, David},
  journal={Mathematical {F}inance},
  volume={9},
  number={3},
  pages={203--228},
  year={1999},
  publisher={Wiley Online Library}
}

@article{huijnen2012hindcast,
  title={Hindcast experiments of tropospheric composition during the summer 2010 fires over western Russia},
  author={Huijnen, V and Flemming, J and Kaiser, JW and Inness, A and Leit{\~a}o, J and Heil, A and Eskes, HJ and Schultz, MG and Benedetti, A and Hadji-Lazaro, Juliette and others},
  journal={Atmospheric Chemistry and Physics},
  volume={12},
  number={9},
  pages={4341--4364},
  year={2012},
  publisher={Copernicus Publications G{\"o}ttingen, Germany}
}

@article{https://doi.org/10.1111/mafi.12313,
author = {Embrechts, Paul and Mao, Tiantian and Wang, Qiuqi and Wang, Ruodu},
title = {Bayes risk, elicitability, and the Expected Shortfall},
journal = {Mathematical Finance},
volume = {31},
number = {4},
pages = {1190-1217},
year = {2021}
}

@article{follmer2010convex,
  title={Convex and coherent risk measures},
  author={F{\"o}llmer, Hans and Schied, Alexander},
  journal={Encyclopedia of Quantitative Finance},
  pages={355--363},
  year={2010},
  publisher={John Wiley \& Sons Hoboken}
}

@article{detlefsen2005conditional,
  title={Conditional and dynamic convex risk measures},
  author={Detlefsen, Kai and Scandolo, Giacomo},
  journal={Finance and {S}tochastics},
  volume={9},
  number={4},
  pages={539--561},
  year={2005},
  publisher={Springer}
}

@book{FolStoc2025, 
title = {Stochastic finance: {A}n introduction in discrete time},
author = {Hans Föllmer and Alexander Schied},
series = {De Gruyter {G}raduate},
edition = {5th revised and extended edition},
publisher = {De Gruyter},
address = {Berlin},
year = {2025}
}

@article{ShafThresh,
author = {Shafer, Glenn},
title = {Testing by betting: A strategy for statistical and scientific communication},
journal = {Journal of the Royal Statistical Society: Series A (Statistics in Society)},
volume = {184},
number = {2},
pages = {407-431},
year = {2021}
}

@article{norton2021calculating,
  title={Calculating {C}{V}a{R} and b{P}{O}{E} for common probability distributions with application to portfolio optimization and density estimation},
  author={Norton, Matthew and Khokhlov, Valentyn and Uryasev, Stan},
  journal={Annals of Operations Research},
  volume={299},
  number={1},
  pages={1281--1315},
  year={2021},
  publisher={Springer}
}

@article{BlierWong2026,
 author = {Blier-Wong, Christopher and Wang, Ruodu},
 title = {Improved thresholds for e-values},
 fjournal = {The Annals of Statistics},
 journal = {Ann. Stat.},
 volume = {54},
 number = {4},
 pages = {1819--1842},
 year = {2026},
 doi = {10.1214/26-AOS2626},
 zbMATH = {8244284}
}

@article{reverse-stress-testing,
 author = {{von Schroeder}, Jonathan and Dickhaus, Thorsten and Bodnar, Taras},
 title = {Reverse stress testing in skew-elliptical models},
 fjournal = {Theory of Probability and Mathematical Statistics},
 journal = {Theory Probab. Math. Stat.},
 volume = {109},
 pages = {101--127},
 year = {2023},
 doi = {10.1090/tpms/1199},
 zbMATH = {7748862}
}

@misc{rbase,
author="{R Development Core Team}", 
title="R: A Language and Environment for Statistical Computing.",
year="2026",
howpublished="Available from: \url{http://www.R-project.org}."
}
\begin{appendices}
\section{Characterization of backtest e-statistics}\label{sec7}
In this section, we take a more general look at backtest e-statistics and focus on assertions concerning the general (non)-existence of backtest e-statistics in general as well as necessary properties of backtest e-statistics for specific risk measures.

\subsection{Existence of backtest e-statistics}
First, we focus on conditions a risk measure $\rho$ has to fulfill in order for a backtest e-statistic to exist. We start with some definitions:
\begin{definition}
  Let $\mathcal{M}$ be a set of distribution functions and $\rho:\mathcal{M}\to \R$. 
  \begin{enumerate}[(i)]
    \item 
    If $\rho(F)\leq \rho(G)$ for all $F,G\in \mathcal{M}$ with $F\leq G$, we call $\rho$ \textbf{monotone}.
    \item  
    If, for all $F\in \mathcal{M}$ and $r>\rho(F)$, there exists $G\in \mathcal{M}$ with $G\geq F$ and $\rho(G)=r$, we call $\rho$ \textbf{uncapped}.
    \item 
    If $\mathcal{M}$ is convex and $\rho(\lambda F+(1-\lambda)G)\leq \max\{\rho(F),\rho(G)\}$ for all $F,G\in \mathcal{M}$ and $\lambda \in [0,1]$, we call $\rho$ \textbf{quasi-convex}.
    \item 
    If $-\rho$ is quasi-convex, we call $\rho$ \textbf{quasi-concave}.
    \item 
    If $\rho$ is both quasi-convex and quasi-concave, we call $\rho$ \textbf{quasi-linear}.
  \end{enumerate}
\end{definition}
The following two lemmata connect the previous definitions nicely with our theory of backtest e-statistics:
\begin{lemma}\label{BacktestIsQuasiConvex}
  Let $\mathcal{M}$ be a convex set of distribution functions and $\rho:\mathcal{M}\to \R$. If there exists a $\mathcal{M}$-backtest e-statistic for $\rho$, then $\rho$ is quasi-convex. 
\end{lemma}
\begin{proof}
  Let $e:\R\times\rho(M)\to \R$ be a  $\mathcal{M}$-backtest e-statistic for $\rho$. Further, let $r\in \R, \lambda \in [0,1]$ and $F,G\in \mathcal{M}$ with $\rho(F)=r$ and $\rho(G)\leq r$. Since $e$ is a one-sided e-statistic, we have
  \begin{align*}
    &\int_\R e(x,r)dF(x)\leq 1\textrm{ and } \int_\R e(x,r)dG(x)\leq 1 \\ 
    \implies &\int_\R e(x,r)d(\lambda F+(1-\lambda)G)(x)\leq 1 \, .
  \end{align*}
  Since $e$ is a backtest e-statistic, this yields $\rho(\lambda F+(1-\lambda)G)\leq r=\max\{\rho(F),\rho(G)\}$.
\end{proof}
\begin{bemerkung}
Regarding the previous lemma, in the original paper regarding E-Backtesting by Wang et al, $\rho$ was also required to be monotone and uncapped (see Proposition 3 in \cite{Hauptquelle}). As shown in our proof, these assumptions are unnecessary.
\end{bemerkung}
The next lemma concerns monotone backtest e-statistics:
\begin{lemma}
   Let $\mathcal{M}$ be a convex set of distribution functions and $\rho:\mathcal{M}\to \R$. If there exists a monotone $\mathcal{M}$-backtest e-statistic for $\rho$, then $\rho$ is quasi-linear. 
\end{lemma}
\begin{proof}
  Quasi-Convexity of $\rho$ follows from Lemma \ref{BacktestIsQuasiConvex}. For Quasi-Concavity, let $e:\R\times\rho(M)\to \R$ be a monotone $\mathcal{M}$-backtest e-statistic for $\rho$. Further, let $r\in \R,\lambda \in [0,1]$ and $F,G\in \mathcal{M}$ with $\rho(F),\rho(G)\geq r$. Assume now that $\rho(\lambda F+(1-\lambda)G):=q<r$. Then, there exists $\epsilon>0$ such that $q+\epsilon<r$. Since monotonicity of $e$ implies that $e(x,r)$ is decreasing in $r$, we have
  \begin{align*}
    &\int_\R e(x,q+\epsilon)d(\lambda F+(1-\lambda)G)(x)\leq \int_\R e(x,q)d(\lambda F+(1-\lambda)G)(x)\leq 1  \\ 
    \textrm{and } & \int_\R e(x,q+\epsilon)dF(x)>1 \textrm{ and } \int_\R e(x,q+\epsilon)dG(x)>1 \, ,
  \end{align*}
  which constitutes a contradiction. Thus, $\rho(\lambda F+(1-\lambda)G)\geq r$, and as this holds for all $r\in\R$ lower bounding both $\rho(F)$and $\rho(G)$, this yields $-\rho(\lambda F+(1-\lambda)G)\leq - \min\{\rho(F),\rho(G)\}=\max\{-\rho(F),-\rho(G)\}$. 
\end{proof}

\subsection{Backtest e-statistics for specific risk measures}
Next, we want to characterize possible one-sided (monotone) backtest e-statistics for common risk measures. We will notice that these e-statistics often have to be upper bounded by some specific e-statistic. Since a sequential test based on the latter e-statistic will trivially have universally higher power than the former, these results substantially limit the possibilities for one to derive different e-statistics that do not perform uniformly worse than another one. \\ 

We start our analysis by looking at the expected value as a risk measure (see Proposition 5 in \cite{Hauptquelle}):
\begin{theorem}\label{CharacMean}
  Let $\alpha \in \R$ and $\mathcal{P}:=\{ F \in \mathcal{M}_1| \operatorname{supp}(X)\in [\alpha,\infty)\}$. Further, let $\rho: \mathcal{P}\to \R$ with $\rho(F)=\mathbb{E}[X], X\sim F$ and $e:\R \times [\alpha,\infty)\to [0,\infty)$ be a $\mathcal{P}$-one-sided e-statistic for $\rho$. 
  \begin{enumerate}[(i)]
    \item 
    There exists some function $h:[\alpha, \infty)\to[0,1]$ such that 
    \begin{align*}
      e(x,r)\leq e'(x,r):=\left\{\begin{array}{ll}
        1+h(r)\frac{x-r}{r-\alpha}, & x\geq \alpha \\ 
        \infty , & x<\alpha 
      \end{array}\right. \qquad \forall x\in \R, r\in [\alpha,\infty) \, .
    \end{align*}
    Further, $e'$ is a  $\mathcal{P}$-one-sided e-statistic for $\rho$. 
    \item 
    $e'$ is a backtest e-statistic for $\rho$ if and only if $h(r)>0$ for all $r\geq \alpha$.
    \item 
    Under the conditions of (ii), $e'$ is a monotone backtest e-statistic for $\rho$ if and only if $h$ and $r\to (r-a)/h(r)$ are increasing.
  \end{enumerate}
\end{theorem}
The proof of this theorem largely relies on the following quite general lemma:
\begin{lemma}\label{HilfsLemmaMean}
  Let $r\geq 0$ and $g:\R\to [0,\infty)$ such that, for all non-negative random variables $X$ with $\mathbb{E}[X]\leq r$, it holds $\mathbb{E}[g(X)]\leq 1$. Then, there exists $h\in [0,1]$ such that $g(x)\leq 1+h\frac{x-r}{r}$ for all $x\geq 0$.
\end{lemma}
\begin{proof}
  First, assume $r=1$. Then, $g(y)\leq y$ for all $y>1$, since otherwise the random variable $X$ with $\mathbb{P}(X=y)=1/y$ and $\mathbb{P}(X=0)=1-1/y$ contradicts the assumptions. Moreover, we have $g(y)\leq 1$ for $y\leq 1$ by using constant random variables. Now, define the following sets:
  \begin{align*}
    \Lambda_0&:=\{\lambda \in [0,1]:\exists y\in (1,\infty):g(y)>1+\lambda y-\lambda\} \\ 
    \Lambda_1&:=\{\lambda \in [0,1]:\exists y\in [0,1):g(y)>1+\lambda y-\lambda\} \, .
  \end{align*}
  Next, assume the assertion does not hold. Then, $\Lambda_0\cup \Lambda_1=[0,1]$ since $g(1)\leq 1$ by our previous argument. Since further $g(y)\leq 1$ for $y<1$, we have $0\in \Lambda_0$ and since $g(y)\leq y$ for all $y>1$, we have $1\in \Lambda_1$. Because $1+\lambda y-\lambda$ is strictly increasing (decreasing) in $\lambda$ for all $y>1$ ($y<1$), $\Lambda_0$ and $\Lambda_1$ are intervals. Assume now that $\sup \Lambda_0\neq 1$ and $\sup \Lambda_0 \in \Lambda_0$. Then, there exists $y_0>1$ such that $g(y_0)>1+\sup \Lambda_0 y_0 - \sup \Lambda_0$. Since $1+\lambda y_0-\lambda$ is continuous in $\lambda$, there exists $\epsilon>0$ such that $g(y_0)>1+\lambda y_0 - \lambda$ for all $\lambda<\sup \Lambda_0+\epsilon$. Since this contradicts $\sup \Lambda_0$ being the supremum of $\Lambda_0$, we derive $\sup \Lambda_0 \not \in \Lambda_0$. If $\sup \Lambda_0=1$, this also holds by our previous argument. Similarly, $\inf \Lambda_1 \not \in \Lambda_1$. Since $\Lambda_0\cup \Lambda_1=[0,1]$, this implies $\inf \Lambda_1 \leq \sup \Lambda_0$ and since $\Lambda_0$ and $\Lambda_1$ are intervals, we derive $\Lambda_0\cap\Lambda_1\neq \emptyset$.

  Let $\lambda_0\in \Lambda_0\cap\Lambda_1$ and $x_0\in [0,1)$ and $y_0\in (1,\infty)$ such that $g(x_0)>1+\lambda_0 x_0-\lambda_0$ and $g(y_0)>1+\lambda_0 y_0-\lambda_0$. Then, define the random variable $X$ by $\mathbb{P}(X=y_0)=\frac{1-x_0}{y_0-x_0}$ and $\mathbb{P}(X=x_0)=\frac{y_0-1}{y_0-x_0}$. Then, we have
  \begin{align*}
    \mathbb{E}[X]=y_0\frac{1-x_0}{y_0-x_0}+x_0 \frac{y_0-1}{y_0-x_0}=\frac{y_0-x_0}{y_0-x_0}=1 \, ,
  \end{align*}
  but also
  \begin{align*}
    \mathbb{E}[g(X)]&=g(y_0)\frac{1-x_0}{y_0-x_0}+g(x_0) \frac{y_0-1}{y_0-x_0}\\ 
    &>(1+\lambda_0 y_0 -\lambda_0)\frac{1-x_0}{y_0-x_0}+(1+\lambda_0 x_0 -\lambda_0)\frac{y_0-1}{y_0-x_0} \\ 
    &=(1-\lambda_0)+\lambda_0 \frac{y_0-y_0x_0+x_0y_0-x_0}{y_0-x_0}=1 \, ,
  \end{align*}
  which contradicts our assumptions. \\ 

  For $r> 0$ notice that for any non-negative random variable $X$ with $\mathbb{E}[X]\leq 1$ we have $\mathbb{E}[r X]\leq r$ and thus $\mathbb{E}[g(rX)]\leq 1$. By our previous argument, there exists $h\in [0,1]$ such that
  \begin{align*}
    g(rx)\leq 1+h x-h \qquad \forall x\geq 0 \, .
  \end{align*}
  Scaling by $\frac{1}{r}$ then yields $g(x)\leq 1+h \frac{x}{r}-h=1+h\frac{x-r}{r}$ \, . \\ 

  Lastly, in the case of $r=0$, notice that $g(0)\leq 1$ by using $X=0$. Since the right-hand side of the desired inequality is infinite for $x,h>0$ and at least $1$ if $x=0$, this implies the assertion.
\end{proof}
Now, we can prove Theorem \ref{CharacMean}:
\begin{proof}[Proof of Theorem \ref{CharacMean}]
  Let $r\geq \alpha$ and $X$ be a non-negative random variable with $\mathbb{E}[X]\leq r-\alpha$. Then, $X+\alpha\in \mathcal{P}$ and $\mathbb{E}[X+\alpha]\leq r$. Since $e$ is a $\mathcal{P}$-one-sided e-statistic, this implies $\mathbb{E}[e(X+\alpha,r)]\leq 1$. By Lemma \ref{HilfsLemmaMean}, there exists $h(r)\in [0,1]$ such that 
  \begin{align*}
    e(x+\alpha,r)\leq 1+h(r)\frac{x-r+\alpha}{r-\alpha} \qquad \forall x\geq 0 \, .
  \end{align*}
  Shifting $x$ by $-\alpha$ then yields the desired upper bound for $e$. \\ 

  To see that $e'$ is a $\mathcal{P}$-one-sided e-statistic, notice that $e'$ is non-negative. Further, for any random variable $X$ with $\operatorname{supp}(X)\in [\alpha,\infty)$ and $\mathbb{E}[X]\leq r$, we have
  \begin{align*}
    \mathbb{E}[e'(X,r)]=1+h(r)\frac{\mathbb{E}[X]-r}{r-\alpha}\leq 1 \, .
  \end{align*} 
  
  For (ii), let  $X$ be a random variable with $\operatorname{supp}(X)\in [\alpha,\infty)$ and $\mathbb{E}[X]> r$. Similarly to the previous argument, we have 
  \begin{align*}
    \mathbb{E}[e'(X,r)]=1+h(r)\frac{\mathbb{E}[X]-r}{r-\alpha}> 1 \, ,
  \end{align*} 
  if and only if $h(r)>0$, which is equivalent to $e'$ being a backtest e-statistic. \\ 

  For (iii), let $\alpha\leq r< r'$ and let first $h$ and $r\to (r-a)/h(r)$ be increasing . For $x\geq r$, we then have
  \begin{align*}
    e'(x,r)=1+h(r)\frac{x-r}{r-\alpha} \geq 1+h(r')\frac{x-r}{r'-\alpha}\geq 1+h(r')\frac{x-r'}{r'-\alpha}=e'(x,r') \, ,
  \end{align*}
  where the second step is due to $r\to (r-a)/h(r)$ being increasing. For $x\leq r$, we have
  \begin{align*}
    e'(x,r)=1+h(r)\frac{x-r}{r-\alpha} \geq 1+h(r')\frac{x-r}{r-\alpha}\geq 1+h(r')\frac{x-r'}{r'-\alpha}=e'(x,r') \, ,
  \end{align*}
  where the second step is due to $h$ being increasing, and the third step is due to $r\to \frac{x-r}{r-\alpha}$ (always) being decreasing for $x\geq \alpha$. Thus, $e'$ is monotone. For the other direction, let $e'$ be monotone. Then, we derive
  \begin{align*}
    h(r)\frac{x-r}{r-\alpha}\geq h(r')\frac{x-r'}{r'-\alpha} \qquad \forall x\geq \alpha \, .
  \end{align*}
  This especially holds for $x=\alpha$ which yields $-h(r)\geq -h(r')\implies h(r)\leq h(r')$. Thus, $h$ is increasing. Now, assume $ \frac{h(r)}{r-\alpha}<\frac{h(r')}{r'-\alpha}$. Then, there exists an $0<\epsilon<0.5$ such that $ \frac{h(r)}{r-\alpha}<(1-\epsilon)\frac{h(r')}{r'-\alpha}$. Setting $x:=\frac{r'-(1-\epsilon)r}{\epsilon}$ and multiplying by $(x-r)$ then yields $ (x-r)\frac{h(r)}{r-\alpha}<(x-r')\frac{h(r')}{r'-\alpha}$, which contradicts the monotonicity of $e'$. Thus, $ \frac{h(r)}{r-\alpha}\geq\frac{h(r')}{r'-\alpha}$, which implies $r\to \frac{r-\alpha}{h(r)}$ being increasing.
\end{proof}

Next, we look at the variance (see Proposition 6 in \cite{Hauptquelle}):
\begin{theorem}\label{CharacVar}
  Let $\psi=(\rho,\phi):= \mathcal{M}_2\to[0,\infty) \times \R$ with $\rho(F):=\operatorname{Var}(X),\phi(F):=\mathbb{E}[X], X\sim F$ and $e:\R \times [0,\infty)\times \R\to [0,\infty)$ be a $\mathcal{M}_2$-one-sided e-statistic for $\psi$. 
  \begin{enumerate}[(i)]
    \item 
    There exists some function $h:[0, \infty)\times \R\to[0,1]$ such that 
    \begin{align*}
      e(x,r,z)\leq e'(x,r,z):=1+h(r,z)\frac{(x-z)^2-r}{r}\qquad \forall x,z\in \R, r\in [0,\infty) \, .
    \end{align*}
    Further, $e'$ is a  $\mathcal{M}_2$-one-sided e-statistic for $\psi$. 
    \item 
    $e'$ is a backtest e-statistic for $\psi$ if and only if $h(r,z)>0$ for all $r\geq 0$ and $z\in \R$.
    \item 
    Under the conditions of (ii), $e'$ is a monotone backtest e-statistic for $\rho$ if and only if $r\to h(r,z)$ and $r\to r/h(r,z)$ are increasing for all $z\in \R$.
  \end{enumerate}
\end{theorem}
The proof of the aforementioned theorem relies on the following technical lemma:
\begin{lemma}\label{HilfsLemmaVar}
  Let $r\geq 0$ and $g:\R\to [0,\infty)$ such that, for all random variables $X$ with $\mathbb{E}[X]=0$ and $\operatorname{Var}(X)\leq r$, it holds $\mathbb{E}[g(X)]\leq 1$. Then, there exists $h\in [0,1]$ such that $g(x)\leq 1+h\frac{x^2-r}{r}$ for all $x\geq 0$.
\end{lemma}
\begin{proof}
  First, let $r=1$. Next, let $y>1$ and assume $\frac{g(y)+g(-y)}{2}> y^2$. Consider the RV $X$ with $\mathbb{P}(X=y)=\mathbb{P}(X=-y)=\frac{y^{-2}}{2}$ and $\mathbb{P}(X=0)=1-y^{-2}$. Then, $\mathbb{E}[X]=0,\operatorname{Var}(X)=1$ but $\mathbb{E}[g(X)]=(1-y^{-2})g(0)+\frac{y^{-2}}{2}(g(y)+g(-y))> 1$, contradicting our assumptions. Thus, $\frac{g(y)+g(-y)}{2}\leq y^2$ for $y>1$. By the usage of random variables equally distributed on $y$ and $-y$, we further derive $\frac{g(y)+g(-y)}{2}\leq 1$ for $0\leq y\leq 1$. Next, assume the assertion does not hold. By a similar argumentation to the one in the proof of Lemma \ref{HilfsLemmaMean}, there exists $\lambda_0\in (0,1)$ such that $\frac{g(y_0)+g(-y_0)}{2}>1+\lambda_0 y_0^2-\lambda_0$ for some $y_0>1$ and $\frac{g(x_0)+g(-x_0)}{2}>1+\lambda_0 x_0^2-\lambda_0$ for some $x_0<1$. \\ 

  Next, let $X$ be an RV with $\mathbb{P}(X=y_0)=\mathbb{P}(X=-y_0)=\frac{1-x_0^2}{2(y_0^2-x_0^2)}$ and $\mathbb{P}(X=x_0)=\mathbb{P}(X=-x_0)=\frac{y_0^2-1}{2(y_0^2-x_0^2)}$. Then, $\mathbb{E}[X]=0$ and
  \begin{align*}
    \operatorname{Var}(X)=y_0^2\frac{1-x_0^2}{y_0^2-x_0^2}+x_0^2 \frac{y_0^2-1}{y_0^2-x_0^2}=1 \, .
     \end{align*}
    Further, we have
\begin{align*}
  \mathbb{E}[g(X)]&=\frac{g(y_0)+g(-y_0)}{2}\frac{1-x_0^2}{y_0^2-x_0^2}+\frac{g(x_0)+g(-x_0)}{2}\frac{y_0^2-1}{y_0^2-x_0^2} \\ 
  &>(1+\lambda_0 y_0^2-\lambda_0)\frac{1-x_0^2}{y_0^2-x_0^2}+(1+\lambda_0 x_0^2-\lambda_0)\frac{y_0^2-1}{y_0^2-x_0^2} \\ 
  &=1-\lambda_0+\lambda_0\frac{y_0^2-x_0^2y_0^2+x_0^2y_0^2-x_0^2}{y_0^2-x_0^2}=1 \, ,
\end{align*}
contradicting our assumptions. Thus, the assertion holds. \\ 

For $r>0$ notice that for any centered RV $X$ with $\operatorname{Var}(X)\leq 1$, we have $\operatorname{Var}(\sqrt{r}X)\leq r$ and therefore $\mathbb{E}[g(\sqrt{r}X)]\leq 1$. By our previous argumentation, this implies the existence of an $h\in [0,1]$ such that
\begin{align*}
  g(\sqrt{r}x)\leq 1+ hx^2-h \qquad \forall x\geq 0\, .
\end{align*}
Scaling by $\frac{1}{\sqrt{r}}$ yields $g(x)\leq 1+h\frac{x^2}{r}-h=1+h\frac{x^2-r}{r}$.  \\ 

For $r=0$ notice again that the right-hand side of the desired inequality is infinite for $|x|,h>0$ and at least $1$ if $x=0$. Since $g(0)\leq 1$ by using $X=0$, the inequality holds for all $r>0$.
\end{proof}

Now to the proof of Theorem \ref{CharacVar}:
\begin{proof}[Proof of Theorem \ref{CharacVar}]
  Let $r\geq 0,z\in \R$ and $X$ be a centered random variable with $\operatorname{Var}(X)\leq r$. Then, $X+z$ has the same variance and mean $z$ and since $e$ is a $\mathcal{M}_2$-one-sided e-statistic, we have $\mathbb{E}[e(X+z,r,z)]\leq 1$. By Lemma \ref{HilfsLemmaVar}, there exists $h(r,z)\in [0,1]$ such that
  \begin{align*}
    e(x+z,r,z)\leq 1+h(r,z)\frac{x^2-r}{r} \qquad \forall z\in \R \, .
  \end{align*}
  Shifting $x$ by $-z$ then yields the desired upper bound. \\ 

  To show that $e'$ is a $\mathcal{M}_2$-one-sided e-statistic, notice again that $e'$ is non-negative. Further, let $X$ be an RV with $\mathbb{E}[X]=z$ and $\operatorname{Var}(X)\leq r$. Then, it holds
  \begin{align*}
    \mathbb{E}[e'(X,r,z)]=1+h(r,z)\frac{\operatorname{Var}(X)-r}{r}\leq 1 \, .
  \end{align*}

  For (ii), let $X$ be an RV with $\mathbb{E}[X]=z_0$ and $\operatorname{Var}(X)>r$. Then, we calculate
  \begin{align*}
    \mathbb{E}[e'(X,r,z)]=1+h(r,z)\frac{\mathbb{E}[(X-z)^2]-r}{r}\geq 1+h(r,z)\frac{\operatorname{Var}(X)-r}{r}\, ,
  \end{align*}
  where the second step is due to $z_0=\argmin_{z\in \R}\mathbb{E}[(X-z)^2]$ which also implies that this inequality becomes an equality for $z=z_0$. Then, notice that $1+h(r,z)\frac{\operatorname{Var}(X)-r}{r}>1$ if and only if $h(r,z)>0$. \\  

  For (iii), let $z\in \R, 0\leq r<r'$ and first let $r\to h(r,z)$ and $r\to r/h(r,z)$ be increasing. For $(x-z)^2\geq r$ we then have
  \begin{align*}
    e'(x,r,z)&=1+h(r,z)\frac{(x-z)^2-r}{r}\geq 1+h(r',z)\frac{(x-z)^2-r}{r'} \\ 
    &\geq  1+h(r',z)\frac{(x-z)^2-r'}{r'}=e'(x,r',z) \, ,
  \end{align*}
  where the second step is due to $r\to r/h(r,z)$ being increasing. For $(x-z)^2< r$, we have on the other hand
  \begin{align*}
    e'(x,r,z)&=1+h(r,z)\frac{(x-z)^2-r}{r}\geq 1+h(r',z)\frac{(x-z)^2-r}{r} \\ 
    &\geq 1+h(r',z)\frac{(x-z)^2-r'}{r'}=e'(x,r'z) \, ,
  \end{align*}
  where the second step is due to $r\to h(r,z)$ being increasing, and the third step is due to $r\to \frac{C-r}{r}$ being decreasing for any $C\geq 0$. Thus, $e'$ is monotone. For the back direction, let $e'$ be monotone. This yields
  \begin{align*}
    h(r,z)\frac{(x-z)^2-r}{r}\geq h(r',z)\frac{(x-z)^2-r}{r} \qquad \forall x \in \R \, .
  \end{align*}
  In particular, this holds for $x=z$, which directly implies that $r\to h(r,z)$ is increasing. Next, assume $\frac{h(r,z)}{r}<\frac{h(r',z)}{r'}$. Then, there exists $\epsilon>0$ such that $\frac{h(r,z)}{r}<(1-\epsilon)\frac{h(r',z)}{r'}$. Choosing $x=\sqrt{\frac{r'-(1-\epsilon)r}{\epsilon}}+z$ and multiplying both sides by $(x-z)^2-r$ yields $h(r,z)\frac{(x-z)^2-r}{r}<h(r',z)\frac{(x-z^2)-r'}{r'}$, contradicting monotonicity of $e'$. Therefore, $\frac{h(r,z)}{r}\geq \frac{h(r',z)}{r'}$, which implies $r\to r/h(r,z)$ being increasing.
\end{proof}

By a similar argument, Theorem \ref{CharacMean} can also be generalized to e-statistics for higher moments:

\begin{theorem}\label{CharacHigherMoments}
  Let $k\in \N$ and $\mathcal{P}:=\{ F \in \mathcal{M}_k| \operatorname{supp}(X)\in [0,\infty)\}$. Further, let $\rho: \mathcal{P}\to \R$ with $\rho(F)=\mathbb{E}[X^k], X\sim F$ and $e:\R \times [0,\infty)\to [0,\infty)$ be a $\mathcal{P}$-one-sided e-statistic for $\rho$. 
  \begin{enumerate}[(i)]
    \item 
    There exists some function $h:[\alpha, \infty)\to[0,1]$ such that 
    \begin{align*}
      e(x,r)\leq e'(x,r):=\left\{\begin{array}{ll}
        1+h(r)\frac{x^k-r}{r}, & x\geq 0 \\ 
        \infty , & x<\alpha 
      \end{array}\right. \qquad \forall x\in \R, r\in [\alpha,\infty) \, .
    \end{align*}
    Further, $e'$ is a  $\mathcal{P}$-one-sided e-statistic for $\rho$. 
    \item 
    $e'$ is a backtest e-statistic for $\rho$ if and only if $h(r)>0$ for all $r\geq \alpha$.
    \item 
    Under the conditions of (ii), $e'$ is a monotone backtest e-statistic for $\rho$ if and only if $h$ and $r\to r/h(r)$ are increasing.
  \end{enumerate}
\end{theorem}
The proof of this theorem is based on the following auxiliary lemma.
\begin{lemma}\label{HilfsLemmaMoments}
  Let $k\in \N,r\geq 0$ and $g:\R\to [0,\infty)$ such that, for all non-negative random variables $X$ with $\mathbb{E}[X^k]\leq r$, it holds $\mathbb{E}[g(X)]\leq 1$. Then, there exists $h\in [0,1]$ such that $g(x)\leq 1+h\frac{x^k-r}{r}$ for all $x\geq 0$.
\end{lemma}
\begin{proof}
  First, let $r=1$. Since for all $y>1$, we obtain a contradiction if $g(y)>y^k$ by considering the random variable $Y$ with $\mathbb{P}(Y=y)=\frac{1}{\sqrt[k]{y}}$ and $\mathbb{P}(Y=0)=1-\frac{1}{\sqrt[k]{y}}$, we infer $g(y)\leq y^k$ for all $y>1$. By considering constant random variables, we also infer $g(y)\leq 1$ for all $y\leq 1$, in particular $g(1)\leq 1$. \\ 

  Next, assume the assertion does not hold. Then, for all $\lambda\in [0,1)$, there exists some $y\in [0,1)$ or $y>1$ such that $g(y)>1+\lambda y^k-\lambda$. Now, define the following sets:
\begin{align*}
    \Lambda_0&:=\{\lambda \in [0,1]:\exists y\in (1,\infty):g(y)>1+\lambda y^k-\lambda\} \\ 
    \Lambda_1&:=\{\lambda \in [0,1]:\exists y\in [0,1):g(y)>1+\lambda y^k-\lambda\} \, .
  \end{align*}
  By a similar argument as in the proof of Lemma \ref{HilfsLemmaMean}, we get that $\Lambda_0\cap \Lambda_1\neq \emptyset$. \\ 

  Let $\lambda_0 \in \Lambda_0\cap \Lambda_1$ and $x_0\in [0,1)$ and $y_0\in (1,\infty)$ such that $g(x_0)>1+\lambda_0 x_0^k-\lambda_0$ and $g(y_0)>1+\lambda_0 y_0^k-\lambda_0$. Next, let $X$ be a random variable with $\mathbb{P}(X=y_0)=\frac{1-y_0^k}{y_0^k-x_0^k}$ and $\mathbb{P}(X=x_0)=\frac{y_0^k-1}{y_0-x_0}$. Then, $\mathbb{E}[X^k]=1$ but also 
\begin{align*}
    \mathbb{E}[g(X)]&=g(y_0)\frac{1-x_0^k}{y_0^k-x_0^k}+g(x_0) \frac{y_0^k-1}{y_0^k-x_0^k}\\ 
    &>(1+\lambda_0 y_0^k -\lambda_0)\frac{1-x_0^k}{y_0^k-x_0^k}+(1+\lambda_0 x_0^k -\lambda_0)\frac{y_0^k-1}{y_0^k-x_0^k} \\ 
    &=(1-\lambda_0)+\lambda_0 \frac{y_0^k-y_0^kx_0^k+x_0^ky_0^k-x_0^k}{y_0^k-x_0^k}=1 \, ,
  \end{align*}
  contradicting our assumptions. \\

Now, let $r>0$ and $X$ be a non-negative random variable with $\mathbb{E}[X^k]\leq 1$. Then, $\mathbb{E}[(\sqrt[k]{r}X)^k]\leq r$ and thus $\mathbb{E}[g(\sqrt[k]{r}(X))]\leq 1$. By our argumentation from before, there exists $h\in [0,1]$ such that
\begin{align*}
  g(\sqrt[k]{r}x)\leq 1+hx^k-h \qquad \forall x\geq 0 \, .
\end{align*}
Scaling by $\frac{1}{\sqrt[k]{r}}$ then yields the assertion. \\  

For $r=0$ notice again that $g(0)\leq 1$ by the usage of $X=0$. Since the right-hand side of the desired inequality is infinite for $x,h \neq 0$ and at least $1$ for $x=0$, the assertion is immediate in this case. 
\end{proof}
We can now prove Theorem \ref{CharacHigherMoments}:

\begin{proof}[Proof of Theorem \ref{CharacHigherMoments}]
  The existence of the proclaimed function $e'$ is immediate from Lemma \ref{HilfsLemmaMoments}. To see that $e'$ is a $\mathcal{P}$-one-sided e-statistic, let $r\geq 0$ and $X$ be a non-negative random variable such that $\mathbb{E}[X^k]\leq r$. Then, we have
  \begin{align*}
    \mathbb{E}[e'(X,r)]=1+h(r)\frac{\mathbb{E}[X^k]-r}{r}\leq 1 \, ,
  \end{align*}
  which proves that $e'$ is an e-statistic as it is also non-negative. \\ 

  Similarly, notice that for any $r\in \R$ and non-negative random variable $X$ with $\mathbb{E}[X^k]\geq r$, we have $\mathbb{E}[e'(X,r)]=1+h(r)\frac{\mathbb{E}[X^k]-r}{r}> 1$ exactly if $h(r)>0$, which shows assertion (ii). \\ 

For (iii), let $0\leq r<r'$ and assume first that $h$ and $r\to \frac{r}{h(r)}$ are increasing. Then, we have for $x^k\geq r$
\begin{align*}
  e'(x,r)= 1+h(r) \frac{x^k-r}{r}\geq 1+h(r') \frac{x^k-r}{r'}\geq 1+h(r') \frac{x^k-r'}{r'}=e'(x,r') \, ,
\end{align*}
where the second step is due to $r\to \frac{r}{h(r)}$ being increasing. For $x^k \leq r$, it holds
\begin{align*}
  e'(x,r)=1+h(r) \frac{x^k-r}{r} \geq 1+h(r') \frac{x^k-r}{r} \geq 1+h(r') \frac{x^k-r'}{r'}=e'(x,r') \, , 
\end{align*}
where the second step is due to $h$ being increasing, and the third step is due to $x^k\geq 0$. Therefore, $e'$ is monotone.\\ 

Finally, assume $e'$ to be monotone. This implies 
\begin{align*}
  h(r)\frac{x^k-r}{r}\geq h(r')\frac{x^k-r'}{r'} \qquad \forall x\geq 0 \, .
\end{align*}
In particular, this holds for $x=0$, which directly implies monotonicity of $h$. Next, assume $\frac{h(r)}{r}<\frac{h(r')}{r'}$. Since the inequality is strict, there exists $\epsilon>0$ such that  $\frac{h(r)}{r}<(1-\epsilon)\frac{h(r')}{r'}$. Letting $x:=\left(\frac{r'-(1-\epsilon)r}{\epsilon}\right)^{1/k}$ and multiplying both sides by $(x^k-r)$ then yields $(x^k-r)\frac{h(r)}{r}<(x^k-r')\frac{h(r')}{r'}$ which contradicts monotonicity of $e'$. Thus, $r\to \frac{r}{h(r)}$ must be increasing.
\end{proof}

\begin{bemerkung}
One difference between Theorems \ref{CharacMean} and \ref{CharacHigherMoments} is that the former can be used for e-statistics backtesting loss variables that are lower bounded by arbitrary constants $\alpha \in \R$ while the latter can only be used for e-statistics backtesting non-negative loss variables. This is due to the proof of Theorem \ref{CharacMean} utilizing the linearity of the expectation in order to convey the assertions for non-negative loss variables to loss variables lower bounded by an arbitrary constant $\alpha$ (see the beginning of the proof of Theorem \ref{CharacMean}). As the $k$-th moment operators are not linear for $k>1$, this argument does not apply when considering higher moments. However, this loss of generality can be remedied by considering the random variable $X^+:=\max\{0,X\}$ instead of an arbitrary random variable $X$.
\end{bemerkung}

Returning to the topic of backtesting the value-at-risk and the expected shortfall, we now address the value-at-risk (see Theorem 4 in \cite{Hauptquelle}):

\begin{theorem}\label{CharacVaR}
  Let $p\in (0,1)$ and $\rho:\mathcal{M}_0\to \R$ with $\rho(F):=\VaR_p(X),X\sim F$ and $e:\R^2 \to [0,\infty)$ be a $\mathcal{M}_0$-one-sided e-statistic for $\rho$.
  \begin{enumerate}[(i)]
    \item 
    There exists some function $h:\R\to [0,1]$ such that
    \begin{align*}
      e(x,r)\leq e'(x,r):=1+h(r)\frac{p-\mathbf{1}_{x\leq r}}{1-p} \qquad \forall x,r \in \R \, .
    \end{align*}
    Further, $e'$ is a $\mathcal{M}_0$-one-sided e-statistic for $\rho$.
    \item 
    $e'$ is a backtest e-statistic for $\rho$ if and only if $h(r)>0$ for all $r\in \R$.
    \item 
    Under the conditions of $(ii)$, $e'$ is a monotone backtest e-statistic for $\rho$ if and only if $h$ is constant.
  \end{enumerate} 
\end{theorem}
This theorem relies on the following lemma:
\begin{lemma}\label{LemCharVaR}
  Let $r\in \R$ and $p\in (0,1)$. Further, let $g:\R\to \R_{\geq 0}$ with $\mathbb{E}[g(L)]\leq 1$ for all random variables $L$ with $\VaR_p(L)\leq r$. Then, there exists $h\in [0,1]$ such that
  \begin{align*}
    g(x)\leq 1+ h \frac{p-\mathbf{1}_{x\leq r}}{1-p} \qquad \forall x\in \R \, .
  \end{align*}
\end{lemma}
\begin{proof}
 First, let $\epsilon>0$ and $\epsilon'\geq 0$ and consider a random variable $X$ with $\mathbb{P}(X=r-\epsilon')=p$ and $\mathbb{P}(X=r+\epsilon)=1-p$. Then, $\VaR_p(X)=r-\epsilon'\leq r$ and therefore $p g(r-\epsilon')+(1-p)g(r+\epsilon)\leq 1$. By rearranging this inequality and using that $\epsilon'$ is arbitrary, we derive
  \begin{align*}
    g(r+\epsilon) \leq \frac{1-\sup_{x\leq r}g(x)p}{1-p} \, .
  \end{align*}
  Rearranging this inequality again and using that $\epsilon$ is arbitrary as well, we infer
  \begin{align*}
    p \sup_{x\leq r} g(x)+(1-p)\sup_{x> r} g(x)\leq 1 \, .
  \end{align*}
  Defining $h:=1-\sup_{x\leq r} g(x)$, this is equivalent to $(1-h)p+(1-p) \sup_{x> r} g(x)\leq 1$. As $\sup_{x\leq r} g(x)\leq 1$ by considering constant random variables and our assumptions, we have $h\in[0,1]$. Solving for $\sup_{x> r} g(x)$ yields
  \begin{align*}
    \sup_{x> r} g(x)\leq \frac{1-(1-h)p}{1-p}=1+h\frac{p}{1-p}\, .
  \end{align*}
  This already proves the result, as for $x\leq r$, we have $g(x)\leq \sup_{x\leq r} g(x)=1-h$ by definition. 
\end{proof}

Now to the proof of Theorem \ref{CharacVaR}:
\begin{proof}[Proof of Theorem \ref{CharacVaR}]
  \begin{enumerate}[(i)]
    \item 
    The existence of $h:\R\to [0,1]$ such that $e(x,r)\leq e(x,r')$ is immediate from Lemma \ref{LemCharVaR}. To see that $e'$ is a one-sided e-statistic, notice first that $e'$ is non-negative. Next, let $r\in \R$ and $X$ be an RV with $\VaR_p(X)\leq r$. Then, we have
    \begin{align*}
      \mathbb{E}[e'(X,r)]\leq p (1-h(r))+(1-p)(1+h(r)\frac{p}{1-p})=1 \, .
    \end{align*}
    \item 
    First, consider $h>0$ and let $r\in \R$ and $X$ be an RV such that $\VaR_p(X)>r$. Then, since  $\mathbb{P}(X\leq r)< p$, we derive
    \begin{align*}
      \mathbb{E}[e'(X,r)]> p (1-h(r))+(1-p)\left(1+h(r)\frac{p}{1-p}\right)=1 \, .
    \end{align*}
    For the back direction, suppose there exists $r\in \R$ such that $h(r)=0$. Then, $\mathbb{E}[e'(X,r)]=1$ for all RVs $X$, contracting $e'$ being a backtest e-statistic.
    \item 
    Let $x\in \R$ and consider the function $f:\R\to \R, r\to e'(x,r)$. For $r> x$, $f$ is decreasing if and only if $h$ is increasing. For $r<x$, $f$ is decreasing if and only if $h$ is decreasing. Since $x$ is arbitrary, this shows that $e'$ is monotone exactly if $h$ is constant. 
  \end{enumerate}
\end{proof}
Finally, we look at the expected shortfall (see Theorem 5 in \cite{Hauptquelle}):
\begin{theorem}\label{CharacES}
  Let $p \in (0,1)$ and $\psi=(\rho,\phi):=\mathcal{M}_1\to \R\times \R$ with $\rho(F):=\ES_p(X),\phi(X):=\VaR_p(X), X\sim F$ and $e:\R^3\to [0,\infty)$ be a $\mathcal{M}_1$-one-sided e-statistic for $\psi$ such that $e(x,r,z)\leq e(z,r,z)$  for all $r,z \in \R$ and $x\leq z$. 
  \begin{enumerate}[(i)]
    \item 
    There exist some functions $h,k:\R^2\to[0,1]$ with $h(r,z)+k(r,z)\leq 1$ for all $r,z\in \R$ such that for all $x,r,z\in \R$
    \begin{align*}
      e(x,r,z)\leq e'(x,r,z):=\left\{\begin{array}{ll}
        1+h(r,z)\left(\frac{(x-z)_+}{(1-p)(r-z)}-1\right)+k(r,z)\frac{p-\textbf{1}_{x\leq z}}{1-p}, & z\leq r  \\ 
        \infty , & z> r 
      \end{array}\right.  \, .
    \end{align*}
    Further, $e'$ is a $\mathcal{M}_1$-one-sided e-statistic for $\psi$. 
    \item 
    $e'$ is a backtest e-statistic for $\psi$ if and only if $h(r,z)>0$ and $k(r,z)=0$ for all $r,z\in \R$.
    \item 
    Under the conditions of (ii), $e'$ is a monotone backtest e-statistic for $\psi$ if and only if $r\to h(r,z)$ and $r\to (r-z)/h(r,z)$ are increasing for all $z\in \R$ and $r>z$.
  \end{enumerate}
\end{theorem}
\begin{bemerkung}
The condition $e(x,r,z)\leq e(z,r,z)$  for all $r,z \in \R$ and $x\leq z$ in the previous theorem is especially fulfilled if $x\to e(x,r,z)$ is increasing for all $r,z\in \R$.
\end{bemerkung}
The proof of Theorem \ref{CharacES} uses the following Lemma:
\begin{lemma}\label{HilfsLemmaES}
  Let $r,z\in \R$ with $r\geq z$ and $p\in (0,1)$. Further, let $g:\R\to [0,\infty)$ be a function such that $g(x)\leq g(z)$ for all $x\leq z$ and $\mathbb{E}[g(X)]\leq 1$ for all random variables $X$ with $\VaR_p(X)=z$ and $\ES_p(X)\leq r$. Then, there exist $h,k\in [0,1]$ with $h+k\leq 1$ and 
  \begin{align*}
    g(x)\leq 1+h\left(\frac{(x-z)_+}{(1-p)(r-z)}-1\right)+k\frac{p-\textbf{1}_{x\leq z}}{1-p} \qquad \forall x\in \R \, .
  \end{align*} 
\end{lemma}
\begin{proof}
  First, consider $z=0$. Using a constant $0$ RV, we first get $\theta:=g(0)\in[0,1]$ from our assumptions. Next, let $X$ be a non-negative RV with $\mathbb{E}[X]\leq r$ and $B$ be a Bernoulli distributed RV with mean $1-p$ that is independent of $X$. Then, $\VaR_p(BX)=0$ and $\ES_p(BX)=\mathbb{E}[X]\leq r$. By our assumptions, this yields $\mathbb{E}[g(BX)]\leq 1$. Independence of $B$ and $X$ then implies
  \begin{align*}
    &p \theta +(1-p) \mathbb{E}[g(X)]&&\leq 1 \\ 
\implies & \mathbb{E}\left[\frac{1-p}{1-p\theta} g(X)\right]&&\leq 1 \, .
  \end{align*}
  By Lemma \ref{HilfsLemmaMean}, there exists $\lambda \in [0,1]$ such that, for all $x\geq 0$, we have
  \begin{align}
    &\frac{1-p}{1-p\theta} g(x)&&\leq 1+\lambda \frac{x-r}{r} \nonumber\\ 
    \implies & g(x)&&\leq \frac{1-p\theta}{1-p} \left(1+\lambda \frac{x-r}{r}\right) \, .\label{HilfGleichChar}
  \end{align}
  Using $g(x)\leq g(0)=\theta$ for $x\leq 0$ implies for all $x\in \R$
  \begin{align}
    g(x)&=g(x)\textbf{1}_{x\leq 0}+g(x)\textbf{1}_{x>0} \nonumber\\ 
    &\leq \theta \textbf{1}_{x\leq 0}+\frac{1-p\theta}{1-p} \left(1+\lambda \frac{x-r}{r}\right)\textbf{1}_{x>0} \nonumber\\ 
    &=\theta+\left(\frac{1-p\theta}{1-p} \left(1-\lambda+\lambda \frac{x}{r} \right)-\theta\right)\textbf{1}_{x>0}\nonumber \\ 
    &=\theta+\frac{1-p\theta}{1-p}\lambda \frac{x_+}{r}+\left(\frac{1-p\theta}{1-p} \left(1-\lambda \right)-\theta\right)\textbf{1}_{x>0} \nonumber\\ 
    &=\theta+(1-p\theta)\lambda\frac{x_+}{(1-p)r}+(1-\lambda-\theta+p\theta \lambda) \frac{\textbf{1}_{x>0}}{1-p} \, .\label{HilfX}
  \end{align} 
  Now, let $h:=(1-p\theta)\lambda$ and $k:=1-\lambda-\theta+p\theta \lambda$. Obviously $h\in [0,1]$. For $k$, notice $k=1-\lambda-\theta+p\theta \lambda=1-\theta-(1-p\theta)\lambda \leq 1$. Further, we have by inequality \eqref{HilfGleichChar}
  \begin{align*}
    &\theta=g(0)&&\leq \frac{1-p\theta}{1-p}(1-\lambda) \\ 
    \implies & 0 &&\leq 1-\lambda +p\theta \lambda -\theta=k \, .
  \end{align*}
  Also, notice $h+k=1-\theta \leq 1$. Plugging this into inequality \eqref{HilfX} yields
  \begin{align*}
    g(x)&\leq 1-k-h +h \frac{x_+}{(1-p)r} +k \frac{\textbf{1}_{x>0}}{1-p} \\ 
&=1+h\left(\frac{x_+}{(1-p)r}-1 \right )+k\frac{p-\textbf{1}_{x\leq 0}}{1-p} \, .
  \end{align*}
  For arbitrary $z\in \R$, let $X$ be an RV with $\VaR_p(X)=0$ and $\ES_p(X)=r-z$. Then, $\VaR_p(X+z)=z$ and $\ES_p(X+z)=r$ and thus  $\mathbb{E}[g(X+z)]\leq 1$ by assumption. Further, it holds $g(x+z)\leq g(0+z)=g(z)$ for all $x\leq 0$. By our previous argumentation, there exist $h,k\in [0,1]$ with $h+k\leq 1$ and 
  \begin{align*}
    g(x+z)\leq 1+h\left(\frac{x_+}{(1-p)(r-z)}-1 \right )+k\frac{p-\textbf{1}_{x\leq 0}}{1-p} \qquad \forall x\in \R\, .
  \end{align*}
  Shifting the argument by $-z$ then yields the desired inequality. 
\end{proof}
Now to the proof of Theorem \ref{CharacES}:
\begin{proof}[Proof of Theorem \ref{CharacES}]
  The boundedness of $e$ by $e'$ is immediate from Lemma \ref{HilfsLemmaES} since $e$ is a one-sided e-statistic for $\psi$. To see that $e'$ is a $\mathcal{M}_1$-one-sided e-statistic for $\psi$, notice first that for $z\leq r$ we have $\frac{(x-z)_+}{(1-p)(r-z)}-1\geq -1$ and $\frac{p-\textbf{1}_{x\leq z}}{1-p}\geq -1$ and thus non-negativity of $e'$ is implied by $h(r,z)+k(r,z)\leq 1$. Next, let $r,z\in \R, z\leq r$ and $X\sim F\in \mathcal{M}_1$ with $\VaR_p(X)=z$ and $\ES_p(X)\leq r$. Then, we have
  \begin{align*}
    \mathbb{E}[e(X,r,z)]&=1+h(r,z)\left(\frac{\mathbb{E}[(X-z)_+]}{(1-p)(r-z)}-1\right)+k(r,z)\frac{p-\mathbb{E}[\textbf{1}_{X\leq z}]}{1-p} \\ 
    &=1+h(r,z)\left(\frac{(1-p)(\ES_p[X]-z)}{(1-p)(r-z)}-1\right)+k(r,z)\frac{p-\mathbb{P}(X\leq \VaR_p(X))}{1-p} \\
    &\leq 1+h(r,z)\frac{r-z}{r-z}+0=1 \, .
  \end{align*} 

  For (ii), let first $h(r,z)>0$ and $k(r,z)=0$ for all $r,z\in \R$. Then, let $r,z\in \R$ and $X\sim F\in \mathcal{M}_1$ with $\ES_p(X)>r$. Obviously, $\mathbb{E}[e'(X,r,z)]=\infty>1$ if $z>r$. For $z\leq r$, notice that $\mathbb{E}[e'(X,r,z)]>1$ is equivalent to 
  \begin{align*}
    &\frac{\mathbb{E}[(X-z)_+]}{(1-p)(r-z)}-1&&>0  \\ 
    \iff &z+\frac{\mathbb{E}[(X-z)_+]}{1-p}&&>r  \, .
  \end{align*}
  By Lemma \ref{BayesPairESVaR}, the left-hand side is lower bounded by $\ES_p(X)>r$. Thus, $e'$ is a backtest e-statistic. \\  

  For the back direction, suppose that $e'$ is a backtest e-statistic for $\psi$. Let $r,z\in \R,z\leq r$ and consider for $\epsilon>0$ and $q\in [p,1)$ the random variable
  \begin{align*}
    X_{\epsilon,q}:=z+\left(\frac{(r-z)(1-p)}{1-q}+\epsilon\right)B_q \, ,
  \end{align*}
   where $B_q$ is a Bernoulli-distributed RV with mean $1-q$. Since $q\geq p$, we have $\VaR_p(X_{\epsilon,q})=z$ and 
  \begin{align*}
    \ES_p(X_{\epsilon,q})=\frac{q-p}{1-p}z+\frac{1-q}{1-p}\left(z+ \frac{(r-z)(1-p)}{1-q}+\epsilon\right)=r+\epsilon\frac{1-q}{1-p}>r \, ,
  \end{align*}
  for all $\epsilon>0$ and $q\in [p,1)$.
  Since $e'$ is a backtest e-statistic, this yields
  \begin{align*}
    1&<\mathbb{E}[e'(X_{\epsilon,q},r,z)] \\ 
    &= 1+h(r,z)\left(\frac{\mathbb{E}[(X_{\epsilon,q}-z)_+]}{(1-p)(r-z)}-1\right)+k(r,z)\frac{p-\mathbb{P}(X_{\epsilon,q}\leq z)}{1-p} \\ 
    &=1+h(r,z)\left(\frac{(1-p)(r-z)+(1-q)\epsilon}{(1-p)(r-z)}-1\right)+k(r,z)\frac{p-q}{1-p} \\ 
    &=1+h(r,z)\frac{(1-q)\epsilon}{(1-p)(r-z)}+k(r,z)\frac{p-q}{1-p} \, .
  \end{align*}
  Fixing $\epsilon>0$ and letting $q\to 1$ yields $k(r,z)=0$. Similarly, setting $q=p$ yields $h(r,z)>0$.  \\ 

  For (iii), let $z,x\in \R$ and assume first that $r\to h(r,z)$ and $r\to (r-z)/h(r,z)$ are increasing for all $r>z$. Let $r,r'\in \R, r\leq r'$ Obviously, $e'(x,r,z)\geq e'(x,r',z)$ if $r<z$. For $r\geq z$, we have 
\begin{align*}
  e'(x,r,z)=1+\frac{(x-z)_+h(r,z)}{(1-p)(r-z)}-h(r,z)\geq 1+\frac{(x-z)_+h(r',z)}{(1-p)(r'-z)}-h(r',z)=e'(x,r',z) \, ,
\end{align*}
due to $r\to h(r,z)$ and $r\to (r-z)/h(r,z)$ being increasing. \\ 

For the back direction, let $e'$ be monotone. Further, let $z\leq r\leq r'$. By monotonicity of $e'$ we have
\begin{align*}
  1+\frac{(x-z)_+h(r,z)}{(1-p)(r-z)}-h(r,z)\geq1+\frac{(x-z)_+h(r',z)}{(1-p)(r'-z)}-h(r',z) \qquad \forall x\in \R \, .
\end{align*} 
Since this especially holds for $x\leq z$, we directly derive that $r\to h(r,z)$ is increasing. Now, assume $\frac{h(r,z)}{r-z}<\frac{h(r',z)}{r'-z}$. Then, there exists $\epsilon>0$ such that $\frac{h(r,z)}{r-z}<(1-\epsilon)\frac{h(r',z)}{r'-z}$. Setting $x=z+(1-p)\frac{r'-(1-\epsilon) r- \epsilon z}{\epsilon}$ (notice that $r'-(1-\epsilon)r-\epsilon z\geq 0$ since $r'\geq r\geq z$) and multiplying both sides by $\frac{(x-z)_+}{1-p}-(r-z)$ yields 
\begin{align*}
\iff   &\frac{(x-z)_+h(r,z)}{(1-p)(r-z)}-h(r,z)&&<\frac{(x-z)_+h(r',z)}{(1-p)(r'-z)}-h(r',z)\\
\iff   &1+\frac{(x-z)_+h(r,z)}{(1-p)(r-z)}-h(r,z)&&<1+\frac{(x-z)_+h(r',z)}{(1-p)(r'-z)}-h(r',z)\, ,
\end{align*}
contradicting monotonicity of $e'$.
\end{proof}
 \begin{bemerkung}
Theorem \ref{CharacES} shows that if we construct our e-process to backtest the expected shortfall according to Lemma \ref{ProcessisMartingale}, the monotone backtest e-statistic suggested in Section \ref{sec3}, i.e. setting $h=1$ and $k=0$ in the formula for $e'$, is in essence unique since different choices for $h$ can be accounted for by varying the betting process accordingly. 
 \end{bemerkung}
\end{appendices}
\end{document}